\documentclass[11pt,reqno]{amsart}
\pdfoutput=1 
\usepackage{geometry}
\usepackage[margin=10mm]{caption}
\usepackage{comment}
\usepackage[utf8]{inputenc}
\usepackage[english]{babel}
\usepackage{enumitem}
\usepackage{amssymb, amsmath, amsthm, amsfonts, eucal, mathalpha, mathrsfs, yfonts, lmodern, slashed, microtype}
\usepackage[sorting=none, backend=biber, natbib=true, style=numeric-comp]{biblatex}
\usepackage{array, booktabs}
\usepackage{graphicx, subcaption, tikz, tikz-cd}
\usetikzlibrary{angles, quotes, calc, shapes.geometric}
\usepackage{csquotes, hyperref}
\hypersetup{colorlinks = false , citecolor=orange, linkcolor=orange, linkbordercolor = orange, citebordercolor = orange, urlbordercolor = orange}
\usepackage{tikz}
\usepackage{tikz-3dplot}
\usepackage{tikz-cd}
\usepackage{xcolor}
\usepackage{adjustbox}

\definecolor{citecolor}{rgb}{0.0, 0.45, 0.45}   
\definecolor{emphcolor}{rgb}{0.8, 0.0, 0.0}      
\definecolor{metacolor}{rgb}{0.58, 0.0, 0.83}     

\theoremstyle{definition}
\newtheorem{theorem}{Theorem}[section]  
\newtheorem{definition}[theorem]{Definition}
\newtheorem{example}[theorem]{Example}
\newtheorem{prop}[theorem]{Proposition}
\newtheorem{lemma}[theorem]{Lemma}
\newtheorem{corollary}[theorem]{Corollary}

\newtheorem{remark}[theorem]{Remark}

\newtheorem*{theoremA}{Theorem A}
\newtheorem*{theoremB}{Theorem B}
\newtheorem*{theoremC}{Theorem C}

\newtheorem*{corollaryD}{Corollary D}

\newtheorem*{theoremE}{Theorem E}

\newcommand{\beq}{\begin{equation}}
\newcommand{\eeq}{\end{equation}}
\newcommand{\bqa} {\begin{eqnarray}}
\newcommand{\eqa} {\end{eqnarray}}
\numberwithin{equation}{section}
\def\[#1\]{%
  \begin{equation}\begin{gathered}#1\end{gathered}\end{equation}%
}

\newcommand{\p}{\partial}
\newcommand{\ups}{\upsilon}
\newcommand{\eps}{\varepsilon}

\renewcommand\Re{\operatorname{Re}}
\renewcommand\Im{\operatorname{Im}}
\def \l {\left(}
\def \r {\right)}
\def \lal {\langle}
\def \ral {\rangle}

\DeclareMathOperator{\Aut}{Aut}
\DeclareMathOperator{\Ad}{Ad}

\DeclareMathOperator{\End}{End}

\DeclareMathOperator{\Hom}{Hom}
\DeclareMathOperator{\Id}{Id}

\DeclareMathOperator{\sign}{sign}
\DeclareMathOperator{\supp}{supp}

\DeclareMathOperator{\Tr}{Tr}

\makeatletter
\newcommand{\colim@}[2]{%
  \vtop{\m@th\ialign{##\cr
    \hfil$#1\operator@font colim$\hfil\cr
    \noalign{\nointerlineskip\kern1.5\ex@}#2\cr
    \noalign{\nointerlineskip\kern-\ex@}\cr}}%
}
\newcommand{\colim}{%
  \mathop{\mathpalette\colim@{\rightarrowfill@\textstyle}}\nmlimits@
}
\makeatother

\newcommand{\CA}{{\mathcal A}}
\newcommand{\CB}{{\mathcal B}}
\newcommand{\CC}{{\mathcal C}}

\newcommand{\CF}{{\mathcal F}}

\newcommand{\CO}{{\mathcal O}}

\newcommand{\NN}{{\mathbb N}}
\newcommand{\NNb}{{\overline{\mathbb N}}}
\newcommand{\NNo}{{\mathbb N}_0}
\newcommand{\ZZ}{{\mathbb Z}}
\newcommand{\RR}{{\mathbb R}}
\newcommand{\CCC}{{\mathbb C}}

\newcommand{\bog}{{\beta}}
\newcommand{\cstar}[1]{{\mathcal{#1}}}
\newcommand{\der}[1]{{\mathsf #1}}
\newcommand{\fparity}{{\Theta}}
\newcommand{\hD}{{-i \p_x}}
\newcommand{\hilb}[1]{{\mathcal #1}}
\newcommand{\interior}[1]{{\kern0pt#1}^{\mathrm{o}}}
\newcommand{\ql}[1]{{\mathcal{#1}}}

\newcommand{\sMat}{{\mathsf{sMat}}}
\newcommand{\Mat}{{\mathsf{Mat}}}
\newcommand{\opMat}{{(\mathsf{s})\mathsf{Mat}}}
\newcommand{\Amaj}{{\ql{A}_{\text{Maj}}}}
\newcommand{\Amajd}{{\ql{A}^{\text{dis.}}_{\text{Maj}}}}
\newcommand{\AEe}{{\ql{A}_{E_8}}}
\newcommand{\AEed}{{\ql{A}^{\text{dis.}}_{E_8}}}

\newcommand{\QLAcat}{{{\mathbf{QLA}}}}
\newcommand{\sQLAcat}{{{\mathbf{sQLA}}}}
\newcommand{\opQLAcat}{{{\mathbf{(s)QLA}}}}
\newcommand{\Azcat}{{{\mathbf{Az}}}}
\newcommand{\sAzcat}{{{\mathbf{sAz}}}}
\newcommand{\opAzcat}{{{\mathbf{(s)Az}}}}

\newcommand{\opMatCatD}{{{\mathbf{(s)Mat}}}}
\newcommand{\DHRcat}{{\mathbf{DHR}}}
\newcommand{\Hilbcat}{{\mathbf{Hilb}}}

\newcommand{\Brgr}{{\text{Br}}}
\newcommand{\sBrgr}{{\text{sBr}}}
\newcommand{\opBrgr}{{(\text{s})\text{Br}}}
\newcommand{\QCA}{{\text{QCA}}}
\newcommand{\sQCA}{{\text{sQCA}}}
\newcommand{\opQCA}{{(\text{s})\text{QCA}}}
\newcommand{\IP}{{\text{IP}}}
\newcommand{\sIP}{{\text{sIP}}}
\newcommand{\opIP}{{(\text{s})\text{IP}}}

\newcommand{\Kit}{{\mathbf{Kit}}}
\newcommand{\sKit}{{\mathbf{s}\mathbf{Kit}}}

\newcommand{\Autbs}{{\text{Aut}_{\text{bs}}}}
\newcommand{\Hombs}{{\text{Hom}_{\text{bs}}}}
\newcommand{\Isobs}{{\text{Iso}_{\text{bs}}}}

\newcommand{\opComm}{{\mathsf{(s)Comm}}}

\newcommand{\CAR}{{\mathfrak{A}_{\text{CAR}}}}

\newcommand{\QFS}[1]{{\varphi_{#1}}}

\newcommand{\Bumpfns}{{C^{\infty}_c (\RR)}}

\makeatletter 
\def\l@subsection{\@tocline{2}{0pt}{2pc}{6pc}{}} \makeatother

\title{Quantum cellular automata and invertible phases of matter}

\author{Corey Jones}
\address{
Department of Mathematics, North Carolina State University, Raleigh, NC 27695, USA}

\author{Nikita Sopenko}
\address{School of Natural Sciences, Institute for Advanced Study, 1 Einstein Drive, Princeton, NJ 08540, USA}

\author{Ryan Thorngren}
\address{Mani L. Bhaumik Institute for Theoretical Physics, Department of Physics and Astronomy, University of California, Los Angeles, CA 90095, USA}

\date{\today}

\begin{document}

\begin{abstract}
We introduce and study (fermionic and bosonic) invertible quasi-local algebras over uniformly locally finite metric spaces $X$ with infinite-dimensional local von Neumann algebras. We show that the group of Brauer equivalence classes of such algebras is isomorphic to both the group of phases of invertible states and the group of stable equivalence classes of quantum cellular automata over $X\times \mathbb{Z}$. Using $K$-theory of the symmetric monoidal category of invertible quasi-local algebras and bounded spread isomorphisms, we propose a definition of an $\Omega$-spectrum of invertible phases as conjectured by Kitaev. We then show that the $c=\frac{1}{2}$ chiral Majorana fermion net and the $(E_{8})_{1}$ conformal net provide Brauer non-trivial invertible quasi-local algebras, thus providing explicit constructions of non-trivial invertible states and quantum cellular automata on $\ZZ^{2}$. In addition, we show that the time-slice nets of rational diagonal conformal field theories admit lattice degrees of freedom, which implies the discretization of any holomorphic conformal net is invertible.
\end{abstract}

\maketitle

\tableofcontents

\section{Introduction}

Phases of matter at zero-temperature, also known as topological phases of matter, lie at the intersection of condensed matter, quantum information, and quantum field theory. From a mathematical point of view, they can be modeled by a suitable class of states on algebras of observables of quantum lattice systems in infinite volume. In this approach, states correspond to systems in the same phase if they are related by a locally generated automorphism, possibly after adding ancillary unentangled degrees of freedom. Establishing a proper framework for this story and rigorously deriving classifications conjectured by physicists has become a major research program \cite{Kitaev_2006,LevinWen2005StringNet,HastingsWen2005QuasiAdiabatic, BravyiHastingsMichalakis2010Stability,BravyiHastings2011ShortProof,BachmannMichalakisNachtergaeleSims2012Automorphic,MichalakisZwolak2013Stability,MoonOgata2020BulkEquivalence,Naaijkens2011Localized,Naaijkens2012HaagDuality,Naaijkens2013KosakiLongo,FiedlerNaaijkens2015HaagDuality,ChaNaaijkensNachtergaele2018GroundStates,ChaNaaijkensNachtergaele2020StabilityCharges,KatoNaaijkens2020EntropicInvariant,SchuchCiracPerezGarcia2010PEPS,Haah2016Invariant,SahinogluEtAl2021MPO,BolsVadnerkar2025Classification,BolsHamdanNaaijkensVadnerkar2026Category,BolsKjaer2025LevinWenI,BolsKjaer2026LevinWenII,JonesNaaijkensPenneysWallick2025LTO,NaaijkensPenneysWallick2026LTOHaagDuality,OgataPerezGarciaRuizDeAlarcon2025HaagDuality,LiuZhao2025ReflectionPositiveTopologicalOrder,bachmann2025stackingtrivialityinvertiblephases}


The simplest class of topological phases to characterize mathematically is the class of \textit{invertible phases} \cite{Kitaev2013ClassificationSRE}.
These are phases which have an inverse under a natural monoidal operation induced by a ``stacking" (or composition) of quantum systems. Although states in an invertible phase do not admit non-trivial topological excitations, their non-triviality is expected to manifest itself in the form of anomalous edge modes in the presence of a boundary. While it is conjectured that these phases are related to bordism groups of smooth manifolds \cite{kapustin2014symmetry, kapustin2015fermionic} and invertible topological quantum field theories (TQFTs) \cite{FreedHopkins2021ReflectionPositivity}, the precise relationship remains unclear.

Invertible phases of states over an $n$-dimensional lattice $\ZZ^n$ naturally form an abelian group under stacking, which we denote $\IP(\mathbb{Z}^n)$. It was conjectured by A. Kitaev that these groups can be organized into a canonical $\Omega$-spectrum whose $n$-th space can be thought of as the classifying space of $n$-dimensional systems in an invertible phase. As a consequence, the homotopy groups $\pi_k$ of the $n$-th space coincide with $\IP(\mathbb{Z}^{n-k})$, as codimension $k$ invertible defects of an $n$-dimensional system correspond to $(n-k)$-dimensional invertible phases. An analytic model for this spectrum has been recently proposed in \cite{kubota2025stable}. However, there is a notable lack of non-trivial examples and invariants distinguishing different phases.

Another closely related problem which has recently attracted significant attention \cite{Watrous1995OneDimensionalQCA,SchumacherWerner2004ReversibleQCA,PerezDelgadoCheung2007LocalUnitaryQCA,ArrighiNesmeWerner2011Localizability,Arrighi2019OverviewQCA,Farrelly2020,freedman2020classification,FreedmanHaahHastings2022GroupStructure,haah2023nontrivial,Haah2023InvertibleSubalgebras} is the problem of classification of quantum cellular automata (QCA), which are strict locality preserving automorphisms of the algebra of observable of a lattice system in infinite-volume. Originally introduced as a model for discrete-time dynamics of a quantum system, they can be thought of as generalizations of finite-depth quantum circuits, and can be used to produce non-trivial entangled states from a trivial product state. Stable equivalence classes of QCA modulo circuits also form an abelian group, $\QCA(\mathbb{Z}^{n})$ \cite{FreedmanHaahHastings2022GroupStructure}.
Applying a QCA to a product state always yields an invertible state, and applying an equivalent QCA produces an invertible state in the same phase, thus furnishing a homomorphism $\QCA(\mathbb{Z}^{n})\rightarrow \IP(\mathbb{Z}^{n})$.
Furthermore, several recent results define an $\Omega$-spectrum for QCA algebraically \cite{LongElse2025,tu2026anomalies,czajka2025anomalies,ji2026quantum,ludewig2026quantum}.
However, the homomorphism $\QCA(\mathbb{Z}^{n})\rightarrow \IP(\mathbb{Z}^{n})$ is not injective, and it is expected to not be surjective in general. This strongly motivates studying the relationship between invertible states and QCA, and determining whether there could be an algebraic construction of the Kitaev spectrum.

Lattice models of invertible phases and equivalence classes of quantum cellular automata discussed in the literature almost universally use finite-dimensional local Hilbert spaces. From a physical perspective, local finite-dimensionality is merely a useful approximation that makes certain problems more tractable. Most realistic theories of quantum many-body systems, including quantum field theories, have infinitely many degrees of freedom per unit volume. This motivates us to consider a more general class of lattice systems modeling topological phases of matter.

In this paper, we introduce and study lattice models of invertible phases and quantum cellular automata which use infinite-dimensional local Hilbert spaces. We show that, in contrast to the setting with finite-dimensional Hilbert spaces, the natural map $\QCA(\mathbb{Z}^{n})\rightarrow \IP(\mathbb{Z}^{n})$ induced by applying a QCA to an unentangled state, \textit{is an isomorphism of abelian groups}.

Our proof of this result factors through an isomorphism with a third abelian group, the Brauer group $\Brgr(\mathbb{Z}^{n-1})$ of quasi-local algebras. Elements of this group are stable equivalence classes of invertible quasi-local algebras, i.e. quasi-local algebras $\ql{A}$ such there exists another quasi-local algebra $\tilde{\ql{A}}$ and a bounded spread isomorphism $\ql{A}\otimes \tilde{\ql{A}} \cong \Mat(\mathbb{Z}^{n-1})$, where $\Mat(\mathbb{Z}^{n-1})$ is the algebra of the infinite tensor product of bounded operators on a separable Hilbert space $B(\hilb{H})$. The connection between QCA and invertible algebras was established in the locally finite-dimensional case by Haah \cite{Haah2023InvertibleSubalgebras}, and our proof of the isomorphism $\Brgr(\mathbb{Z}^{n-1})\cong \QCA(\mathbb{Z}^{n})$ is a straightforward generalization of Haah's arguments. Our main contribution in this direction is establishing an isomorphism $\IP(\mathbb{Z}^{n})\cong \Brgr(\mathbb{Z}^{n-1})$ using a new concept of the boundary quasi-local algebra associated to an invertible state.
This can be seen as a mathematical incarnation of the bulk-boundary correspondence, relating the anomalous boundary modes, described by the boundary algebra, to the phase of the bulk invertible state.
Interestingly, even if the invertible state is defined on a locally finite-dimensional quasi-local algebra, the boundary algebra is always locally infinite dimensional, which demonstrates the necessity of working with locally infinite lattice Hilbert spaces.

We demonstrate our results in the general setting where degrees of freedom live at the points of a uniformly locally finite metric space $X$. An emerging perspective is that many properties concerning quantum phases of matter should only depend on the coarse equivalence class of the underlying metric space, and uniformly locally finite metric spaces provide a flexible enough playground as they are coarsely equivalent to many continuous spaces of interest (e.g., bounded geometry manifolds). Our main results compare the group $\opQCA(X\times \mathbb{Z})$ of (fermionic) QCA over $X\times \mathbb{Z}$ (Definition \ref{def:QCA}), the (super) Brauer group $\opBrgr(X)$ of invertible (super) quasi-local algebras over $X$ (Definition \ref{def:Brgrp} and Definition \ref{def:sBrgrp}), the group $\opIP(X\times \mathbb{Z})$ of invertible (fermionic) phases over $X\times \mathbb{Z}$ (Definition \ref{def:InvPh}). Our notation is intended to indicate that we are treating the super/fermionic cases in parallel with the bosonic case.

\begin{theoremA}\label{thm:A}
Let $X$ be a uniformly locally finite metric space. Then there are isomorphisms of groups $\opQCA(X\times \mathbb{Z})\cong \opBrgr(X)\cong \opIP(X\times \mathbb{Z}) $.
\end{theoremA}

The proof of this result is a combination of two theorems, establishing the isomorphism  $\opQCA(X\times \mathbb{Z})\cong \opBrgr(X)$ (Theorem \ref{theorembrauerblendiso}) and the isomorphism $\opBrgr(X)\cong \opIP(X\times \mathbb{Z})$ (Theorem \ref{thm:BAlgState}) separately. The isomorphisms are constructed through compositions of basic maps, represented schematically in Figure \ref{fig:triangle}.

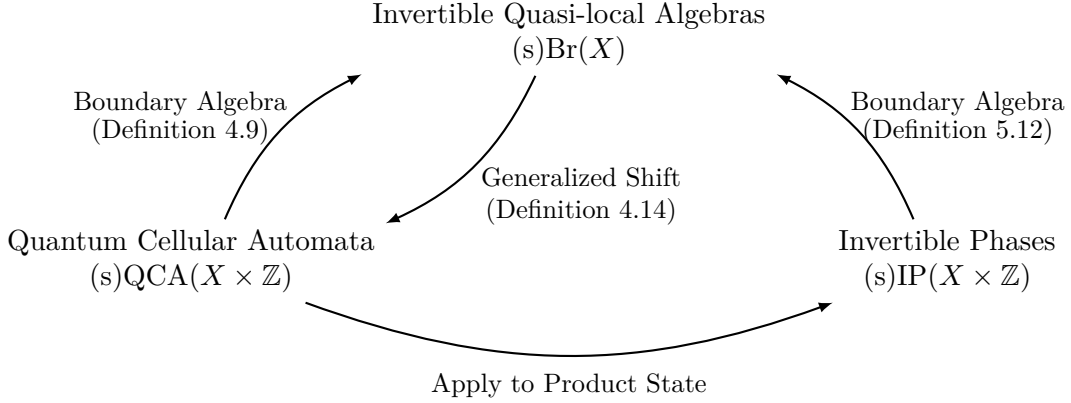
\begin{figure}[htbp]
\begin{tikzpicture}[>=latex, thick]
  \node[align=center] (QCA) at (-5,0) {Quantum Cellular Automata \\ $\opQCA(X\times \mathbb{Z})$};
  \node[align=center] (Br)  at (0,3)    {Invertible Quasi-local Algebras \\ $\opBrgr(X)$};
  \node[align=center] (IP)  at (5,0)  {Invertible Phases \\ $\opIP(X\times \mathbb{Z})$};

  \draw[->]
    (QCA)
    to[bend left=20]
    node[midway, above, yshift=0.1cm,xshift=-1.3cm]
      {\small Boundary Algebra}
    node[midway, above, yshift=-0.3cm,xshift=-1.3cm]
      {\small (Definition \ref{def:bdryalgQCA})}
    (Br);

    \draw[<-]
    (QCA)
    to[bend right=20]
    node[midway, below, yshift=0.1cm,xshift=1.4cm]
      {\small Generalized Shift}
    node[midway, below, yshift=-0.3cm,xshift=1.4cm]
    {\small (Definition \ref{def:ingeneralizedshift})}
    (Br);

  \draw[<-]
    (Br)
    to[bend left=20]
    node[midway, above, yshift=0.1cm, xshift=1.3cm]
      {\small Boundary Algebra}
    node[midway, above, yshift=-0.3cm, xshift=1.3cm]
      {\small (Definition \ref{def:BALGfromSTATE})}
    (IP);

    \draw[->]
    (QCA)
    to[bend right=20]
    node[midway, below, yshift=-0.1cm,xshift=0]
      {\small Apply to Product State}
    (IP);
\end{tikzpicture}
\caption{Schematic of maps in proof of Theorem A}
\label{fig:triangle}
\end{figure}

Our next result leverages the recent construction of $\Omega$-spectra for QCA \cite{ji2026quantum, ludewig2026quantum, czajka2025anomalies}, to build an $\Omega$-spectrum in the locally infinite-dimensional setting. Following these proposals, we consider the symmetric monoidal category $\Azcat(X)$ of invertible quasi-local algebra, and take its $K$-theory to obtain $\Omega$-spectrum. The novelty of our $\Omega$-spectrum compared to these previous results is that our isomorphism $\QCA(\mathbb{Z}^{n})\cong \IP(\mathbb{Z}^{n})$ implies that we can interpret this as a spectrum \textit{of invertible phases} (after a minor modification of the highest homotopy group) conjectured by Kitaev \cite{Kitaev2013ClassificationSRE}. This provides an algebraic alternative to a recent analytic proposal \cite{kubota2025stable}. The results of Section \ref{sec:KitaevSpectrum} are summarized in the following theorem.

\begin{theoremB}\label{thm:B}
Let $\Azcat(\ZZ^{d-1})$ (resp. $\sAzcat(\ZZ^{d-1})$) be a symmetric monoidal category of bosonic (resp. fermionic) invertible quasi-local algebras over $\ZZ^{d-1}$ with morphisms being bounded spread isomorphisms. Then its $K$-theory groups are $K_n(\opAzcat(\ZZ^{d-1})) = \opIP(\ZZ^{d-n})$ for $0 \leq n \leq d$, $K_{d+1}(\opAzcat(\ZZ^{d-1})) = \RR/\ZZ$ and $K_n(\opAzcat(\ZZ^{d-1})) = 0$ for $n>d+1$.
\end{theoremB}

One might be concerned that by allowing infinite-dimensional local Hilbert spaces, we may trivialize the classification of QCAs and invertible phases. However, the one dimensional Majorana translation which creates the Kitaev wire from a trivial state, remains a non-trivial QCA. Moreover, using the general theory developed above, we can construct non-trivial invertible states in $n$ spatial dimensions by finding Brauer non-trivial quasi-local algebras in $n-1$ dimensions.
To this end, we turn our attention to the quasi-local algebras associated to chiral conformal nets \cite{FredenhagenJoerss1996ConformalHaagKastlerNets,GuidoLongoWiesbrock1998Extensions,KawahigashiLongoMuger2001MultiIntervalSubfactors}. Conformal nets are an extremely rich analytical framework for describing chiral conformal field theories. At first glance, the quasi-local algebras associated to chiral conformal nets appear fundamentally different from the sort of algebras we typically consider on the lattice, since the local von Neumann algebras are assigned to continuous intervals, and are generically type $\rm{III}_{1}$ factors. However, there is a general discretization procedure (see Remark \ref{rem:discretizations}) that allows to think of them as subalgebras of a tensor product algebra $\Mat(\ZZ)$ of infinite type $\rm{I}$ factors, rendering the quasi-local algebras of conformal nets into our framework. 

There is a natural expectation that holomorphic conformal nets (namely those with trivial superselection category) should be invertible, since they describe conformal field theories (CFTs) expected to live at the chiral boundaries of invertible phases. The simplest non-trivial fermionic example of a chiral conformal net is the chiral $c=\frac{1}{2}$ Majorana fermion theory, built from the quasi-free ground state of left translations on the CAR algebra associated to the Hilbert space $L^{2}(\mathbb{R})$  (see Section \ref{subsec:ChiralMajorana}). The simplest non-trivial bosonic is $\ql{A}_{E_{8}}$, associated to the chiral $(E_{8})_{1}$ theory. We prove the following.

\begin{theoremC}
Let $\mathcal{A}^{\text{dis.}}_{\text{Maj}}$ and $\ql{A}^{\text{dis.}}_{E_{8}}$ be the discretizations of the quasi-local algebras associated to the chiral $c=\frac{1}{2}$ Majorana fermion and the $(E_{8})_{1}$ conformal net, respectively. Then $\mathcal{A}^{\text{dis.}}_{\text{Maj}}$ and $\ql{A}^{\text{dis.}}_{E_{8}}$ are invertible, and define non-trivial classes in $\sBrgr(\mathbb{Z})$ and $\Brgr(\mathbb{Z})$, respectively.
\end{theoremC}

The invertibility of these quasi-local algebras is demonstrated in Section \ref{sec:invertiblefreefermions}, while the non-triviality is proved in \ref{sec:cmInvariant}. Combined with our structural results above, Theorem C implies there are non-trivial (fermionic and bosonic) QCA and invertible (fermionic and bosonic) states when we allow our local Hilbert spaces to be infinite-dimensional, which sharply contrasts with results in the literature when Hilbert spaces are restricted to be finite-dimensional (\cite{freedman2020classification, HaahKwanLong}).

\begin{corollaryD}
The groups $\opBrgr(\mathbb{Z})\cong \opIP(\mathbb{Z}^{2})\cong \opQCA(\mathbb{Z}^{2})\ne 0$.
\end{corollaryD}

We can actually show a more general result (in the bosonic case) which may be of independent interest. For any rational conformal net with quasi-local algebra $\ql{A}$, there is an associated quasi-local algebra $\ql{B}$ called the Longo-Rehren extension $\ql{A}\otimes \ql{A}^{op}\subseteq \ql{B}$ \cite{LongoRehren1995NetsOfSubfactors}. The Longo-Rehren extension can be physically interpreted as the time slice observables of the full diagonal $(1+1)$-dimensional CFT associated with the chiral theory $\ql{A}$. We prove the following theorem, which shows that lattice degrees of freedom can always be constructed on the time slice net locally. The following theorem is a summary of Theorem \ref{thm:latticedegreesoffreedom}.

\begin{theoremE}
Let $\ql{A}$ be the quasi-local algebra over $\mathbb{R}$ obtained via the restriction of a rational (bosonic) conformal net over $S^{1}$. If $\ql{A}\otimes \ql{A}^{op}\subseteq \ql{B}$ is the Longo-Rehren extension, then the discretization $\ql{B}^{\text{dis.}}$ is bounded spread isomorphic to $\text{Mat}(\mathbb{Z})$. In particular, the discretization of any holomorphic conformal net is invertible.
\end{theoremE}

From the observed relation between orbifold anomalies and chiral central charge \cite{johnson2019moonshine,Bischoff2020AnomaliesCyclicPermutationOrbifolds, sopenko2025reflection}, we can conclude that for holomorphic conformal nets with central charge not divisible by 3, $\ql{A}^{\text{dis.}}$ defines a non-trivial element of the Brauer group. However, we will show in future work that non-zero central charge implies non-triviality. We conjecture that the group $\IP(\mathbb{Z}^{2})\cong\Brgr(\mathbb{Z})$ is isomorphic to $\mathbb{Z}$, and that it is generated by the $E_{8}$ state constructed above. This agrees with proposed classifications for invertible phases in the literature \cite{Kitaev2013ClassificationSRE}.

Our results have connections with (and shed some light on) related problems, which we discuss in the following remarks.

\begin{remark}
For spin systems, the classification of QCAs has been obtained in dimension one \cite{Gross_2012} and two \cite{freedman2020classification}. Various QCAs have been constructed in dimension three \cite{haah2023nontrivial} and higher \cite{chen2021exactly, fidkowski2025quantum} which were conjectured to be non-trivial. Both established and conjectured classifications are very different from the classification of invertible states in the same dimension, and as we mentioned above the relationship between them is obscure. Some non-trivial QCAs (such as a shift in 1d), when acting on an invertible state, do not change its phase, while some invertible states (such as a ground state of a 2d $p+ip$ superconductor) can not be produced by any QCA.

In contrast, after stabilization with infinite-dimensional on-site Hilbert spaces both problems turned out to be exactly the same. As we show, any invertible state over $\ZZ^d$ can be produced by a QCA from an unentangled state and every QCA that preserves an unentangled state is in a trivial equivalence class. That trivializes some QCAs which were non-trivial in the setting of finite-dimensional Hilbert spaces, such as shift QCA, leading to a simpler classification and providing new non-trivial QCAs in dimension two. We leave it as an open question whether three-dimensional QCAs, which are conjectured to be non-trivial in the finite-dimensional case, become trivial after stabilization\footnote{In fact, the expectation that there are no non-trivial invertible phases in three dimensions (without imposing symmetry) suggests that all such QCAs should trivialize.
}
\end{remark}

\begin{remark}
In this paper, we consider invertible states with zero correlations between observables separated by the correlation length. In other words, we assume that an invertible state together with its inverse can be prepared from an unentangled product state by a quantum circuit. While ground states of quantum lattice systems in general have gradually decaying correlations even if the Hamiltonian is gapped and has finite-range interactions, we expect that any invertible phase admits a representative with finite-range correlations if we use infinite-dimensional ancillas. We support this claim by showing that chiral two-dimensional phases have such representatives, which is believed to be not the case in the setting with finite-dimensional local Hilbert spaces. 

We also note that any invertible state in our setting admits a gapped Hamiltonian. Indeed, as we show, any such state can be produced by a QCA from an entangled state, and the desired Hamiltonian can be produced from an on-site Hamiltonian for an unentangled state using this QCA.
\end{remark}

\begin{remark}
One of the expected pieces of data characterizing two-dimensional topological order is the chiral central charge $c_-$ \cite{Kitaev_2006}. It is believed that two-dimensional topological phases with $c_- \neq 0$ can not be realized by spin systems with Hamiltonians $H = \sum_{x \in \ZZ^d} h_x$ being a sum of commuting terms $[h_x, h_y] = 0$ with $h_x$ being localized near $x \in \ZZ^d$. A primary reason for this conjecture is the physical expectation that such phases have a universal non-vanishing energy current at finite temperature near the edge (in the presence of a boundary), and a way to express this current through the bulk quantity \cite{Kitaev_2006} that vanishes identically for Hamiltonians as above. In particular, this conjecture implies that the ground states of systems in an invertible phase with non-trivial $c_-$ can not be prepared by quantum cellular automata.

The equivalence of the classification of quantum cellular automata and invertible states for lattice systems with infinite-dimensional on-site Hilbert spaces implies that, in contrast to spin systems, there is a way to realize phases with non-zero $c_-$ via Hamiltonians which are are sums of commuting local terms. Indeed, to get such a Hamiltonian we can just apply the corresponding QCA to a Hamiltonian of an unentangled state of this form (even if each local term is an unbounded operator).\footnote{In contrast to lattice systems with finite-dimensional on-site Hilbert spaces, Hamiltonians of the form $H = \sum_{x \in \ZZ^d} h_x$ with bounded $h_x$ (e.g., projectors) are not physically relevant as they have infinite free energy density at finite temperature.} We note, however, that if one wants to construct a Hamiltonian for a system with a boundary, one would have to introduce Hamiltonian terms near the edge that would allow for a non-trivial energy current.

\end{remark}

\begin{remark}
The results of our paper also have implications for quantum field theory (QFT). One of the expected pieces of data of any genuine relativistic QFT in $(d+1)$-dimensions is the net of von Neumann algebras associated with open bounded subspaces of the Minkowski space $\RR^{1,d}$ or the corresponding quasi-local algebra over the spatial slice $\RR^d$. A natural question is how to classify such algebras up to some notion of equivalence.

While in general it is not true that the quasi-local algebra $\ql{A}$ of a QFT can be represented as a tensor product of the algebra $\ql{A}_U$ of a given region $U$ and the algebra of its complement $U^c$, it is true that one can always find a slightly bigger region\footnote{in fact, arbitrarily close to $U$} $V$ such that there is a tensor product factorization of the global algebra into factors which contain $\ql{A}_U$ and $\ql{A}_{V^c}$. This fundamental result is known as {\it split property} and has been shown to be a consequence of normal thermodynamic properties which guarantee the existence of a KMS state at any finite temperature \cite{buchholz1986causal, buchholz1987universal}. 

One can wonder if it is possible to perform this tensor product factorization at different locations simultaneously. More precisely, given a quasi-local algebra $\ql{A}$ over $\RR^d$, one can ask if there is an isomorphism between the algebra $\ql{A}$ and an infinite tensor product algebra $\bigotimes_{j \in \ZZ^d} B(\hilb{H}_j)$ of algebras of bounded operators over the lattice $\ZZ^d \subset \RR^d$, such that the image of $\ql{A}_U$ is contained in $\bigotimes_{j \in U^{+r}} B(\hilb{H}_j)$ for an $r$-neighborhood $U^{+r}$ of $U$. In other words, if it is possible to identify the quasi-local algebra of a QFT with a tensor product algebra while preserving locality. We call it {\it extended split property}.

In this paper, we show that it is indeed the case for quasi-local algebras over $\RR$ associated with relativistic quantum field theories describing non-chiral free fermions in $\RR^{1,1}$ (Section \ref{sec:invertiblefreefermions}) and rational diagonal (1+1)-dimensional conformal field theories (Section \ref{sec:bosonicholomorphic}) with vanishing gravitational anomaly. While these concrete results are limited to (1+1)-dimensions, we expect that they can be extended to higher dimensions. Moreover, we conjecture that the quasi-local algebra of {\it any} genuine relativistic QFT is an {\it invertible quasi-local algebra}, and that it has an extended split property if and only if the gravitational anomaly vanishes. It provides a criterion for whether a given quantum field theory can be described as a lattice system.

\end{remark}

\subsection{Notations and conventions.} 

Here we record some notations and conventions used throughout the paper:

\begin{itemize}
 
 \item By Hilbert spaces we always mean complex Hilbert spaces unless specified otherwise. We denote the inner product on a Hilbert space $\hilb{H}$ by $\lal\,\cdot\,, \,\cdot\, \ral: \hilb{H} \times \hilb{H} \to \CCC$ and assume it is anti-linear in the first argument, and linear in the second. We denote the von Neumann algebra of bounded operators by $B(\hilb{H})$ and the group of unitary elements by $U(\hilb{H}) \subset B(\hilb{H})$.
    
\item We set $\NNo = \{0,1,2,...\}$, $\NN = \{1,2,3,...\}$, $\NNb = \{1,2,3,...\} \cup \{\infty\}$.

\item For a metric space $X$, we denote the distance between $x,y \in X$ by $|x-y| \in \RR_{\geq 0}$. For a subset $U \subset X$, we let $U^{+r}$ be the subset of all points $x \in X$ such that $|x-u| \leq r$ for some $u \in U$. We let $\interior{U}$ be the {\it interior} of $U$ that consists of all the points $y \in U$ such that $U$ contains an open ball with the center at $y$ (or, equivalently, the largest open subset contained within $U$).

\item A metric space is called {\it proper} if all bounded subsets are compact. It is called uniformly locally finite if for any $r > 0$ there exists $N_r > 0$ the number of points in any ball of radius $r$ is less than $N_r$. We say that a metric space $X$ has {\it bounded geometry} if it admits a uniformly locally finite subset $Y \subset X$ which is coarsely dense, i.e. there exists $\delta > 0$ such that for any $x \in X$ there exists $y \in Y$ satisfying $|x-y| < \delta$.

\item When we regard $\RR^d$ and $\ZZ^d$ as metric spaces, we assume they are equipped with Euclidean metric. More generally, for any metric spaces $X$ and $Y$, we consider the Euclidean product metric $|(x,y)-(x',y')|=\sqrt{|x-x'|^{2}+|y-y'|^{2}}$. We typically only care about metrics up to coarse equivalence, so this convention is inessential.

\item A {\it $*$-algebra} is an algebra $\cstar{A}$ over $\CCC$ equipped with a map $a \to a^*$, $a \in \cstar{A}$ which is involutive $(a^*)^* = a$, anti-linear $(\lambda a + \mu b)^* = \bar{\lambda} a^* + \bar{\mu} b^*$, $\lambda, \mu \in \CCC$, and is an anti-homomorphism $(ab)^* = b^* a^*$.

\item A {\it state} on a unital $*$-algebra $\ql{A}$ is a linear functional $\omega:\cstar{A} \to \CCC$ that is positive $\omega(a^* a) \geq 0$, $a \in \cstar{A}$ and normalized $\omega(1) = 1$. For a $*$-algebra $\ql{A}$ and a state $\omega$ on $\ql{A}$, we denote the associated Gelfand–Naimark–Segal (GNS) representation by $(\pi_{\omega}, \hilb{H}_{\omega}, \Omega_{\omega})$, where $\pi_{\omega}: \ql{A} \to B(\hilb{H}_{\omega})$, $\Omega_{\omega} \in \hilb{H}_{\omega}$, $\omega(a) = \lal \Omega_{\omega}, \pi_{\omega}(a) \Omega_{\omega}\ral$.

\item A $*$-algebra $\cstar{A}$ is a {\it $C^*$-algebra} if it admits a complete norm with $\|a^* a\| = \|a\|^2$ for any $a \in \cstar{A}$. Such a norm is unique if exists and can be defined algebraically by 
$$\|a\| = \sup |\{\lambda \in \CCC: (a^* a - \lambda \cdot 1)\text{ is not invertible}\}|, \,\,\, a \in \cstar{A}.$$ 
A $*$-algebra is a {\it von Neumann algebra} if it is a $C^*$-algebra which has a predual, i.e. there exists a Banach space\footnote{which can be interpreted as the space of physically realizable states on the algebra of observables} whose dual Banach space is isometrically isomorphic to the Banach space of the $C^*$-algebra. It comes equipped with a natural ultraweak topology induced by the predual. We say that maps between von Neumann algebras are {\it normal} if they are continuous is this topology. We mostly consider separable von Neumann algebras, i.e. von Neumann algebras with separable predual. Standard references for $C^*$-algebras and von Neumann algebras are \cite{Takesaki2002I,Takesaki2003II,Takesaki2003III}, and, for a more physics oriented approach, \cite{bratteli2012operator}.

\end{itemize}

\subsection{Acknowledgments} C.J. would like to thank Jeongwan Haah and David Long for useful comments. N.S. would like to thank Meng Cheng, Anton Kapustin and Sahand Seifnashri for discussion. C.J. is supported by NSF Grant DMS-2247202. N.S. is supported by NSF Grant PHY-2514611 and the Ambrose Monell Foundation. R.T. is supported by the Mani L. Bhaumik Presidential Term Chair. This work arose from discussions at the Stanford Q-Farm workshop ``Quantum Cellular Automata" in March, 2026. Chat GPT Sol 5.6, Claude Fable 5 and Opus 5 were used for editorial purposes in the preparation of this manuscript.

\section{Superalgebras}

By {\it vector superspaces}, we mean $\ZZ/2$-graded vector spaces $V = V^{(0)} \oplus V^{(1)}$. By $V \otimes W$ for two vector superspaces $V$, $W$, we always mean the $\ZZ/2$-graded tensor product. Similarly, by {\it Hilbert superspaces}, we mean $\ZZ/2$-graded Hilbert spaces $\hilb{H} = \hilb{H}^{(0)} \oplus \hilb{H}^{(1)}$, and use $\otimes$ to denote the $\ZZ/2$-graded tensor product. The elements of vector superspaces which have a definite grading are called {\it homogeneous}. The {\it superdimension} of a vector space $V$ is the pair $\text{sdim}(V) := n|m$, where $n = \dim V^{(0)}$, $m = \dim V^{(1)}$.

By {\it superalgebras} we mean $\ZZ/2$-graded associative algebras $\cstar{A} = \cstar{A}^{(0)} \oplus \cstar{A}^{(1)}$ with the grading given by an involutive automorphism which we call {\it fermionic parity} and denote $\fparity$. The elements $a \in \cstar{A}^{(0)}$ satisfy $\fparity(a) = a$ and are called {\it even}, while the elements $a \in \cstar{A}^{(1)}$ satisfy $\fparity(a) = -a$ and are called {\it odd}. If an element is even or odd then it is called {\it homogeneous}. For a homogeneous element $a \in \cstar{A}$, we let $|a| \in \{0,1\}$ be the grading: $\fparity(a) = (-1)^{|a|} a$. For any element $a \in \cstar{A}$, we let $(a+\fparity(a))/2$ and $(a - \fparity(a))/2$ be its {\it even} and {\it odd part}. We say that two elements $a,b \in \cstar{A}$ {\it supercommute} if $a b = b_0 a + b_1 \fparity(a)$, where $b_0$ and $b_1$ are even and odd parts of $b$, respectively. 

A {\it supercommutant} of a subset $\cstar{S} \subset \cstar{A}$ of a superalgebra $\cstar{A}$ is the set of all elements of $\cstar{A}$ which supercommute with elements from this subset. We denote the supercommutant of $\cstar{S} \subset \cstar{A}$ by $\cstar{S}' \subset \cstar{A}$. A {\it supercenter} $Z(\cstar{A})$ is defined to the supercommutant of the whole algebra $\cstar{A}$ in itself. The usual center of $\cstar{A}$ is denoted $Z_0(\cstar{A})$. We use the following notation for the supercommutator: $[a,b]_{\pm} := [a,b_0] + a b_1 - b_1 \Theta(a)$. We call an automorphism of $\cstar{A}$ {\it graded} if it commutes with $\fparity$. A linear functional on a superalgebra is called {\it graded} if it is $\fparity$-invariant.

In this text, we will be considering {\it $*$-superalgebras}, i.e. superalgebras $\ql{A}$ equipped with the structure of a $*$-algebra satisfying $\fparity(a^*) = \fparity(a)^*$. In particular, we have notions of {\it $C^*$-superalgebra} and {\it von Neumann superalgebra} which are $*$-superalgebras satisfying the property of being a $C^*$-algebra and a von Neumann algebra, respectively.

\begin{example}{(\textbf{Bounded operators on Hilbert superspaces})} Let $\hilb{H}$ be a Hilbert superspace. The von Neumann algebra $B(\hilb{H})$ is a von Neumman superalgebra. The grading is conveniently defined by $\fparity:=\text{Ad}_{\Gamma}$, where $\Gamma$ is the parity unitary $\Gamma:\hilb{H}\rightarrow \hilb{H}$ given by $\Gamma(\xi_{0}+\xi_{1})=\xi_{0}-\xi_{1}$, $\xi_{0}\in \hilb{H}^{(0)}$, $\xi_{1}\in \hilb{H}^{(1)}$.
\end{example}

The representations of $C^*$-superalgebras are graded $*$-homomorphisms to $B(\hilb{H})$ for a Hilbert superspace $\hilb{H}$. Every von Neumann superalgebra is isomorphic to a strongly closed supersubalgebra of $B(\hilb{H})$ for some Hilbert superspace $\hilb{H}$. Furthermore, von Neumann superalgebras satisfy the super version of the bicommutant theorem, i.e. the von Neumann algebra $\cstar{A} \subset B(\hilb{H})$ coincides with its double supercommutant $\cstar{A}''$ (see, for example, \cite[Appendix A.2]{Bols2021}). Hence, given a representation $\pi:\cstar{A} \to B(\hilb{H})$ of a $C^*$-superalgebra, there is a natural way to produce a von Neumann superalgebra by taking the double supercommutant $\pi(\cstar{A})''$. In particular, given a graded state $\omega$ on $\cstar{A}$ we have an associated von Neumann superalgebra $\pi_{\omega}(\cstar{A})''$.

By a tensor product of von Neumann superalgebras $\cstar{A}, \cstar{B}$ we always mean the super version of the spatial tensor product $\cstar{A} \overline{\otimes} \cstar{B}$ obtained by representing $\cstar{A}$, $\cstar{B}$ on Hilbert superspaces $\hilb{H}$, $\hilb{K}$ and completing the $\ZZ/2$-graded algebraic tensor product $\cstar{A} \hat{\otimes}_{\mathbb{C}} \cstar{B}$ represented on $\hilb{H} \otimes \hilb{K}$ via the double supercommutant. We use $\otimes$ to denote this $\ZZ/2$-graded spatial tensor product $\cstar{A} \otimes \cstar{B}$ instead of $\overline{\otimes}$, \textit{which differs from the standard conventions in the literature.} For two normal maps between von Neumann superalgebras $\phi:\cstar{A}\rightarrow \tilde{\cstar{A}}$ and $\psi: \cstar{B}\rightarrow \tilde{\cstar{B}}$ (not necessarily preserving $\ZZ/2$-grading), their tensor product $\phi\otimes \psi: \cstar{A}\otimes \cstar{B}\rightarrow \tilde{\cstar{A}} \otimes \tilde{\cstar{B}}$ is defined via
$$(\phi\otimes \psi)(a\otimes b):=\phi(a_{0})\otimes \psi(b)+\phi(a_1)\otimes \fparity(\psi(\fparity(b))),$$
which is also a normal map.

Note that for a von Neumann superalgebra $\cstar{A}$, the odd part of $Z(\cstar{A})$ is trivial, because if there was an odd element $x \in Z(\cstar{A})$, then $x x^* + x^* x = 0$, and since both summands are positive, $x = 0$. Hence, we have $Z(\cstar{A}) = \cstar{A}^{(0)} \cap Z_0(\cstar{A})$. We say that a von Neumann superalgebra $\cstar{A}$ is a {\it superfactor} if its supercenter $Z(\cstar{A})$ is trivial.

\begin{prop}    \label{prop:SuperfactorCharacterization}
Let $\cstar{A}$ be a superfactor. Then either $\cstar{A}$ is a factor, or $\cstar{A} \cong \cstar{M} \otimes \text{Cl}_1$ for a factor $\cstar{M}$ with a trivial grading, where $\text{Cl}_1$ is the Clifford superalgebra generated by $1$ and an odd element $u$, $u^2 = 1$.
\end{prop}
\begin{proof}   
Suppose there exists a non-zero odd element $x \in \cstar{A}^{(1)} \cap Z_0(\cstar{A})$. Since $(e^{i \theta} x+ e^{-i \theta}x^*)$ is non-zero for some $\theta$, there exists a non-zero self-adjoint odd element $y \in \cstar{A}^{(1)} \cap Z_0(\cstar{A})$. Since the even part of $Z_0(\cstar{A})$ is trivial, $y^2 = \lambda$ for some positive real number $\lambda$. It follows that there exists an odd self-adjoint unitary $u \in \cstar{A}^{(1)} \cap Z_0(\cstar{A})$, $u^2 = 1$. In this case, $\cstar{A} = \cstar{A}^{(0)} \oplus u \cstar{A}^{(0)}$ and therefore $\cstar{A} \cong \cstar{M} \otimes \text{Cl}_1$ for a factor $\cstar{M} \cong \cstar{A}^{(0)}$.
\end{proof}

\begin{prop}    \label{prop:SuperalgebraSplitProperty}
Let $\cstar{A}$, $\cstar{B} \subset B(\hilb{H})$ be supercommuting superfactors acting on a separable Hilbert superspace $\hilb{H}$, such that the von Neumann superalgebra $\cstar{A}\vee \cstar{B}$ is a superfactor. Suppose there exists a homogeneous vector $\xi\in \hilb{H}$ such that the associated vector state $\phi_{\xi}$ satisfies $\phi_{\xi}(ab)= \phi_{\xi}(a) \phi_{\xi}(b)$ for any $a \in \cstar{A}$, $b \in \cstar{B}$. Then there exists an isomorphism of von Neumann superalgebras $\cstar{A} \vee \cstar{B} \cong \cstar{A} \otimes \cstar{B}$. 
\end{prop}
\begin{proof}
Consider the Hilbert subspace $\hilb{V}=\overline{\cstar{A}\cstar{B}\xi}\subseteq \hilb{H}$, and let $(\pi_{\cstar{A}},\hilb{H}_{\cstar{A}})$, $(\pi_{\cstar{B}},\hilb{H}_{\cstar{B}})$ be the GNS representations of $\cstar{A}$, $\cstar{B}$, respectively, associated with the state $\phi_{\xi}$. The hypothesis implies that the map $ab\xi\rightarrow a\Omega_{\phi_{\xi}|_{\ql{A}}} \otimes b\Omega_{\phi_{\xi}|_{\ql{B}}} \in \hilb{H}_{\phi_{\xi}|_{\ql{A}}}\otimes \hilb{H}_{\phi_{\xi}|_{\ql{B}}}$ extends to a unitary intertwiner of Hilbert space representations of $*$-algebra $\cstar{A}\otimes_{\mathbb{C}} \cstar{B}$. Thus $\cstar{A}\vee\cstar{B}|_{\hilb{V}}\cong \cstar{A}\otimes \cstar{B}$. The restriction map $\cstar{A}\vee\cstar{B}\rightarrow \cstar{A}\vee\cstar{B}|_{\hilb{V}}$ is a surjective, graded normal $*$-homomorphism. But since $\cstar{A}\vee\cstar{B}$ is a superfactor, any graded normal $*$-homomorphism is injective, and thus the restriction gives the desired isomorphism. 
\end{proof}

We will use the following generalization of the result from \cite{ge1996tensor}:

\begin{lemma}\label{lma:gekadison}
(Theorem A of \cite{ge1996tensor}) Let $\cstar{A}$, $\cstar{B}$ be von Neumann superalgebras, and $\cstar{C}$ be the von Neumann supersubalgebra of $\cstar{A} \otimes \cstar{B}$ satisfying
\[\cstar{A} \otimes 1 \subset \cstar{C} \subset \cstar{A} \otimes \cstar{B}.\]
Suppose $\cstar{A}$ is a type $\rm{I}$ factor. Then    
\[\cstar{C} = \cstar{A} \otimes (\cstar{C} \cap (1 \otimes \cstar{B})).\]
\end{lemma}
\begin{proof}
Since $\cstar{A}$ is a type $\rm{I}$ factor, there exists a unitary element $\Gamma \in \cstar{A}$ such that $\Theta(a) = \Gamma a \Gamma$ for any $a \in \cstar{A}$, and $\Gamma^2 = 1$. The {\it Klein transform} is a bijective map
\[f:\CA \otimes \CB \to \CA \otimes_0 \CB,\]
to the ungraded spatial tensor product $\CA \otimes_0 \CB$ defined via
\[f(a \otimes 1) = a \otimes_0 1, \\
f(1 \otimes b) = \Gamma^{|b|} \otimes_0 b\]
for a homogeneous $b$. Since $\cstar{A}$ is a factor as an ungraded von Neumann algebra, by \cite[Theorem A]{ge1996tensor}, we have $f(\cstar{C}) = \cstar{A} \otimes_0 \cstar{S}$ for $\cstar{S} = f(\cstar{C}) \cap \l 1 \otimes_0 \cstar{B} \r$ which inherits a $\ZZ/2$-grading from $\cstar{B}$. Thus, every operator in $f(\cstar{C})$ can be expressed as
\[\sum_i a_i \otimes_0 s_i = \sum_i (a_i \Gamma^{|s_i|} \otimes_0 1) (\Gamma^{|s_i|} \otimes_0 s_i).\]
for homogeneous $\{s_i\}$. Now applying $f^{-1}$, we find that every operator in $\CC$ can be expressed as
\[
\sum_i a_i \Gamma^{|s_i|} \otimes s_i,
\]
implying $\cstar{C} = \cstar{A} \otimes \cstar{S}$.
\end{proof}

To close this section, we record a version of the slice map property for von Neumann algebras \cite{Tomiyama1976Fubini} adapted to the super setting. We say that a linear functional $\varphi$ on a superalgebra $\ql{A} = \ql{A}^{(0)} \oplus \cstar{A}^{(1)}$ is \textit{homogeneous} if there is a $d \in \{0,1\}$ such that $\varphi\circ\Theta=(-1)^{d}\varphi$. We denote this $d$ by $|\varphi|$ and call it {\it degree} of $\varphi$. If $|\varphi|=0$, $\varphi$ is supported on $\ql{A}^{(0)}$, and if $|\varphi|=1$, $\varphi$ is supported on $\ql{A}^{(1)}$. Any linear functional $\varphi$ has a canonical decomposition 

$$\varphi:=\varphi_{0}+\varphi_{1},$$

\medskip

\noindent where $\varphi_{0}=\frac{1}{2}(\varphi+\varphi\circ \fparity)$, and $\varphi_{1}=\frac{1}{2}(\varphi-\varphi\circ \fparity)$ are homogeneous of degree $0$ and $1$ respectively. Thus, the set of homogeneous linear functionals separate points in $\ql{A}$, i.e. for any two elements $a,b \in \cstar{A}$, there exists a homogeneous linear functional $\varphi$ such that $\varphi(a) \neq \varphi(b)$.

If $\cstar{A}$ and $\cstar{B}$ are von Neumann superalgebras, then any normal linear functional $\varphi$ on $\cstar{B}$ defines the \textit{right slice map} $R_{\varphi}: \cstar{A}\otimes \cstar{B}\rightarrow \cstar{A}$, given by

$$R_{\varphi}:=\text{Id}_{A}\otimes \varphi,$$

\noindent satisfying

$$R_{\varphi}(a \otimes b) = \varphi_{0}(b)a+\varphi_{1}(b)\fparity(a).$$

\noindent The key feature of slice maps is that for any $x \in \cstar{A} \otimes \cstar{B}, a \in \ql{A}$, we have 

\begin{equation}\label{eqn:supercommutator}
[R_{\varphi}(x), a]_{\pm} = R_{\varphi}\l [x, a \otimes 1]_{\pm} \r. 
\end{equation}

\medskip

This is elementary to verify using parity cases for simple tensors, and extends to arbitrary $x$ by continuity. It is also clear that the collection of $R_{\psi}$, where $\psi$ varies over normal functionals on $\cstar{B}$, separates points. In other words, for $x \in \cstar{A}\otimes \cstar{B}$, $x = 0$ if and only if $R_{\psi}(x)=0$ for all $\psi \in \cstar{B}_{*}$. We have the following proposition, which allows us to detect subalgebras of a tensor factor via slice maps.

\begin{prop}\label{prop:superslice}
Let $\cstar{A}$ and $\cstar{B}$ be von Neumann superalgebras, and $\cstar{C}\subseteq \cstar{A}$ be a von Neumann supersubalgebra. Then $\cstar{C}\otimes \cstar{B}=\{ x \in \cstar{A}\otimes \cstar{B}\ :\ R_{\varphi}(x)\in \cstar{C}\ \text{for all}\ \varphi\in \cstar{B}_{*}\}$. 
\end{prop} 

\begin{proof}
Clearly $\cstar{C}\otimes \cstar{B}\subseteq \{ x \in \cstar{A}\otimes \cstar{B}\ :\ R_{\varphi}(x)\in \cstar{C}\ \text{for all}\ \varphi\in \cstar{B}_{*}\}$. For the reverse inclusion, suppose $\hilb{H}$ and $\hilb{K}$ are Hilbert superspaces for representations of $\ql{A}$ and $\ql{B}$, respectively, so that we have an inclusion of superalgebras $\ql{A}\otimes \ql{B}\subseteq B(\hilb{H}\otimes \hilb{K})$. Since normal functionals on $\cstar{B}$ always have normal extensions to $B(\hilb{K})$, we have

$$\{ x \in \cstar{A}\otimes \cstar{B}\ :\ R_{\varphi}(x)\in \cstar{C}\ \text{for all}\ \varphi\in \cstar{B}_{*}\}=\{ x \in \cstar{A}\otimes \cstar{B}\ :\ R_{\psi}(x)\in \cstar{C}\ \text{for all}\ \psi\in B(\hilb{K})_{*}\},$$

\medskip

\noindent where here $R_{\psi}$ acts on $\cstar{A} \otimes \cstar{B}$ via the restriction of $R_{\psi}:B(\hilb{H})\otimes B(\hilb{K})\rightarrow B(\hilb{H})$.

Let $\mathcal{C}'$ be the supercommutant of $\mathcal{C}$ in $B(\hilb{H})$. Then for any $x \in \{ x \in \cstar{A}\otimes \cstar{B}\ :\ R_{\widetilde{\varphi}}(x)\in \cstar{C}\ \text{for all}\ \widetilde{\varphi}\in \cstar{B}(\hilb{K})_{*}\}$, $c\in \ql{C}'$, and $\psi\in B(\hilb{K})_{*}$, we have 

$$0=[R_{\psi}(x),c]_{\pm}=R_{\psi}([x, c\otimes 1]_{\pm}).$$

Then the above implies (for example, by \cite[Theorem 4.4]{CrismaleRossiZurlo2023}) $[x, c \otimes 1]_{\pm}=0$, so $x \in (\ql{C}' \otimes 1)' \cap (\ql{A}\otimes \ql{B})=\ql{C}\otimes \ql{B}$, where $(\ql{C}' \otimes 1)'$ denotes the supercommutant of $\ql{C}' \otimes 1$ in $B(\hilb{H} \otimes \hilb{K})$. This concludes the proof.
\end{proof}

\section{Quasi-local algebras}  \label{sec:QuasiLocalAlgebras}

Quasi-local algebras have been originally introduced by Haag and Kastler \cite{HaagKastler1964} as a way to encode locality in the operator algebraic approach to quantum field theory. Since then, they have found extensive applications to quantum statistical mechanics \cite{bratteli2012operator} as they provide a natural framework to study systems directly in the thermodynamic limit. There are varying definitions for quasi-local algebras depending on the context, and we record here the definition that is most natural for our purposes.

\begin{definition}\label{def-quasi-local-algebra}
Let $X$ be a metric space, and let $\CMcal{B}(X)$ be the collection of bounded subsets of $X$. A {\it quasi-local algebra} over $X$ is a unital $C^*$-superalgebra $\ql{A}$ with a net $\{\ql{A}_U\}_{U \in \CMcal{B}(X)}$ of separable von Neumann superalgebras with a common identity such that
\begin{itemize}
    \item if $U \subseteq V$, then $\ql{A}_U \subseteq \ql{A}_V$;
    \item $\ql{A}$ is the norm completion of the union $\bigcup_{U \in \CMcal{B}(X)} \ql{A}_U$;
    \item for $U \cap V = \emptyset$, the algebras $\ql{A}_U$ and $\ql{A}_V$ supercommute.
\end{itemize}
The $*$-automorphism implementing the grading is denoted $\fparity$ and is called {\it fermionic parity}. When $\fparity$ is trivial and we can disregard the grading, we call quasi-local algebras {\it bosonic}. When we want to emphasize that there is no restriction on $\fparity$, we call quasi-local algebras {\it fermionic}.
\end{definition}

\begin{remark}
Equivalently, a quasi-local algebra over a metric space $X$ is a functor $F:\mathbf{B}(X) \to \mathbf{sVN}_{\text{inj}}$ from the category $\mathbf{B}(X)$ of bounded subsets of $X$ with morphisms being inclusions to the category of von Neumann superalgebras $\mathbf{sVN}_{\text{inj}}$ with morphisms being normal graded injective $*$-homomorphisms satisfying the following locality condition:
\begin{itemize}
    \item for $U \cap V = \emptyset$, $F(U)$ and $F(V)$ supercommute inside $F(U \cup V)$.
\end{itemize}
\end{remark}

\begin{remark}
In this paper, we are mainly interested in $X$ being a uniformly locally finite metric space (such as $\ZZ^d$ or its subset for quantum lattice systems), i.e., metric spaces whose balls of radius $R$ have finitely many points bounded by some function that depends on $R$ only. For such spaces, the bounded subsets are precisely the finite subsets. More generally, we can consider metric spaces of bounded geometry which are also relevant for physical applications (such as $\RR^d$ or its subset for quantum field theories). However, since the latter always have coarsely dense uniformly locally finite metric subspaces, there is a natural way to discretize the corresponding quasi-local algebras without sacrificing the structures we care about in this paper.

In fact, most of our results depend only on the underlying coarse structure of the metric space (in the sense of Roe \cite{Roe2003CoarseGeometry}), as in \cite{ludewig2026quantum, EloklJones2025UniversalCoarseGeometry}. However, to avoid an additional mathematical formalism which may obscure the physical nature of our arguments, we stick to the class of uniformly locally finite metric spaces.

\end{remark}

Let $\ql{A}$, $\ql{B}$ be quasi-local algebras over a metric space $X$. We say that a map $\varphi: \ql{A} \to \ql{B}$ {\it has bounded spread} if for some $l \geq 0$, $\varphi(\ql{A}_U) \subseteq \ql{B}_{U^{+l}}$. In this case we say that $\varphi$ {\it has spread} $l$. We say that $\varphi$ is {\it locally normal} if $\varphi:\ql{A}_U \to \ql{B}_{U^{+l}}$ are normal for all $U \subset X$. A state on a quasi-local algebra $\ql{A}$ is called {\it locally normal} if it is normal on any $\ql{A}_U$.
\begin{definition}
Let $\ql{A}$, $\ql{B}$ be quasi-local algebras over a metric space $X$. A {\it bounded spread homomorphism} between $\ql{A}$ and $\ql{B}$ is a locally normal map $\varphi$ that has bounded spread $l$ for some $l \geq 0$ and such that $\varphi:\ql{A}_U \to \ql{B}_{U^{+l}}$ are $*$-homomorphisms of von Neumann algebras. The set of such homomorphisms is denoted $\Hombs(\ql{A},\ql{B})$. A {\it bounded spread isomorphism} is a $*$-isomorphism $\varphi$ such that $\varphi \in \Hombs(\ql{A},\ql{B})$ and $\varphi^{-1} \in \Hombs(\ql{B},\ql{A})$. The set of such isomorphisms is denoted $\Isobs(\ql{A},\ql{B})$. The set of {\it bounded spread automorphisms} of $\ql{A}$ is denoted $\Autbs(\ql{A}) := \Isobs(\ql{A},\ql{A})$. 
\end{definition}

\begin{remark}  \label{rmk:automaticUWcont}
Suppose $\cstar{M}$ is a von Neumann algebra that doesn't have a summand of the form $M_n(\cstar{N})$ for an infinite dimensional abelian von Neumann algebra $\cstar{N}$. Then by \cite[Theorem 6.5]{baudier2024embeddings}, any $*$-homomorphism from $\cstar{M}$ to a separable von Neumann algebra is automatically normal. Therefore, if for a quasi-local algebras $\ql{A}$ over $X$ all von Neumann algebras $\ql{A}_U$, $U \subset X$ have the same property as $\cstar{M}$, then any bounded spread homomorphism from $\ql{A}$ to a quasi-local algebra $\ql{B}$ is automatically locally normal. For example, stable quasi-local algebras $\ql{A}$ over a uniformly locally finite $X$ defined below satisfy this property automatically.
\end{remark}

We record the following observation as a lemma.

\begin{lemma}\label{bsinverse}
Suppose $\ql{A},\ql{B}$ are quasi-local algebras, $\varphi \in \Hombs(\ql{A},\ql{B})$ is injective, and there exists an $l\ge 0$ such that $\ql{B}_{U}\subseteq \varphi(\ql{A}_{U^{+l}})$. Then $\varphi\in \Isobs(\ql{A},\ql{B})$.
\end{lemma}

\begin{proof}
Since $\ql{B}_{U}\subseteq \varphi(\ql{A})$, clearly $\varphi$ is surjective, hence is an isomorphism of $C^*$-algebras. Now, to see that $\varphi^{-1}$ has bounded spread, our hypothesis implies that $\varphi^{-1}(\ql{B}_{U})\subseteq \varphi^{-1}(\varphi(\ql{A}_{U^{+l}}))=\ql{A}_{U^{+l}}$. thus $\varphi^{-1}\in \Hombs(\ql{B},\ql{A})$, and $\varphi\in \Isobs(\ql{A},\ql{B})$.
\end{proof}

\begin{definition}
For two quasi-local algebras $\ql{A}$, $\ql{B}$ over a metric space $X$, their tensor product is defined by the quasi-local algebra $\ql{A} \otimes \ql{B}$ over $X$, whose local algebras are $(\ql{A} \otimes \ql{B})_U = \ql{A}_U \otimes \ql{B}_U$ with natural inclusions $(\ql{A} \otimes \ql{B})_U\subseteq (\ql{A} \otimes \ql{B})_{V}$ for $U \subseteq V$, and whose underlying $C^*$-algebra is the colimit (in the category of $C^*$-algebras) over bounded subsets.
\end{definition}

\begin{definition}For a metric space $X$, we let $\QLAcat(X)$ (respectively, $\sQLAcat(X)$) be a symmetric monoidal category whose objects are bosonic (resp. fermionic) quasi-local algebras over $X$, whose morphisms are (graded) unital completely positive (UCP) maps between quasi-local algebras, and whose monoidal structure is given by the tensor product of quasi-local algebra.
\end{definition}

Let us introduce various terminology that will be useful in the following.
\begin{definition}
Let $\ql{A}$ be a quasi-local algebra over a metric space $X$. A {\it quasi-local subalgebra} $\ql{B}$ of $\ql{A}$ is a quasi-local algebra over $X$ such that $\ql{B}_U \subseteq \ql{A}_U$ for any bounded subset $U$ of $X$. In this case, we write $\ql{B} \subseteq \ql{A}$.
\end{definition}

Let $\ql{B} \subset \ql{A}$ be a quasi-local subalgebra of a quasi-local algebra $\ql{A}$ over a metric space $X$. Then for bounded subsets $U,V \subset X$, the algebra $\{a \in \ql{A}_U: \text{for any }b  \in \ql{B}_V \text{ we have }[a,b]_{\pm}=0\}$ is a von Neumann (super)algebra. Since the intersection of any number of von Neumann (super)algebras is again a von Neumann (super)algebra, the algebra $\{a \in \ql{A}_U: \text{for any }b  \in \ql{B} \text{ we have }[a,b]_{\pm}=0\}$ is a von Neumann (super)algebra as well.

\begin{definition}
Let $\ql{A}$ be a quasi-local algebra over a metric space $X$. 

\begin{enumerate}
\item 
The {\it (super)commutant} of a quasi-local subalgebra $\ql{B}$ of $\ql{A}$ is the quasi-local subalgebra $\opComm(\ql{B},\ql{A})$ defined by 
$$\opComm(\ql{B},\ql{A})_U = \{a \in \ql{A}_U: \text{for any }b  \in \ql{B} \text{ we have }[a,b]_{\pm}=0\}$$
for a bounded subset $U$ of $X$. 

\medskip

\item 
We say that $\ql{B} \subset \ql{A}$ is a {\it tensor factor} of $\ql{A}$ if the multiplication map $\ql{B} \times \opComm(\ql{B},\ql{A}) \to \ql{A}$, $b \times b' \to b b'$ extends to a bounded spread isomorphism $\mu:\ql{B} \otimes \opComm(\ql{B},\ql{A}) \to \ql{A}$.
\end{enumerate}
\end{definition}

Following \cite{doplicher1984standard}, we say that an inclusion of von Neumann (super)algebras $\cstar{B} \subset \cstar{A}$ is {\it split} if there exists a type $\rm{I}$ factor $\cstar{N}$ such that $\cstar{B} \subset \cstar{N} \subset \cstar{A}$. We say that it is {\it quasi-split} if there exists a type $\rm{I}$ factor $\cstar{M}$ such that the induced inclusion $\cstar{B} \otimes \CCC \subset \cstar{A} \otimes \cstar{M}$ is split. 

\begin{definition}
A quasi-local algebra $\ql{A}$ over a metric space $X$ is called {\it $l$-split} (resp. {\it $l$-quasi-split}) for a given $l \geq 0$ if for any bounded subset $U \subset X$, the inclusion $\ql{A}_U \subset \ql{A}_{U^{+l}}$ is split (resp. quasi-split). We say that it is {\it weakly split} (resp. {\it weakly quasi-split}) if it is $l$-split (resp. $l$-quasi-split) for some $l \geq 0$.
\end{definition}

\begin{definition}
A quasi-local algebra $\ql{A}$ over a metric space $X$ is called {\it $l$-additive} for a given $l \geq 0$ if for any finite collection of bounded subsets $U \subset X$, $\{V_a \subset X\}$ such that $U = \cup_a V_a$, we have $\ql{A}_U \subset \bigvee_{a} \ql{A}_{V^{+l}_a}$. We say that it is {\it weakly additive} if it is $l$-additive for some $l \geq 0$.
\end{definition}

\begin{definition}
Let $\ql{A}$ be a quasi-local algebra over a metric space $X$. A {\it complex conjugate quasi-local algebra} is a quasi-local algebra $\overline{\ql{A}}$ defined by the same data as $\ql{A}$, but with the scalar multiplication being $\lambda \cdot a = \bar{\lambda} a$, $\lambda \in \CCC$, $a \in \overline{\ql{A}}$. An {\it opposite quasi-local algebra} is a quasi-local algebra $\ql{A}^{op}$ defined by the same data as $\ql{A}$, but with the multiplication being $a \cdot b = b a$.
\end{definition}

\subsection{Coarse geometric perspective on quasi-local algebras}

Coarse geometry is the study of geometric structure viewed from the large-scale perspective \cite{Roe2003CoarseGeometry}. In particular, coarse equivalence between metric spaces ignores the local details of the geometry, but preserves certain global structure. For example, $\mathbb{Z}^{n}$ is coarsely equivalent to $\mathbb{R}^{n}$, but $\mathbb{R}^{n}$ is coarsely equivalent to $\mathbb{R}^{m}$ if and only if $n=m$. Coarse equivalence is best thought of as isomorphism in the coarse category of metric spaces, which we define.

Let $X$, $Y$ be metric spaces. Recall that a map $f:X\rightarrow Y$ is proper if the preimage of bounded sets is bounded. A proper map $f:X\rightarrow Y$ is called \textit{coarse} if for all $R\ge 0$, there is some $T\ge 0$ such that $|x-x'|\le R$ implies $\|f(x)-f(x')\|\le T$. Two coarse maps $f,g:X\rightarrow Y$ are said to be \textit{close} if there exists some $R\ge 0$ with $|f(x)-g(x)|\le R$ for all $x\in X$. Closeness defines an equivalence relation on the collection of coarse maps from $X$ to $Y$.

\begin{definition} The coarse category of metric spaces, denoted $\textbf{Coarse}$, is the category whose objects are metric spaces and morphisms are coarse maps up to closeness, with categorical composition induced from the composition of maps.
\end{definition}

A coarse equivalence is a representative of an isomorphism in the coarse category. Two metric spaces are coarsely equivalent if there exists a coarse equivalence between them, or equivalently if they are isomorphic in the coarse category.

\begin{definition} Let $\textbf{SymMonCat}$ denote the category whose objects are symmetric monoidal categories, and whose morphisms are symmetric monoidal functors, taken up to monoidal natural isomorphism.
\end{definition}

We will see that the assignment $X\mapsto \opQLAcat(X)$ extends to a functor $\textbf{Coarse}\rightarrow \textbf{SymMonCat}$. The vehicle for this functor is the pushforward.

\begin{definition}
Let $f:X \to Y$ be a proper map of proper metric spaces. Given a quasi-local algebra $\ql{A}$ over $X$, the {\it pushforward quasi-local algebra} $f_* \ql{A}$ over $Y$ is defined by $(f_* \ql{A})_U := \ql{A}_{f^{-1}(U)}$.
\end{definition}

\begin{prop}
If $f:X\rightarrow Y$ is a coarse map of metric spaces, then the pushforward quasi-local algebra operation on quasi-local algebras over $X$ extends to a symmetric monoidal functor of symmetric monoidal categories $f_{*}: \opQLAcat(X)\rightarrow \opQLAcat(Y)$. If $f$ is close to $g$, then the functors $f_{*}$ and $g_{*}$ are monoidally naturally isomorphic.  
\end{prop}

\begin{proof}
First we see that $f_{*}$ indeed naturally extends to a functor. Let $\ql{A},\ql{B}\in \opQLAcat(X)$, and $\psi:\ql{A}\rightarrow\ql{B}$ a (graded) UCP map with bounded spread.
Note that as $C^*$-algebras, we naturally have $f_{*}\ql{A}\cong \ql{A}$ and $f_{*}\ql{B} \cong \ql{B}$, with the difference being the quasi-local structure. In particular, using this natural identification, we have a well defined UCP map $f_{*}\psi:f_{*}\ql{A}\rightarrow f_{*}\ql{B}$. It remains to show $f_{*}\psi$ has bounded spread. Pick any bounded region $U\subseteq Y$. If $\psi$ has spread $l$, note that the coarse condition guarantees there is some $t$ such that $f^{-1}(U^{+t})$ contains $f^{-1}(U)^{+l}$. Then
 
$$f_{*}\psi(f_{*}\ql{A}_{U})=\psi(\ql{A}_{f^{-1}(U)})\subseteq \ql{B}_{f^{-1}(U)^{+l}}\subseteq \ql{B}_{f^{-1}(U^{+t})}=f_{*}\ql{B}_{U^{+t}}.$$

\medskip

\noindent Clearly from the definitions, if $\psi:\ql{A}\rightarrow \ql{B}$ and $\phi:\ql{B}\rightarrow \ql{C}$, then $f_{*}(\psi\circ \phi)=f_{*}\psi \circ f_{*}\phi$. The fact that $f_{*}$ is symmetric monoidal is trivial from the definitions.

Now suppose $f$ is close to $g$, with constant $R$. The then claim is that the composite of the canonical isomorphisms of $C^*$-algebras $f_{*}\ql{A}\cong \ql{A}\cong g_{*}\ql{A}$ is a bounded spread isomorphism. Let us denote this composite $\beta_{f,g}$. Then

$$\beta_{f,g}(f_{*}\ql{A}_{U})=\ql{A}_{f^{-1}(U)}\subseteq \ql{A}_{f^{-1}(U)^{+R}}\subseteq \ql{A}_{g^{-1}(U^{+T})}=\ql{B}_{U^{+T}}$$

\medskip

\noindent from the definitions, this is clearly a monoidal natural isomorphism, that preserves the braiding.

\end{proof}

\begin{corollary}
The assignment $X\mapsto \opQLAcat(X)$ extends to a functor $\textbf{Coarse}\rightarrow \textbf{SymMonCat}$. In particular, if $X$ is coarsely equivalent to $Y$, then $\opQLAcat(X)\cong \opQLAcat(Y)$ as symmetric monoidal categories.
\end{corollary}

\begin{remark}\label{rem:discretizations}The above proposition allows us to coherently perform \textit{discretizations}. Suppose $\ql{A}\in \opQLAcat(X)$, and suppose we have a discrete subset $L\subseteq X$, such that the inclusion $L\hookrightarrow X$ is a coarse equivalence. If we let $j:X\rightarrow L$ be a coarse inverse, then the pushforward functor $j_{*}:\opQLAcat(X)\cong \opQLAcat(L)$ is an equivalence of categories. Then for $\ql{A}\in \opQLAcat(X)$,  $j_{*}\ql{A}\in \opQLAcat(L)$ is called a \textit{discretization}. Functoriality allows us to keep track of discretizations obtained from picking different discrete subsets, so that any two discretizations will be canonically bounded spread isomorphic. Of particular importance for us is the standard inclusion $\mathbb{Z}^d\subseteq \mathbb{R}^d$. \end{remark}

\subsection{Invertible quasi-local algebras and the Brauer group}\label{subsecbrauergroup}

\begin{definition}\label{definitionbosonicproductalgebra}
Let $X$ be a uniformly locally finite metric space and $s:X \to \NNb$ be a function. A {\it bosonic product algebra} over $X$ of spin $s$ is the quasi-local algebra $\Mat(X,s)$ such that for any finite $U \subset X$, we have $\Mat(X,s)_U = B(\bigotimes_{j \in U} \hilb{H}_j)$ for separable Hilbert spaces $\{\hilb{H}_j\}_{j \in X}$ with $\dim \hilb{H}_j = s(j)$. A {\it bosonic stable product algebra} over $X$ is the quasi-local algebra $\Mat(X):= \Mat(X,s)$ for $s$ satisfying $s(j) = \infty$ for any $j \in X$. Note that $\Mat(X)$ is absorbing: $\Mat(X,s) \otimes \Mat(X) \cong \Mat(X)$.
\end{definition}

\begin{definition}\label{definitionfermionicproductalgebra}
Let $X$ be a uniformly locally finite metric space and $s:X \to \NNb|\NNb$ be a function\footnote{Here by $\NNb|\NNb$ we mean the set of pairs $n|m$ of elements $n,m$ of $\NNb$.}. A {\it fermionic product algebra} $\sMat(X,s)$ over $X$ of spin $s$ is a fermionic quasi-local algebra such that for any finite $U \subset X$, we have $\sMat(X,s)_U = B(\bigotimes_{j \in U} \hilb{H}_j)$ for separable Hilbert superspaces $\{\hilb{H}_j\}_{j \in X}$ with $\text{sdim}(\hilb{H}_j) = s(j)$. A {\it fermionic stable product algebra} over $X$ is the quasi-local algebra $\sMat(X):= \sMat(X,s)$ for $s$ satisfying $s(j) = \infty|\infty$ for any $j \in X$. Note that $\sMat(X)$ is absorbing: $\sMat(X,s) \otimes \sMat(X) \cong \sMat(X)$.
\end{definition}

\begin{definition}
We say that two bosonic quasi-local algebras $\ql{A}$, $\ql{B}$ over a uniformly locally finite metric space $X$ are Brauer equivalent if $\ql{A} \otimes \Mat(X)$ and $\ql{B} \otimes \Mat(X)$ are bounded spread isomorphic. We say that two fermionic quasi-local algebras $\ql{A}$, $\ql{B}$ over a uniformly locally finite metric space $X$ are super Brauer equivalent if $\ql{A} \otimes \sMat(X)$ and $\ql{B} \otimes \sMat(X)$ are bounded spread isomorphic.    
\end{definition}

We say that a bosonic (resp. fermionic) quasi-local algebra $\ql{A}$ over a uniformly locally finite metric space $X$ is {\it stable} if $\ql{A} \otimes \opMat(X) \cong \ql{A}$. Two stable quasi-local algebras $\ql{A}$, $\ql{B}$ are (super) Brauer equivalent if they are bounded spread isomorphic. 

\begin{definition}  \label{def:Brgrp}
We say that a bosonic quasi-local algebra $\ql{A}$ over a uniformly locally finite metric space $X$ is {\it invertible} if there exists a bosonic quasi-local algebra $\tilde{\ql{A}}$ over $X$, such that $\ql{A} \otimes \tilde{\ql{A}}$ is bounded spread isomorphic to $\Mat(X)$. 

\begin{enumerate}
\item 
We say that $\ql{A}$ has spread $l$ if there exists $\tilde{\ql{A}}$ and a bounded spread isomorphism $\ql{A} \otimes \tilde{\ql{A}} \to \Mat(X)$  of spread $l$. 
\item
Brauer equivalence classes of invertible bosonic quasi-local algebras form an abelian group $\Brgr(X)$ called {\it Brauer group} defined via $[\ql{A}] \cdot [\ql{B}] = [\ql{A} \otimes \ql{B}]$.
\item 
We denote by $\Azcat(X)$ the symmetric monoidal subcategory of $\QLAcat(X)$ consisting of invertible bosonic quasi-local algebras as objects and bounded spread isomorphisms as morphisms.
\end{enumerate}
\end{definition}

\begin{definition}  \label{def:sBrgrp}
We say that a quasi-local algebra $\ql{A}$ over a uniformly locally finite metric space $X$ is {\it invertible} if there exists a quasi-local algebra $\tilde{\ql{A}}$ over $X$, such that $\ql{A} \otimes \tilde{\ql{A}}$ is bounded spread isomorphic to $\sMat(X)$. 

\begin{enumerate}
\item 
We say that $\ql{A}$ has spread $l$ if there exists $\tilde{\ql{A}}$ and a bounded spread isomorphism $\ql{A} \otimes \tilde{\ql{A}} \to \sMat(X)$  of spread $l$. 
\item
Brauer equivalence classes of invertible quasi-local algebras form an abelian group $\sBrgr(X)$ called the {\it super Brauer group} defined via $[\ql{A}] \cdot [\ql{B}] = [\ql{A} \otimes \ql{B}]$.
\item 
We denote by $\sAzcat(X)$ the symmetric monoidal subcategory of $\sQLAcat(X)$ consisting of invertible  quasi-local algebras as objects and bounded spread isomorphisms as morphisms.
\end{enumerate}
\end{definition}

\begin{prop}    \label{prop:0dBrauerGroup}
We have $\Brgr(\rm pt) = 0$ and $\sBrgr(\rm pt) = \ZZ/2$.
\end{prop}
\begin{proof}
For $X = \text{pt}$, a quasi-local algebra is just a von Neumann (super)algebra. 

Let $\cstar{A}$, $\tilde{\cstar{A}}$ be von Neumann algebras such that $\cstar{A} \otimes \tilde{\cstar{A}} \cong B(\hilb{H})$ for a separable infinite-dimensional Hilbert space. $\cstar{A}$, $\tilde{\cstar{A}}$ can not have a non-trivial center since otherwise the center of $B(\hilb{H})$ would be non-trivial. A tensor product of two factors is type $\rm{I}$ if and only if both of them are type $\rm{I}$. Thus, $\cstar{A} \otimes B(\hilb{H}) \cong B(\hilb{H})$ and $\Brgr(\rm pt) = 0$.

Let $\cstar{A}$, $\tilde{\cstar{A}}$ be von Neumann superalgebras such that $\cstar{A} \otimes \tilde{\cstar{A}} \cong B(\hilb{H})$ for a separable Hilbert superspace of superdimension $\infty|\infty$. $\cstar{A}$, $\tilde{\cstar{A}}$ can not have a non-trivial supercenter since otherwise the supercenter of $B(\hilb{H})$ would be non-trivial. Hence, both $\cstar{A}$, $\tilde{\cstar{A}}$ are superfactors.

By Proposition \ref{prop:SuperfactorCharacterization}, both $\cstar{A}$ and $\tilde{\cstar{A}}$ are either factors or are isomorphic to $\cstar{M} \otimes \text{Cl}_1$ for a factor $\cstar{M}$ with a trivial $\ZZ/2$-grading. In both cases, since $\cstar{A} \otimes \tilde{\cstar{A}} \cong B(\hilb{H})$, the factors must be type $\rm{I}$. Thus, either $\cstar{A} \otimes B(\hilb{H}) \cong \tilde{\cstar{A}} \otimes B(\hilb{H}) \cong B(\hilb{H})$ or $\cstar{A} \otimes B(\hilb{H}) \cong \tilde{\cstar{A}} \otimes B(\hilb{H}) \cong B(\hilb{H}) \otimes \text{Cl}_1$, implying $\sBrgr(\rm pt) = \ZZ/2$.
\end{proof}

Let us show some basic properties of invertible quasi-local algebras.

\begin{prop}[Quasi-split property] \label{prop:SplitPropertyForQLAlgebras}
An invertible quasi-local algebra $\ql{A}$ is weakly quasi-split. Moreover, if $\ql{A}$ has spread $l$, then $\ql{A}$ is $3l$-quasi-split. 
\end{prop}

\begin{proof}
Let $\varphi:\ql{A} \otimes \tilde{\ql{A}} \to \opMat(X)$ be a bounded spread isomorphism of spread $l$. Let $\varphi_{23}:\ql{A} \otimes \opMat(X) \to \ql{A} \otimes \tilde{\ql{A}} \otimes \ql{A}$ and $\varphi_{12}:\ql{A} \otimes \tilde{\ql{A}} \otimes \ql{A} \to \opMat(X) \otimes \ql{A}$ be the induced bounded spread isomorphisms. We have
$\varphi_{12}(\varphi_{23}(\ql{A}_U \otimes \CCC)) \subset \opMat(X)_{U^{+l}} \otimes \CCC$ and $\varphi^{-1}_{23}( \varphi^{-1}_{12} (\opMat(X)_{U^{+l}} \otimes \CCC)) \subset \ql{A}_{U^{+3l}} \otimes \opMat(X)_{U^{+3l}}$. It follows that the inclusion $\ql{A}_U \subset \ql{A}_{U^{+3l}}$ is quasi-split.
\end{proof}
By applying stabilization, we immediately get the following
\begin{corollary}[Split property in the stable case] 
A stable invertible quasi-local algebra $\ql{A}$ is weakly split. Moreover, if $\ql{A}$ has spread $l$, then $\ql{A}$ is $3l$-split. 
\end{corollary}

\begin{remark}
There are unstable invertible quasi-local algebra that are not weakly split. It was recently shown that the quasi-local algebras associated to anyon chains are typically invertible in our sense \cite{bunner2026universalfusioncategorysymmetries}. However, for quasi-local algebras associated to anyon chains built from fusion categories with non-integer dimension objects, we typically do not have the weak split property. Indeed, if we did, by finite-dimensionality the intertwining type $\rm{I}$ factors would necessarily be finite-dimensional, hence the quasi-local algebras would be increasing unions of matrix algebras, hence UHF algebras. Such algebras $\ql{A}$ have ordered $K_{0}$ isomorphic to an additive subgroup of $\mathbb{Q}$. In particular, for a fusion category $C$ with an object of non-integer dimension, there is no unitary tensor functor $C \rightarrow \text{Bim}(\ql{A})$ \cite[Proposition 5.2]{ChenHernandezPalomaresJonesPenneys2022}. However, for an anyon chain built from a fusion category $C$ with quasi-local algebra $\ql{A}$, there is a canonical tensor functor $C\rightarrow \text{Bim}(\ql{A})$, \cite[Example 3.18]{evans2026operatoralgebraicapproachfusion}. The simplest example is the Golden chain, built from the fusion category \textbf{Fib} \cite{Feiguin2007GoldenChain}.
\end{remark}

\begin{prop}[Additivity] \label{prop:AdditivityForQLAlgebras}
An invertible quasi-local algebra $\ql{A}$ is weakly additive. Moreover, if $\ql{A}$ has spread $l$, then $\ql{A}$ is $2l$-additive.
\end{prop}
\begin{proof}
Suppose $\ql{A}$ is defined over a uniformly locally finite metric space $X$ and has spread $l$. There exists a spread $l$ isomorphism $\varphi:\ql{A} \otimes \tilde{\ql{A}} \to \opMat(X)$. Let $\ql{B} = \ql{A} \otimes \tilde{\ql{A}}$ and $U$, $V_1$, $V_2$ be bounded subsets of $X$ such that $U = V_1 \cup V_2$. We have $\varphi(\ql{B}_U) \subseteq \opMat(X)_{U^{+l}} \subseteq \opMat(X)_{V_1^{+l}} \vee \opMat(X)_{V_2^{+l}} \subseteq \ql{B}_{V_1^{+2l}} \vee \ql{B}_{V_2^{+2l}}$. Therefore, $\ql{B}$ is $2l$-additive. $2l$-additivity of $\ql{A} \otimes \tilde{\ql{A}}$ implies $2l$-additivity of each tensor factor.
\end{proof}

In the following, we record some basic results on tensor factors, which are very useful when considering invertible quasi-local algebras. The next several lemmas follow \cite[Section 2]{ludewig2026quantum}, with some modifications due to the fact that we are in the super, locally infinite-dimensional setting. 

\begin{lemma}\label{lem:alternatetensorfactor}
Let $\ql{A}$ be a quasi-local algebra over a metric space $X$. A quasi-local subalgebra $\ql{B} \subset \ql{A}$ is a tensor factor if and only if there exists $l \geq 0$ such that for any bounded $U \subset X$, we have $\ql{A}_U \subset \ql{B}_{U^{+l}} \vee \opComm(\ql{B},\ql{A})_{U^{+l}}$ and the multiplication map $\mu:\ql{B}_U \otimes \opComm(\ql{B},\ql{A})_U \to \ql{A}_U$ is injective.
\end{lemma}

\begin{proof}
Let $\tilde{\ql{B}} = \opComm(\ql{B},\ql{A})$. If $\ql{B}$ is a tensor factor, then there exists a bounded spread isomorphism $\varphi:\ql{B} \otimes \tilde{\ql{B}} \to \ql{A}$ which implies injectivity and $\ql{A}_U \subset \ql{B}_{U^{+l}} \otimes \tilde{\ql{B}}_{U^{+l}}$.

Conversely, suppose the hypotheses of the lemma are satisfied. Then the multiplication defines an injective map $\mu:\ql{B}_U \otimes \opComm(\ql{B},\ql{A})_U \to \ql{A}_U$. Since the image of $\ql{B}_{U^{+l}} \otimes \opComm(\ql{B},\ql{A})_{U^{+l}}$ is $ \ql{B}_{U^{+l}} \vee \opComm(\ql{B},\ql{A})_{U^{+l}}$, then the second hypothesis ensures that $\mu$ is is a bounded spread isomorphism by Lemma \ref{bsinverse}.
\end{proof}

\begin{lemma}   \label{lma:Lud228}
Let $\varphi:\ql{A} \to \ql{B}$ be a bounded spread isomorphism of quasi-local algebras, and let $\ql{C} \subseteq \ql{A}$ be quasi-local subalgebra which is a tensor factor with $\tilde{\ql{C}} = \opComm(\ql{C},\ql{A})$. Then the assignments
$$
U\longmapsto \ql{B}_U\cap\varphi(\ql{C}),
\qquad
U\longmapsto \ql{B}_U\cap\varphi(\tilde{\ql{C}})
$$
define quasi-local subalgebras $\ql{D},\tilde{\ql{D}}\subseteq \ql{B}$ such that
\begin{enumerate}
\item $\varphi$ restricts to bounded spread isomorphisms $\ql{C}\to\ql{D}$ and
$\tilde{\ql{C}}\to\tilde{\ql{D}}$;
\item $\tilde{\ql{D}}=\opComm(\ql{D},\ql{B})$;
\item $\ql{D}$ is a tensor factor of $\ql{B}$.
\end{enumerate}    
\end{lemma}
\begin{proof}
Let $\mu_{\ql{A}}:\ql{C}\otimes\tilde{\ql{C}}\to\ql{A}$ be the multiplication bounded spread isomorphism, and let
$l$ bound the spreads of $\varphi$, $\varphi^{-1}$, $\mu_{\ql{A}}$ and $\mu_{\ql{A}}^{-1}$.

Let $c \in \ql{C}$ satisfy $\varphi(c) \in \ql{B}_U$. Since $\varphi^{-1}$ and $\mu^{-1}_{\ql{A}}$ have spread $l$ and $\mu^{-1}_{\ql{A}}(c) = c \otimes 1$, we have
$$
c \otimes 1 \in (\ql{C} \otimes \tilde{\ql{C}})_{U^{+2l}}
=\ql{C}_{U^{+2l}} \otimes \tilde{\ql{C}}_{U^{+2l}}.
$$
Applying the slice map $R_{\psi}$ for a normal state $\psi$ of degree $0$ on
$\tilde{\ql{C}}_{U^{+2l}}$ gives $c = R_{\psi}(c \otimes 1) \in \ql{C}_{U^{+2l}}$. Hence,
$$
\ql{D}_U := \ql{B}_U \cap \varphi(\ql{C}) = \ql{B}_U \cap \varphi \bigl(\ql{C}_{U^{+2l}} \bigr)
$$
is a von Neumann (super)algebra. Thus, $\ql{D}$ is a quasi-local subalgebra
of $\ql{B}$ with underlying $C^*$-algebra $\varphi(\ql{C})$.  Slicing in the first tensor factor instead, the same argument applied to $1 \otimes \tilde{c}$ for $\tilde{c} \in \tilde{\ql{C}}$ satisfying $\varphi(\tilde{c}) \in \ql{B}_U$ gives the corresponding statements for $\tilde{\ql{D}}$.

(1) Since $\varphi(\ql{C}_U)\subseteq\ql{D}_{U^{+l}}$ and
$\ql{D}_U\subseteq\varphi(\ql{C}_{U^{+2l}})$, $\varphi:\ql{C}\to\ql{D}$ is a bounded spread isomorphism. Likewise for $\tilde{\ql{C}} \to \tilde{\ql{D}}$.

(2) An element of $\tilde{\ql{D}}_U$ lies in $\ql{B}_U$ and in $\varphi(\tilde{\ql{C}})$, which
(super)commutes with $\varphi(\ql{C}) \supseteq \ql{D}$. Hence,
$\tilde{\ql{D}} \subseteq \opComm(\ql{D},\ql{B})$. Conversely, let
$b \in \opComm(\ql{D},\ql{B})_U$. The algebra $\varphi(\ql{C})$ is the norm closure of
$\bigcup_V \ql{D}_V$, so $b$ (super)commutes with all of $\varphi(\ql{C})$. Therefore, the element $\varphi^{-1}(b)$ (super)commutes with $\ql{C}$, i.e. $\varphi^{-1}(b) \in \tilde{\ql{C}}$. Hence, $b\in\ql{B}_U\cap\varphi(\tilde{\ql{C}})=\tilde{\ql{D}}_U$.

(3) The map $\varphi\otimes\varphi:\ql{C}\otimes\tilde{\ql{C}}\to\ql{D}\otimes\tilde{\ql{D}}$ is
a bounded spread isomorphism. Hence, $\mu_{\ql{B}} = \varphi \mu_{\ql{A}} (\varphi\otimes\varphi)^{-1}$ defines a bounded spread isomorphism $\mu_{\ql{B}}: \ql{D} \otimes \tilde{\ql{D}} \to \ql{B}$ which extends the multiplication map $ \ql{D} \times \tilde{\ql{D}} \to \ql{B}$.

\end{proof}

\begin{lemma}   \label{lma:Lud230}
Let $\ql{A} \subseteq \ql{B} \subseteq \ql{C}$ be quasi-local algebras over a metric space $X$. If $\ql{A}$ and $\ql{B}$ are tensor factors of $\ql{C}$, then $\ql{A}$ is a tensor factor of $\ql{B}$. 
\end{lemma}
\begin{proof}
Let $\tilde{\ql{A}} = \opComm(\ql{A},\ql{C})$ and $\tilde{\ql{B}} = \opComm(\ql{B},\ql{C})$. Let $\ql{D}=\tilde{\ql{A}}\cap \ql{B}$. Then since $\mu:\ql{A} \otimes \tilde{\ql{A}} \rightarrow \ql{C}$ is injective, its restriction to $\ql{A} \otimes \ql{D} \rightarrow \ql{B}$ is injective. 

By Lemma \ref{lem:alternatetensorfactor}, there is an $l$ such that for any $U$, $\ql{C}_{U}\subseteq \ql{B}_{U^{+l}}\vee \tilde{\ql{B}}_{U^{+l}}\cong \ql{B}_{U^{+l}}\otimes \tilde{\ql{B}}_{U^{+l}}$.

Thus 

$$\tilde{\ql{A}}_{U}=\tilde{\ql{A}}\cap \ql{C}_{U}\subseteq \tilde{\ql{A}} \cap (\ql{B}_{U^{+l}} \otimes \tilde{\ql{B}}_{U^{+l}}).$$

\medskip

\noindent We claim $\tilde{\ql{A}} \cap (\ql{B}_{U^{+l}}\otimes \tilde{\ql{B}}_{U^{+l}})=(\tilde{\ql{A}}\cap \ql{B}_{U^{+l}})\otimes \tilde{\ql{B}}_{U^{+l}}$. The inclusion 

$$(\tilde{\ql{A}} \cap \ql{B}_{U^{+l}})\otimes \tilde{\ql{B}}_{U^{+l}} \subseteq \tilde{\ql{A}} \cap (\ql{B}_{U^{+l}}\otimes \tilde{\ql{B}}_{U^{+l}})$$ 

\medskip

\noindent is obvious, so it remains to show the other direction.

Now, let $\psi$ be a normal functional on $\tilde{\ql{B}}_{U^{+l}}$, and let $R_{\psi}:\ql{B}_{U^{+l}}\otimes \tilde{\ql{B}}_{U^{+l}}\rightarrow \ql{B}_{U^{+l}}$ be the super slice map defined above. Now, let $c\in \tilde{\ql{A}}\cap (\ql{B}_{U^{+l}}\otimes \tilde{\ql{B}}_{U^{+l}})$. Then for any $a\in \ql{A}$, recalling Eq. (\ref{eqn:supercommutator}),

$$[R_{\psi}(c),a]_{\pm}=R_{\psi}([c,a\otimes 1]_{\pm})=0.$$

\noindent In particular, $R_{\psi}(c)\in \tilde{\ql{A}}\cap \ql{B}_{U^{+l}}$, and thus by Proposition \ref{prop:superslice}, $c \in (\tilde{\ql{A}}\cap \ql{B}_{U^{+l}}) \otimes \tilde{\ql{B}}_{U^{+l}}$.

Using Lemma \ref{lem:alternatetensorfactor}, this implies
$$\ql{A}\otimes \ql{D}\rightarrow \ql{B}$$
is a bounded spread isomorphism. Hence, $\ql{A}$ is a tensor factor of $\ql{B}$.

\end{proof}

\section{Quantum Cellular Automata} \label{sec:QCAs}

In this section, we study quantum cellular automata (QCA) on product algebras (see the review \cite{Farrelly2020} and references there-in) and their relationship to invertible quasi-local algebras.

\begin{definition}
    Let $X$ be a uniformly locally finite metric space. A {\it quantum cellular automaton} (QCA) on $\opMat(X,s)$ is a bounded spread automorphism $\alpha$ of $\opMat(X,s)$ for some $s$. In this case we say that $\alpha$ is a QCA {\it over $X$}.\\ 
\end{definition}

A standard example of QCAs is (finite-depth) quantum circuits.

\begin{definition}
Let $X$ be a uniformly locally finite metric space, and $\ql{A}$ be a quasi-local algebra over $X$. A {\it quantum circuit} on $\ql{A}$ consists of the following data:
\begin{itemize}
    \item A collection of subsets $\{U^{(a)}_k \subset X\}^{a \in I_k}_{k=1,...,n}$ of uniformly bounded diameter satisfying $U^{(a)}_k \cap U^{(b)}_k = \emptyset$ for any $k=1,...,n$ and $a,b \in I_k$, $a \neq b$. The integer $n$ is called the {\it number of layers} in a quantum circuit, and the quantum circuit is called {\it $n$-layered}.
    \item A collection of (even) unitary elements $u^{(a)}_k \in \ql{A}_{U^{(a)}_k}$ called {\it gates}. We say that $U^{(a)}_k$ is the {\it support} of the gate $u^{(a)}_k$.
    \item A (graded) $*$-automorphism $\alpha = \beta_1 \beta_2 ... \beta_n$ of $\ql{A}$, where $\beta_k = \prod_{a \in I_k} \Ad_{u^{(a)}_k}$.
\end{itemize}
A (graded) $*$-automorphism $\alpha$ of $\ql{A}$ is called a {\it quantum circuit automorphism} if it appears in the data of a quantum circuit.
\end{definition}

Note that quantum circuit automorphisms on $\opMat(X,s)$ are QCAs by construction. Quantum circuits are bounded spread versions of ``locally generated" automorphisms of quasi-local algebras that are commonly used in the classification of topological phases of matter. We consider circuits as ``trivial" QCAs, and are led to the following notions.

\begin{definition}
A {\it stabilization} of a QCA $\alpha$ on $\opMat(X,s)$ is a QCA on $\opMat(X)$ of the form $\sigma^{-1}(\alpha \otimes \Id) \sigma$, where $\sigma: \opMat(X) \to \opMat(X,s) \otimes \opMat(X)$ is a zero-spread isomorphism. A QCA $\alpha$ is called {\it stable} if it is related to its stabilization by a zero-spread isomorphism.
\end{definition}

\begin{definition}\label{def:QCA}
Two QCAs $\alpha, \beta$ over $X$ are \textit{in the same phase}, denoted $\alpha\sim \beta$, if there exists a quantum circuit $\gamma$ on $\opMat(X)$ such that $\gamma \tilde{\alpha} = \tilde{\beta}$ for stabilizations $\tilde{\alpha}$, $\tilde{\beta}$ of $\alpha$, $\beta$, respectively. Being in the same phase defines an equivalence relation. We denote the collection of equivalence classes $\opQCA(X)$. 
\end{definition}

Note that a QCA on $\opMat(X)$ is stable if we have  $\alpha= \sigma^{-1}_{2}(\alpha\otimes \Id)\sigma_{1}$ for some $0$-spread isomorphisms $\sigma_{1},\sigma_{2}: \opMat(X)\rightarrow \opMat(X)\otimes \opMat(X)$. We also have that $\alpha \sim \sigma^{-1}_{2}(\alpha\otimes \Id)\sigma_{1}$ for any QCA $\alpha$, which follows from the fact that zero-spread automorphisms of $\opMat(X)$ are one-layered quantum circuits. Thus, every QCA is in the same phase as any of its stabilizations, and for the analysis of equivalence classes it is enough to consider QCAs on $\opMat(X)$.

By the usual QCA arguments $\opQCA(X)$ forms an abelian group. We record this in the following lemma:

\begin{lemma}   \label{lma:QCACompositionInBrGr}
$\opQCA(X)$ is an abelian group with the product of equivalence classes induced by composition.
\end{lemma}

\begin{proof} 
Let $\tau$ be a (super)swap automorphism of $\opMat(X,s)\otimes \opMat(X,t)$ defined by $\tau(a \otimes b) = b \otimes a$. It is a one-layered quantum circuit. Hence, for any zero-spread isomorphism $\sigma:\opMat(X,s \cdot t)\cong \opMat(X,s)\otimes \opMat(X,t)$, where $(s \cdot t)(x) = s(x) t(x)$, we have

$$\beta \alpha\sim  \sigma^{-1}( \tau(\beta\otimes \Id) \tau (\alpha\otimes \Id ))\sigma=\sigma^{-1}(\alpha\otimes \beta)\sigma=\sigma^{-1}((\alpha\otimes\Id) \tau(\beta\otimes\Id)\tau)\sigma\sim \alpha\beta.$$

\end{proof}

\begin{prop} \label{prop:QCAs_on_a_half-space_are_QC}
For a uniformly locally finite metric space $X$, we have $\opQCA(X\times \ZZ_{\ge 0}) = 0$.
\end{prop}
\begin{proof}
Without loss of generality, we can assume that $\alpha$ is stable. Suppose $\alpha$ has spread $l$. By the proof of Lemma \ref{lma:QCACompositionInBrGr}, the automorphism $\alpha \otimes \alpha^{-1}$ on $\opMat(X\times \ZZ_{\ge0})\otimes \opMat(X\times \ZZ_{\ge 0})$ is a quantum circuit.

Let $\tilde{\ql{A}}$ be a quasi-local algebra over $X \times \ZZ_{\geq 0}$ that is defined by a formal composition of quasi-local algebras $\{ \opMat(X \times \ZZ_{\geq n}) \}_{n \in \ZZ_{\geq 0}}$. Since any bounded subset $U \subset X \times \ZZ_{\geq 0}$ overlaps with $\{X \times \ZZ_{\geq n}\}_{n \in \ZZ_{\geq 0}}$ only for finitely many $n$ and $B(\hilb{H})^{\otimes k} \cong B(\hilb{H})$ for any $k \in \NN$ and any separable Hilbert (super)space $\hilb{H}$ of infinite dimension (of superdimension $\infty|\infty$), we have a zero-spread isomorphism $\sigma: \tilde{\ql{A}} \to \opMat(X\times \ZZ_{\ge 0})$. 

Let $\alpha^{(n)}$ be QCAs on $\opMat(X \times \ZZ_{\geq n})$ naturally induced via shifts $(x,m) \to (x,m+n)$, $(x,m) \in X \times \ZZ_{\geq 0}$ by $\alpha$, and let $\beta$ be a QCA on $\tilde{\ql{A}}$ that is formally defined by an infinite tensor product $(\alpha^{(0)})^{-1} \otimes \alpha^{(1)} \otimes (\alpha^{(2)})^{-1} \otimes ...$. Since $\alpha \otimes \alpha^{-1}$ has spread $l$, a QCA $\alpha \otimes \beta$ on $\opMat(X\times \ZZ_{\ge 0}) \otimes \tilde{\ql{A}}$ is related to $\alpha \otimes \Id$ and $\Id \otimes \Id$ by a quantum circuit of spread $l+1$ (see Fig. \ref{fig:qcahalfspaceswindle}). Since $\tilde{\ql{A}} \cong \opMat(X\times \ZZ_{\ge 0})$ by a zero-spread isomorphism, $\alpha \otimes \Id$ is a quantum circuit.
\end{proof}

\begin{figure}
    \centering
    \includegraphics[width=10cm]{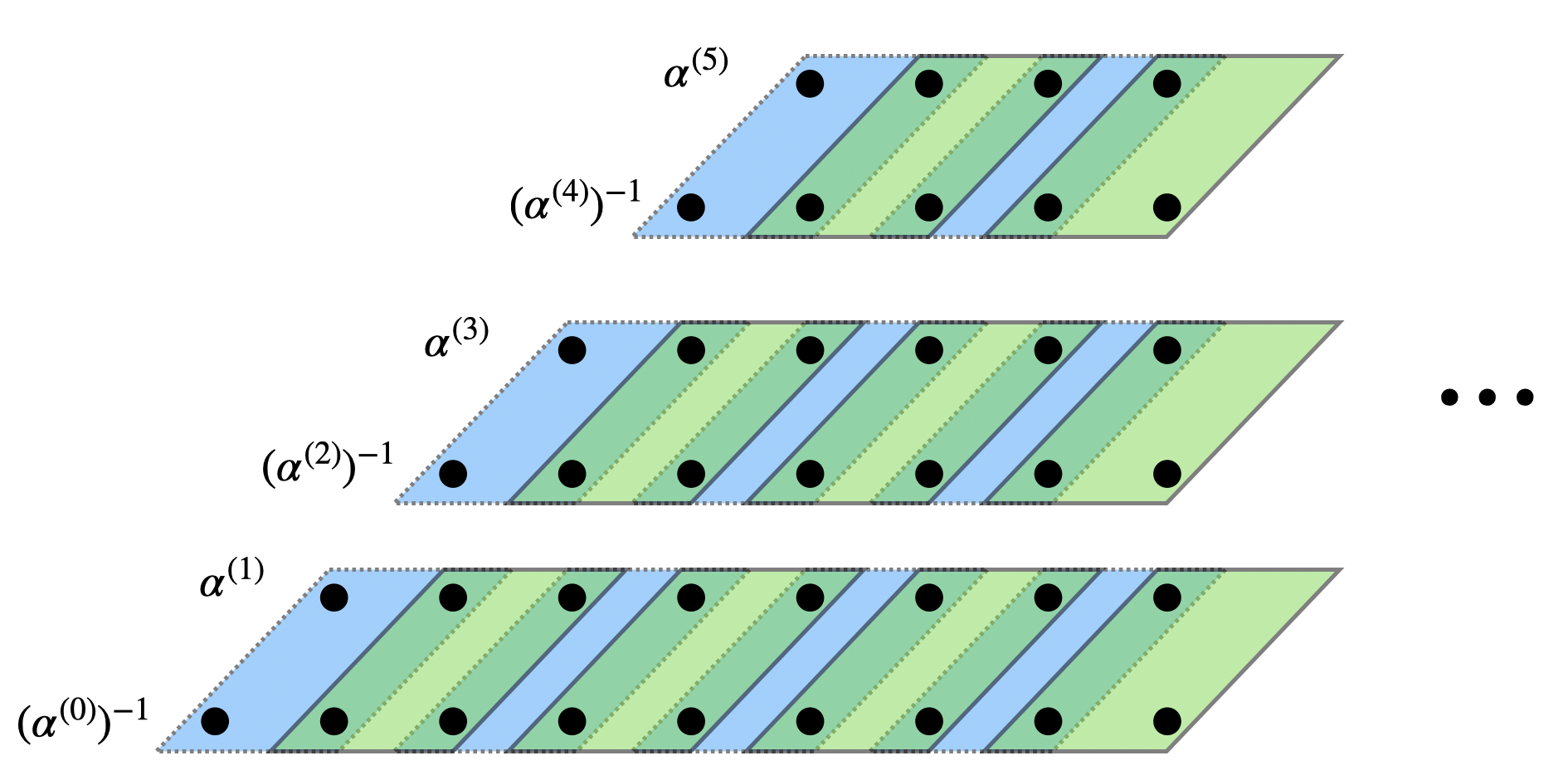}
    \caption{A depiction of the quantum circuit $\beta$ constructed in the proof of Proposition \ref{prop:QCAs_on_a_half-space_are_QC}, which shows that all QCA acting only in a half-space are quantum circuits after stabilization. One can think of this construction as first forming an infinite tower of rows, with the vertical direction denoting the ancillas, and then applying the circuit $C$ which on each pair $(2n,2n+1)$ produces $\alpha^{-1} \otimes \alpha$. Then applying a "shear" transformation we obtain a circuit on the wedge above. Projection to the horizontal axis (the $\ZZ$ factor of $X \times \ZZ$) groups a finite amount of points in the vertical direction for each point of $\ZZ$. Thus, we can regard the resulting circuit as defined in a product algebra over $X \times \ZZ_{\ge 0}$.}
    \label{fig:qcahalfspaceswindle}
\end{figure}

There is a useful criterion for two QCAs to be in the same phase for uniformly locally finite metric spaces of the form $X \times \ZZ$.

\begin{definition}\label{definblendequivalence}
Given two QCAs $\alpha,\beta$ on $\opMat(X\times \ZZ)$, a {\it blend from $\alpha$ to $\beta$} is another QCA $\gamma$ on $\opMat(X\times \ZZ)$ such that there is an interval $I = [a,b] \cap \ZZ$ so that

\medskip

    \begin{enumerate}
        \item $\gamma(a)=\alpha(a)$ for any $a \in \opMat(X\times \ZZ_{<a})$,

\medskip
        
        \item $\gamma(a)=\beta(a)$ for any $a \in \opMat(X\times \ZZ_{>b})$.
    \end{enumerate}

    \medskip

\end{definition}

\begin{prop}\label{prop:circuitandblendequiv}
    Let $\alpha$, $\beta$ be two QCAs on $\opMat(X \times \ZZ)$. There exists a blend from $\alpha$ to $\beta$ if and only if $\alpha$ and $\beta$ are in the same phase.
\end{prop}

\begin{proof}
    Note that $\alpha$ and $\beta$ are in the same phase if and only if $\alpha \beta^{-1}$ is in a trivial phase, and existence of a blend from $\alpha$ to $\beta$ is equivalent to existence of a blend from $\alpha \beta^{-1}$ to the identity. Therefore, without loss of generality, we can assume $\beta = \Id$.

    Suppose $\alpha$ is in a trivial phase. Then the stabilization $\alpha\otimes \text{Id}$ of $\alpha$ is a quantum circuit automorphism of some spread $l$. Let $\gamma$ be a quantum circuit obtained by replacing all the gates with the support not strictly inside $X \times \ZZ_{< 0}$ by the identity element. This quantum circuit defines a blend from $\alpha\otimes \text{Id}_{\opMat(X \times \ZZ)}$ on the left to $\text{Id}_{\opMat(X \times \ZZ)} \otimes \text{Id}_{\opMat(X \times \ZZ)}$ on the right. Note that for some $n \in \NN$, $\gamma$ acts as $\alpha \otimes \Id$ on $X \times \ZZ_{<-n}$ and as $\Id \otimes \Id$ on $X \times \ZZ_{> n}$. Therefore, by choosing a zero-spread isomorphism $\kappa: \opMat(X \times I) \cong \opMat(X \times I) \otimes \opMat(X \times I)$ for $I = [-n,n] \cap \ZZ$, we naturally get a blend from $\alpha$ to $\Id$ on $\opMat(X \times \ZZ)$.

    For the other direction, suppose $\alpha$ has spread $l$. Let $\gamma_{L}$ be a blend from $\alpha$ to the identity automorphism, which acts as $\alpha$ on $X\times \ZZ_{< -n} $ and as the identity on $X\times \ZZ_{> n}$ for some $n \in \NN$. The automorphism $\gamma_R = \alpha \gamma_L^{-1}$ acts as the identity on $X\times \ZZ_{< -(n+l)} $ and as $\alpha$ on $X\times \ZZ_{> (l+n)}$. By Proposition \ref{prop:QCAs_on_a_half-space_are_QC}, $\gamma_L \otimes \Id$ and $\gamma_R \otimes \Id$ are quantum circuit automorphisms. Hence, $\alpha \otimes \Id$ is a quantum circuit automorphism.
\end{proof}

\subsection{Invertible Quasi-Local Algebras and QCAs}\label{subsecboundaryalgebras}

The goal of this section is to show that there is a canonical isomorphism between the group of phases of QCAs over uniformly locally finite metric spaces of the form $X \times \ZZ$ and the Brauer group of invertible quasi-local algebras over $X$, generalizing the results of \cite{Haah2023InvertibleSubalgebras} beyond the context of nets of finite-dimensional $*$-algebras. This result, together the construction of Brauer non-trivial invertible quasi-local algebras over $\ZZ$ defined in Sections~\ref{sec:invertiblefreefermions} and Section~\ref{sec:bosonicholomorphic}, provides examples of non-trivial QCAs over $\ZZ^2$.

Given a QCA on $\opMat(X \times \ZZ)$ for a uniformly locally finite metric space $X$, there is a natural way to associate a quasi-local algebra.

\begin{definition}  \label{def:bdryalgQCA}
Let $X$ be a uniformly locally finite metric space, and let $\alpha$ be a QCA on $\opMat(X \times \ZZ)$. For $l \in \NN$ and $a,b,n \in \ZZ$ such that $\alpha$ has spread $l$ and $a < n-3l$, $n+3l < b$, the $(n,l,a,b)$-left boundary algebra associated with $\alpha$ is the quasi-local algebra $\ql{B}^{(n,l,a,b)}_L$ over $X$ defined via
\beq
(\ql{B}^{(n,l,a,b)}_L)_U :=  \alpha(\opMat(U^{+l} \times [a-l,n-1])) \cap \opMat(U \times [a,b]) ,
\eeq
We say that $\ql{B}_L$ is a left boundary algebra associated with $\alpha$ if it is zero-spread isomorphic to $\ql{B}^{(n,l,a,b)}_L$ for some $n,l,a,b$. Similarly, the $(n,l,a,b)$-right boundary algebra associated with $\alpha$ is the quasi-local algebra $\ql{B}^{(n,l,a,b)}_R$ defined via
\beq
(\ql{B}^{(n,l,a,b)}_R)_U := \alpha(\opMat(U^{+l} \times [n,b+l])) \cap \opMat(U \times [a,b]).
\eeq
We say that $\ql{B}_R$ is a right boundary algebra associated with $\alpha$ if it is zero-spread isomorphic to $\ql{B}^{(n,l,a,b)}_R$ for some $n,l,a,b$.
\end{definition}

\begin{figure}
\centering
\includegraphics[width=10cm]{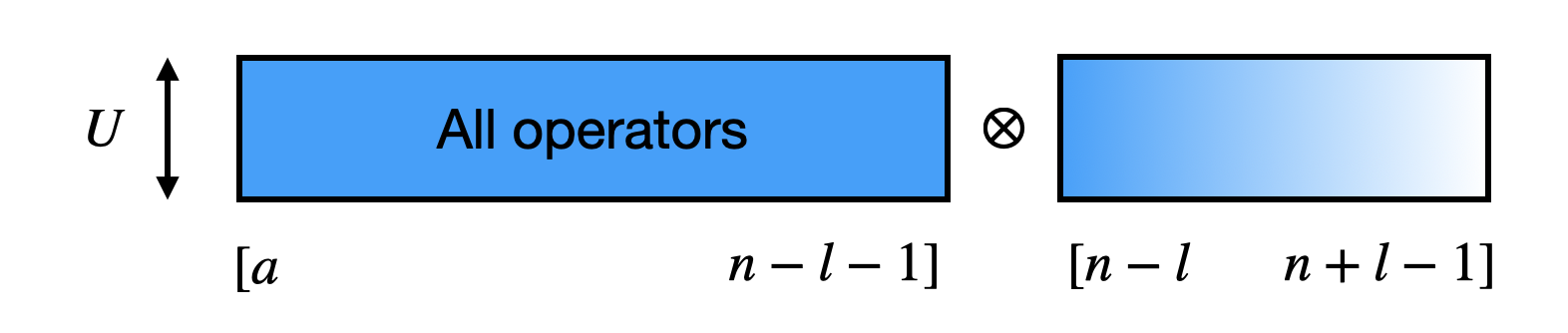}
\caption{A picture of a local piece of the left boundary algebra $\ql{B}^{(n,l,a,b)}_L$ for a bounded subset $U \subset X$. The algebra $\ql{B}^{(n,l,a,b)}_L$ is supported in $U \times [a,n+l-1]$ and contains all observables supported in $U \times [a,n-l-1]$. By Lemma \ref{lma:BLRdecomp}, $\ql{B}^{(n,l,a,b)}_L$ is a tensor product of this algebra and its subalgebra supported in $U\times [n-l,n+l-1]$. With this definition, the width of the right (latter) piece depends on $l$, while the width of the left (former) piece may be chosen by us and acts as a buffer region for proofs that we choose to be $a < n-3l$.}
\label{fig:boundary-algebra-structure}
\end{figure}

\begin{lemma}   \label{lma:BLRdecomp}
For boundary algebras $\ql{B}^{(n,l,a,b)}_{L,R}$ associated with a given QCA $\alpha$ on $\opMat(X \times \ZZ)$, we have
\beq    \label{eq:BLdecomp}
(\ql{B}^{(n,l,a,b)}_L)_U = \opMat(U\times [a,n-l-1]) \otimes (\ql{C}^{(n,l)}_L)_U,
\eeq
\beq
(\ql{B}^{(n,l,a,b)}_R)_U = \opMat(U\times [n+l,b]) \otimes (\ql{C}^{(n,l)}_R)_U,
\eeq
where $\ql{C}^{(n,l)}_{L,R}$ are quasi-local algebras defined via
\beq
(\ql{C}^{(n,l)}_L)_U := \alpha(\opMat(U^{+l} \times [n-2l,n-1])) \cap \opMat(U \times [n-l,n+l-1]),
\eeq
\beq
(\ql{C}^{(n,l)}_R)_U := \alpha(\opMat(U^{+l} \times [n,n+2l-1])) \cap \opMat(U \times [n-l,n+l-1]).
\eeq
\end{lemma}
\begin{proof}
Since $\alpha^{-1}$ has spread $l$, for any bounded $U \subset X$ we have $\opMat(U \times [a,n-l-1]) \subset (\ql{B}^{(n,l,a,b)}_L)_U$. Therefore, by Lemma \ref{lma:gekadison}, we have 
\begin{multline}    
(\ql{B}^{(n,l,a,b)}_L)_U = \opMat(U\times [a,n-l-1]) \otimes \l (\ql{B}^{(n,l,a,b)}_L)_U \cap \opMat(U\times [n-l,b]) \r = \\
= \opMat(U\times [a,n-l-1]) \otimes \l \alpha(\opMat(U^{+l} \times [a-l,n-1])) \cap \opMat(U \times [n-l,b]) \r.
\end{multline}
Any element $x \in \alpha(\opMat(U^{+l} \times [a-l,n-1])) \cap \opMat(U \times [n-l,b])$ (super)commutes with $\opMat(U \times [n+l,b])$ and with $\alpha(\opMat(U \times [a,n-2l-1]))$. Hence, $x \in (\ql{C}^{(n,l)}_L)_U$.

The case of the right boundary algebras $\ql{B}^{(n,l,a,b)}_R$  can be treated similarly.
\end{proof}

The following lemma shows that for QCAs on $\opMat(X \times \ZZ)$ there is essentially unique left/right boundary algebra associated with it.
\begin{lemma} \label{lma:boundaryalgebrasarebruaerequivalent}
The quasi-local algebras $\ql{B}^{(n,l,a,b)}_{L,R}$ associated with a given QCA $\alpha$ on $\opMat(X \times \ZZ)$ for different values of $(n,l,a,b)$ are bounded spread isomorphic.
\end{lemma}
\begin{proof}
The fact that we have zero-spread isomorphisms $\ql{B}^{(n,l,a,b)}_L \cong \ql{B}^{(n,l,a,b+1)}_L$ and $\ql{B}^{(n,l,a,b)}_L \cong \ql{B}^{(n,l,a-1,b)}_L$ follows from Lemma \ref{lma:BLRdecomp}. Since $\opMat(U \times [a,b])$ (super)commutes with $\alpha(a)$ for any observable $a$ in the complement of $U^{+l} \times [a-l,n-1]$ inside $U^{+(l+1)} \times [a-l-1,n-1]$, we also have $\ql{B}^{(n,l,a,b)}_L \cong \ql{B}^{(n,l+1,a,b)}_L$.

To show isomorphism for different values of $n$, let $\ql{D}$ be a quasi-local algebra defined by
\beq
\ql{D}_U := \alpha(\opMat(U^{+l} \times \{n\})) \cap \opMat(U \times [a,b]).
\eeq
The multiplication map
$$
\alpha(\opMat(U^{+l} \times [a-l,n-1] )) \times \alpha(\opMat(U^{+l} \times \{n\})) \to \alpha(\opMat(U^{+l} \times [a-l,n] ))
$$
extends to an isomorphism of von Neumann (super)algebras 
$$
\alpha(\opMat(U^{+l} \times [a-l,n-1] )) \otimes \alpha(\opMat(U^{+l} \times \{n\})) \to \alpha(\opMat(U^{+l} \times [a-l,n] )).
$$
Hence, the multiplication map
$$
(\ql{B}^{(n,l,a,b)}_L)_U \times \ql{D}_U \to (\ql{B}^{(n+1,l,a,b)}_L)_U
$$
extends to an injective bounded spread homomorphism
$$
\mu: \ql{B}^{(n,l,a,b)}_L \otimes \ql{D} \to (\ql{B}^{(n+1,l,a,b)}_L).
$$
To show surjectivity of $\mu$, we first note that by the same reasoning as in the proof of Lemma \ref{lma:BLRdecomp}, we have 
\beq   
(\ql{B}^{(n+1,l,a,b)}_L)_U = \opMat(U\times [a,n-l-1]) \otimes \l (\ql{B}^{(n+1,l,a,b)}_L)_U \cap \opMat(U\times [n-l,n+l]) \r.
\eeq
The first tensor factor is contained in $(\ql{B}^{(n,l,a,b)}_L)_{U^{+2l}}$ by Lemma \ref{lma:BLRdecomp}. Let $x$ be an element in the second tensor factor. Then $\alpha^{-1}(x)$ belongs to $\opMat(U^{+l} \times [n-2l,n])$. Since
\beq
\alpha(\opMat(U^{+l} \times [n-2l,n-1])) \subset (\ql{B}^{(n,l,a,b)}_L)_{U^{+2l}},
\eeq
\beq
\alpha(\opMat(U^{+l} \times \{n \})) \subset \ql{D}_{U^{+2l}},
\eeq
$x$ belongs to $(\ql{B}^{(n,l,a,b)}_L)_{U^{+2l}} \otimes \ql{D}_{U^{+2l}}$. Thus, $\mu$ is surjective.

The case of $\ql{B}^{(n,l,a,b)}_R$ can be treated similarly.

\end{proof}

\begin{prop}\label{prop:boundaryalgebraisinvertible}
Let $X$ be a uniformly locally finite metric space, and $\alpha$ be a QCA on $\opMat(X \times \ZZ)$. Then the associated boundary algebras $\ql{B}_{L,R}$ are invertible, and $[\ql{B}_L] + [\ql{B}_R] = 0$ in $\opBrgr(X)$.
\end{prop}

\begin{proof}
Fix an admissible quadruple $(n,l,a,b)$ for the boundary algebras, and let $\ql{B}_{L,R}$ be $\ql{B}^{(n,l,a,b)}_{L,R}$.The multiplication map
$$
\alpha(\opMat(U^{+l} \times [a-l,n-1] )) \times \alpha(\opMat(U^{+l} \times [n,b+l])) \to \alpha(\opMat(U^{+l} \times [a-l,b+l] ))
$$
for bounded $U \subset X$ extends to an isomorphism of von Neumann (super)algebras 
$$
\alpha(\opMat(U^{+l} \times [a-l,n-1] )) \otimes \alpha(\opMat(U^{+l} \times [n,b+l])) \to \alpha(\opMat(U^{+l} \times [a-l,b+l] )).
$$
Hence, the multiplication map
$$
(\ql{B}_L)_U \times (\ql{B}_R)_U \to \opMat(X \times [a,b])_U
$$
extends to an injective bounded spread homomorphism
$$
\mu: \ql{B}_L \otimes \ql{B}_R \to \opMat(X \times [a,b])
$$
of quasi-local algebras over $X$, where we regard $\opMat(X \times [a,b])$ as a quasi-local algebra over $X$. On the other hand, since $\alpha$, $\alpha^{-1}$ have spread $l$, we have
\begin{multline}
\opMat(U \times [n-l,n+l-1]) \subset \alpha(\opMat(U^{+l} \times [n-2l,n+2l-1])) = \\ =\alpha(\opMat(U^{+l} \times [n-2l,n-1])) \otimes \alpha(\opMat(U^{+l} \times [n,n+2l-1]))
\end{multline}
and
\bqa
\alpha(\opMat(U^{+l} \times [n-2l,n-1])) \subset \opMat(U^{+2l} \times [a,b]), \\
\alpha(\opMat(U^{+l} \times [n,n+2l-1])) \subset \opMat(U^{+2l} \times [a,b]).
\eqa
Hence, we obtain
\beq
\opMat(U \times [n-l,n+l-1]) \subset \mu((\ql{B}_L)_{U^{+2l}} \otimes (\ql{B}_R)_{U^{+2l}}),
\eeq
that together with Lemma \ref{lma:BLRdecomp} implies surjectivity of $\mu$. Thus, $\mu: \ql{B}_L \otimes \ql{B}_R \to \opMat(X \times [a,b])$ is a bounded spread isomorphism.
\end{proof}

\begin{prop}\label{prop:boundaryalgblendinvariant}
Let $X$ be a uniformly locally finite metric space, and let $\alpha$, $\beta$ be QCAs on $\opMat(X \times \ZZ)$ which are in the same phase. Then the associated boundary algebras are Brauer equivalent.
\end{prop}

\begin{proof}
Let $\gamma$ be a blend from $\alpha$ to $\beta$, the existence of which is guaranteed by Proposition \ref{prop:circuitandblendequiv}. Suppose $\alpha$, $\beta$ have spread $l$. Then for large enough $n$, the $(-n,l,-n-3l-1,-n+3l+1)$-left boundary algebras of $\alpha$, $\gamma$ and the $(n,l,n-3l-1,n+3l+1)$-right boundary algebras of $\gamma$, $\beta$ coincide. Hence, by Lemma \ref{lma:boundaryalgebrasarebruaerequivalent}, all boundary algebras associated with $\alpha$ and $\beta$ are Brauer equivalent.
\end{proof}

The proposition above shows that associating a boundary algebra with a QCA yields a well-defined map
\[B_{\text{QCA}\to\text{BAlg}}:\opQCA(X\times \ZZ) \to \opBrgr(X).\]
which is naturally a homomorphism. We show below that this map is an isomorphism.

\begin{definition}\label{def:ingeneralizedshift}
Let $\ql{B}_L$, $\ql{B}_R$ be invertible quasi-local algebras over a uniformly locally finite metric space $X$ with a bounded spread isomorphism $\varphi:\ql{B}_L \otimes \ql{B}_R \to \opMat(X)$. Let $\ql{A}$ be a quasi-local algebra over $X \times \ZZ$ formally defined by an infinite tensor product $\bigotimes_{n \in \ZZ} \l \ql{B}_L^{(n)} \otimes \ql{B}_R^{(n)} \r$, $\ql{B}^{(n)}_L \cong \ql{B}_{L}$, $\ql{B}^{(n)}_R \cong \ql{B}_{R}$, and let $\tau$ be an automorphism of $\ql{A}$ defined via isomorphisms $\ql{B}_L^{(n)} \to \ql{B}_L^{(n+1)}$ and identity maps on $\ql{B}_R^{(n)}$. Since $\opMat(X \times \ZZ)$ is also formally an infinite tensor product $\bigotimes_{n \in \ZZ} \opMat(X)$, the isomorphism $\varphi$ naturally induces an isomorphism $\tilde{\varphi}: \ql{A} \to \opMat(X \times \ZZ)$. A {\it generalized shift QCA} associated with $\varphi$, is a QCA $\alpha$ on $\opMat(X \times \ZZ)$ defined by $\alpha = \tilde{\varphi} \tau \tilde{\varphi}^{-1}$.
\end{definition}

\begin{lemma}\label{lma:generalizedshiftblend}
Let $X$ be a uniformly locally finite metric space, $\alpha$ be a QCA on $\opMat(X \times \ZZ)$, and $\ql{B}_L$, $\ql{B}_R$ be the associated with $\alpha$ boundary algebras. Let $\varphi:\ql{B}_L \otimes \ql{B}_R \to \opMat(X)$ be a bounded spread isomorphism. Then there exists a blend from $\alpha$ to a generalized shift QCA associated with $\varphi$.
\end{lemma}

\begin{proof}
Let $\sigma$ be a "shift" automorphism along $\ZZ$ that maps $(x,n) \in X \times \ZZ$ to $(x,n+1)$. For $l \in \NN$ and $a,b,n \in \ZZ$ such that $\alpha$ has spread $l$ and $a < n-3l$, $n+3l < b$, the automorphism $\alpha \sigma^{l}$ induces a bounded spread isomorphism $\opMat(X \times \ZZ_{<a}) \to \opMat(X \times \ZZ_{<a}) \otimes \ql{B}^{(0,l,a,b)}_L$ for the left boundary algebra $\ql{B}^{(0,l,a,b)}_L$ associated with $\alpha$. Such an isomorphism naturally induces a blend from $\alpha \sigma^l$ to a generalized shift QCA, and therefore there exists a blend $\gamma$ from $\alpha$.
\end{proof}

\begin{theorem}\label{theorembrauerblendiso}
The map $B_{\text{QCA}\to\text{BAlg}}:\opQCA(X\times \mathbb{Z}) \to \opBrgr(X)$ is an isomorphism.
\end{theorem}

\begin{proof}
Let $\ql{B}_L$, $\ql{B}_R$ be invertible quasi-local algebras over a uniformly locally finite metric space $X$ with a bounded spread isomorphism $\varphi:\ql{B}_L \otimes \ql{B}_R \to \opMat(X)$. The generalized shift QCA associated with $\varphi$ has a left boundary algebra which is bounded spread isomorphic to $\ql{B}_L$ by construction. It follows that the map $B_{\text{QCA}\to\text{BAlg}}$ is surjective.

Let $\alpha$ be a QCA on $\opMat(X \times \ZZ)$ such that its left boundary algebra is Brauer trivial. Since a generalized shift QCA associated with such a quasi-local algebra is in a trivial phase, by Lemma \ref{lma:generalizedshiftblend} and Proposition \ref{prop:circuitandblendequiv}, $\alpha$ is in a trivial phase as well. Thus, the map $B_{\text{QCA}\to\text{BAlg}}$ is injective.
\end{proof}

\subsection{Classification of QCAs over $\ZZ$}\label{subsecqca1d}

As a corollary of Theorem \ref{theorembrauerblendiso} and Proposition \ref{prop:0dBrauerGroup}, we obtain the classification of 1d QCAs:
\begin{prop} \label{prop:classificationof1dQCAs}
We have $\sQCA(\ZZ) = \ZZ/2$ and $\QCA(\ZZ) = 0$.   
\end{prop} 
 
\begin{remark}
Note that the classification of QCAs on product algebras with only finite-dimensional type $\rm{I}$ factors is very different. In particular, it was shown in \cite{Gross_2012} that a shift automorphism represents a 1d QCA that is non-trivial in this setting. In contrast, in our stabilized setting, a shift automorphism of $\opMat(\ZZ)$ can be represented by a two layered quantum circuit in the following way. Let $\hilb{H}$ be a separable Hilbert space of infinite dimension (or a superspace of superdimension $\infty|\infty$). Let $w: \hilb{H} \to \hilb{H} \otimes \hilb{H}$ be a (graded) isomorphism, $v \in U(\hilb{H}^{\otimes 3})$ be a unitary defined by $v(\xi_1 \otimes \xi_2 \otimes \xi_3) = w(\xi_1) \otimes w^{-1}(\xi_2 \otimes \xi_3)$, and let $\{ u_n \in \opMat(\ZZ) \}_{n \in \ZZ}$ be unitary elements that act as $v$ on $\opMat(\ZZ)_{\{n,n+1,n+2\}}$. Then a quantum circuit
$$
\alpha = \l \prod_{k \in \ZZ} \Ad_{u_{3k+2}} \r \l \prod_{k \in \ZZ} \Ad_{u_{3k}} \r
$$
defines a shift automorphism.

\end{remark}

\section{Invertible states} \label{sec:InvertibleStates}

In this section, we introduce invertible states and their equivalence classes called invertible phases. As in the previous section, we treat both fermionic and bosonic cases in parallel. 

\begin{definition}
Let $X$ be a uniformly locally finite metric space. We call a state $\omega$ on $\opMat(X)$ {\it unentangled} if it is (graded) locally normal pure and for any $U,V \subset X$ such that $U \cap V = \emptyset$ and any $a \in \opMat(X)_U$, $b \in \opMat(X)_V$, we have $\omega(ab) = \omega(a) \omega(b)$.
\end{definition}

\begin{definition}
Let $X$ be a uniformly locally finite metric space. We call a state $\omega$ on $\opMat(X)$ {\it invertible} if it is (graded) pure and there exists a state $\tilde{\omega}$ on $\opMat(X)$ and a quantum circuit $\alpha$ on $\opMat(X) \otimes \opMat(X) \cong \opMat(X)$ such that $(\omega \otimes \tilde{\omega})\alpha$ is unentangled. In this case, we call $\tilde{\omega}$ {\it an inverse} of $\omega$.
\end{definition}

\begin{remark}  \label{rmk:InvStateIsLocNormal}
Since, by definition, any unentangled state is locally normal, any invertible state is also locally normal. 
\end{remark}

\begin{remark} \label{rmk:QCAtoUS}
An application of a QCA $\alpha$ to an invertible state produces an invertible state. Indeed, if $\alpha$ is a QCA on $\opMat(X)$ for a uniformly locally finite metric space $X$, then $\alpha \otimes \alpha^{-1}$ is a quantum circuit automorphism on $\opMat(X) \otimes \opMat(X)$. Therefore, if $\omega$ is an invertible state on $\opMat(X)$ with an inverse $\tilde{\omega}$, then $\omega \alpha$ is an invertible state with an inverse $\tilde{\omega} \alpha^{-1}$.
\end{remark}

\begin{definition}\label{def:InvPh}
Let $X$ be a uniformly locally finite metric space. We say that states $\omega_1$, $\omega_2$ are {\it in the same phase} if for an unentangled state $\omega_0$ on $\opMat(X)$ the states $\omega_1 \otimes \omega_0$ and $\omega_2 \otimes \omega_0$ are related by a quantum circuit. This is an equivalence relation, and the equivalence classes are called {\it phases}. The set of phases of invertible states is denoted $\opIP(X)$. A tensor product of states and a (zero-spread) isomorphism $\opMat(X) \otimes \opMat(X) \cong \opMat(X)$ equips $\opIP(X)$ with the structure of an Abelian group. The phase of unentangled states corresponds to the unit element of $\opIP(X)$ and is called {\it trivial}. 
\end{definition}

\begin{definition}
A {\it stabilization} of a state $\omega$ is a state on $\opMat(X)$ of the form $(\omega \otimes \omega_0)\alpha$ for an unentangled state $\omega_0$ on $\opMat(X)$ and a (zero-spread) isomorphism $\alpha: \opMat(X) \to \opMat(X) \otimes \opMat(X)$. A state $\omega$ is called {\it stable} if it is related to its stabilization by a zero-spread isomorphism.
\end{definition}
\begin{remark}  \label{rmk:stablestates}
Every invertible state is in the same phase as its stabilization. For stable states, being in the same phase is equivalent to being related by a quantum circuit automorphism. This makes it convenient to work with stable states and quantum circuits when studying $\opIP$.
\end{remark}

\begin{prop} \label{prop:XxZg0istrivial}
For a uniformly locally finite metric space $X$, we have $\opIP(X \times \ZZ_{\geq 0}) = 0$.    
\end{prop}
\begin{proof}
Let $\omega$ be an invertible state on $\ql{A} = \opMat(X \times \ZZ_{\geq 0})$. Choose a state $\tilde{\omega}$ on $\ql{A}$ and a quantum circuit $\alpha$ on $\ql{A} \otimes \ql{A}$ such that $(\omega \otimes \tilde{\omega}) \alpha = \omega_0 \otimes \omega_0$ is unentangled. Assume $\alpha$ has spread $l$.

Let $\ql{B}$ be a quasi-local algebra on $X \times \ZZ_{\geq 0 }$ that is defined by a formal composition of quasi-local algebras $\{ \opMat(X \times \ZZ_{\geq n}) \}_{n \in \ZZ_{\geq 0}}$. Since any bounded subset $U \subset X \times \ZZ_{\geq 0}$ overlaps with $\{X \times \ZZ_{\geq n}\}_{n \in \ZZ_{\geq 0}}$ only for finitely many $n$ and $B(\hat{\hilb{H}})^{\otimes k} \cong B(\hat{\hilb{H}})$ for any $k \in \NN$ and any separable $\hat{\hilb{H}}$ of superdimension $\infty|\infty$, we have an isomorphism $\sigma: \ql{B} \to \ql{A}$. 

Let $\omega^{(n)}$, $\tilde{\omega}^{(n)}, \omega^{(n)}_0$ be states on $\opMat(X \times \ZZ_{\geq n})$ naturally induced via shifts by $\omega$, $\tilde{\omega}$, $\omega_0$, respectively, and let $\psi_0$ and $\psi$ be states on $\ql{B}$ that are formally defined by infinite tensor products $\omega_0^{(0)} \otimes \omega^{(1)}_0 \otimes \omega^{(2)}_0 \otimes ...$ and $\tilde{\omega}^{(0)} \otimes \omega^{(1)} \otimes \tilde{\omega}^{(2)} \otimes ...$, respectively. Since $\alpha$ has spread $l$, a state $\omega \otimes \psi$ on $\ql{A} \otimes \ql{B}$ can be transformed into states $\omega \otimes \psi_0$ and $\omega_0 \otimes \psi_0$ by a quantum circuit of spread $l+1$ (see Fig. \ref{fig:halfspaceswindle}). Since $\ql{B} \cong \ql{A}$, the states $\omega \otimes \omega_0$ and $\omega_0 \otimes \omega_0$ are related by a quantum circuit. Thus, any invertible state $\omega$ on $X \times \ZZ_{\geq 0}$ is in a trivial phase.
\end{proof}

\begin{remark}
    This argument would not work for $X\times\ZZ$ since it relies on the shift between layers so that we are only ever considering finite tensor products of Hilbert spaces. For $X \times \ZZ_{\ge 0}$, it also works in a finite-dimensional context, although the dimensions of the local Hilbert spaces are unbounded.
\end{remark}

\begin{figure}
    \centering
    \includegraphics[width=10cm]{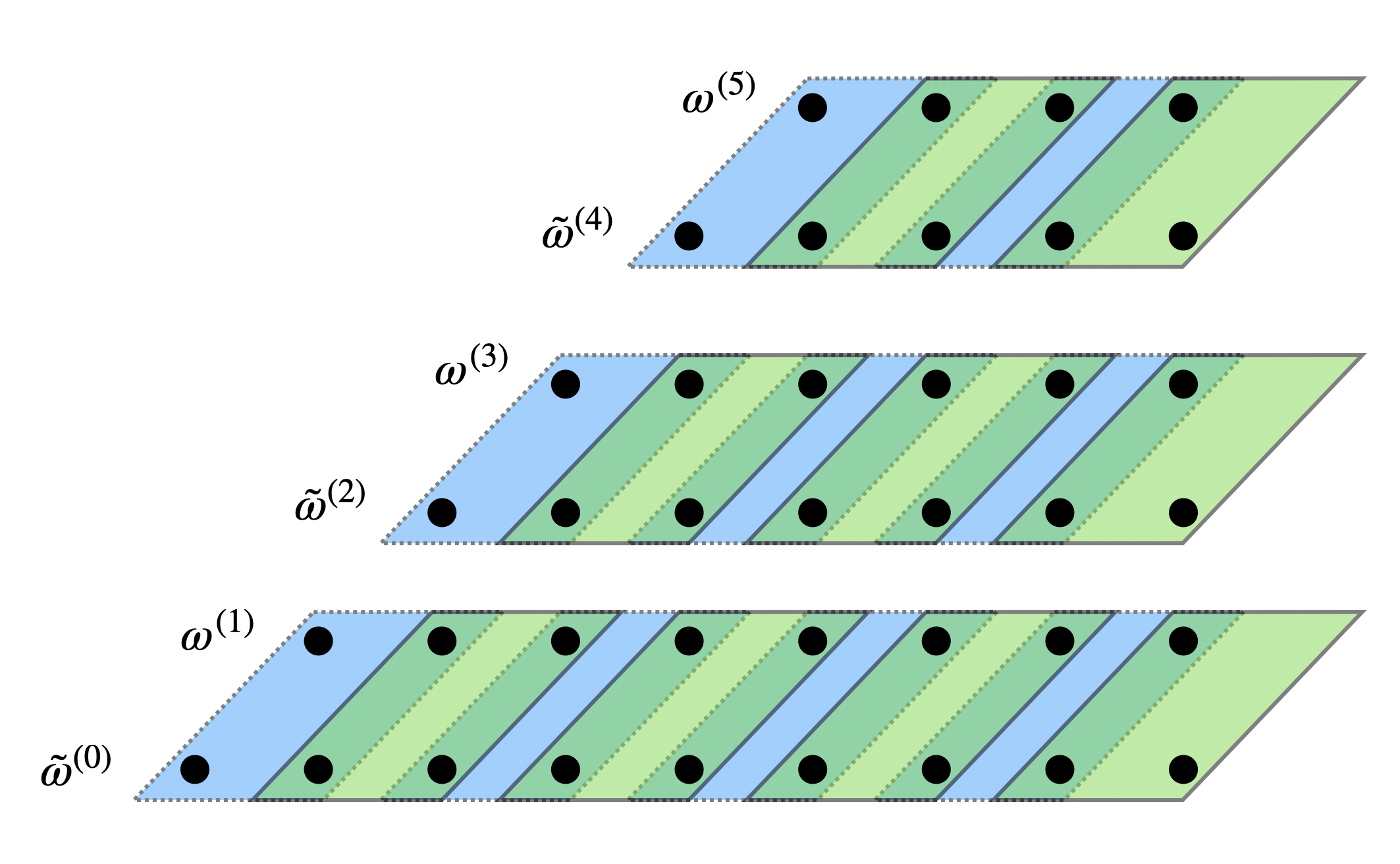}
    \caption{This figure depicts the construction in Proposition \ref{prop:XxZg0istrivial}, where the state $\psi = {\tilde{\omega}}^{(0)} \otimes \omega^{(1)} \otimes \cdots$ is created from $\psi_0 = \omega_0^{(0)} \otimes \omega_0^{(1)} \otimes \cdots$ by a circuit. Here we have drawn the circuit mapping $\omega_0 \otimes \omega_0$ to $\omega \otimes \tilde{\omega}$ as a depth-two circuit consisting of blue and green circuit elements, drawing only the $\ZZ_{\ge 0}$ coordinate. This picture is then compressed onto the $\ZZ_{\ge 0}$ axis to give a circuit on $\ZZ_{\ge 0}$. This results in a circuit, since on each finite $U \subset X \times \ZZ_{\ge 0}$, we only take a finite tensor product of sites and unitaries. The picture shows that the spread is bounded by $l+1$ if the spread of the original circuit is $l$.}
    \label{fig:halfspaceswindle}
\end{figure}

\begin{lemma}   \label{lma:stateonahalfplane}
Let $\omega_1$, $\omega_2$ be invertible states on $\ql{A} = \opMat(X\times\ZZ)$, such that their restrictions to $\opMat(X \times \ZZ_{< n})$, for some $n \in \ZZ$ coincide. Then they are in the same phase.
\end{lemma}
\begin{proof}
Since we can tensor both states with an inverse $\tilde{\omega}_2$ for $\omega_2$, without loss of generality, we can assume that $\omega_2$ is in a trivial phase. Furthermore, since for a quantum circuit $\alpha$ on $\ql{A}$ the restrictions of the states $\omega_1 \alpha$ and $\omega_2 \alpha$ to $\opMat(X \times \ZZ_{< (n-l)})$ coincide for some $l \geq 0$, without loss of generality, we can assume that $\omega_2$ is unentangled and $n = 0$.

Choose a state $\tilde{\omega}_1$ on $\ql{A}$ and a quantum circuit $\alpha$ on $\ql{A} \otimes \ql{A}$ such that $(\omega_1 \otimes \tilde{\omega}_{1})\alpha$ is unentangled. Suppose $\alpha$ has spread $l$. We can represent $\alpha$ as a composition of quantum circuits $\alpha_-$ and $\alpha_+$ whose unitary gates are all in $\opMat(X \times \ZZ_{<0})$ and $\opMat(X \times \ZZ_{\geq - l})$, respectively. It follows that we can choose $\tilde{\omega}_1$ such that its restriction to $\opMat(X \times \ZZ_{< - l})$ is unentangled, and $\alpha$ such that all its unitary gates are inside $\opMat(X \times \ZZ_{\geq - l})$. Thus, $\omega_1$ defines an invertible state on $\opMat(X \times \ZZ_{\geq -l})$. By Proposition \ref{prop:XxZg0istrivial}, this state is in a trivial phase. It follows that $\omega_1$ itself is in a trivial phase.
\end{proof}

A concrete inverse for a given invertible state is provided by a reflected state:

\begin{prop}    \label{prop:ReflectedState}
Let $\omega$ be an invertible state on $\ql{A} = \opMat(X \times \ZZ)$. Let $\sigma$ be a $*$-automorphism of $\ql{A}$ that maps $\ql{A}_U$ to $\ql{A}_{\tilde{U}}$, where $\tilde{U}$ is the reflection of $U \subset X \times \ZZ$ in a plane $X \times \{x\} \subset X \times \RR$ for $x \in \ZZ+1/2 \subset \RR$, and let $\tilde{\omega} = \omega \sigma$. Then $[\omega] = - [\tilde{\omega}]$ in $\IP(X \times \ZZ)$. 
\end{prop}
\begin{proof}
Without loss of generality we can assume $x = -1/2$.

Let $\omega_0$ be an unentangled state on $\ql{A}$, and let $\rho$ be an automorphism of $\ql{A} \otimes \ql{A} \cong \ql{A}$ that swaps $\ql{A}_{X \times \ZZ_{<0}} \otimes 1$ and $1 \otimes \ql{A}_{X \times \ZZ_{\geq 0}}$ and acts trivially on $\ql{A}_{X \times \ZZ_{\geq 0}} \otimes \ql{A}_{X \times \ZZ_{< 0}}$. The state $\psi = (\omega \otimes \omega_0) \rho$ is invertible and its restriction to $X \times \ZZ_{<0}$ is unentangled. By Lemma \ref{lma:stateonahalfplane}, it is in a trivial phase. 

Let $\alpha$ be a quantum circuit on $\ql{A} \otimes \ql{A}$ such that $\psi \alpha$ is unentangled. Then the restriction of $(\omega \otimes \tilde{\omega}) \alpha$ to $X \times \ZZ_{\geq l}$ is unentangled for some $l>0$. Since $(\omega \otimes \tilde{\omega})$ is invertible, by Lemma \ref{lma:stateonahalfplane}, it is in a trivial phase.
\end{proof}

\subsection{Invertible Quasi-Local Algebras and invertible states}\label{ssec:InvAlgebrasAndStates}

With any invertible state $\omega$ on $\opMat(X\times\ZZ)$, one can associate a natural quasi-local algebra over $X$ in the following way.
\begin{definition} \label{def:BALGfromSTATE}
Let $X$ be a uniformly locally finite metric space. An {\it $n$-th right boundary algebra} associated with a locally normal state $\omega$ on $\ql{A} = \opMat(X\times\ZZ)$ is a quasi-local algebra $\ql{B}^{(n)}_R$ over $X$ defined via $U \to \pi_{\omega}(\ql{A}_{U \times \ZZ_{\geq n}})''$ for bounded subsets $U \subset X$. We let $\ql{B}_R = \ql{B}^{(0)}_R$ and call it a {\it right boundary algebra} associated with $\omega$.  Similarly, an {\it $n$-th left boundary algebra} associated with $\omega$ is a quasi-local algebra $\ql{B}^{(n)}_L$ over $X$ defined via $U \to \pi_{\omega}(\ql{A}_{U \times \ZZ_{< n}})''$. We let $\ql{B}_L = \ql{B}^{(0)}_L$ and call it a {\it left boundary algebra} associated with $\omega$. See Fig. \ref{fig:boundaryalgebraofstate}.
\end{definition}

\begin{remark}
The algebras $\ql{B}^{(n)}_{L,R}$ are stable and therefore are all isomorphic to each other and $\ql{B}_{L,R}$. Moreover, a stabilization of a state $\omega$ has boundary algebras which are isomorphic to the boundary algebras of $\omega$.
\end{remark}

\begin{figure}
    \centering
    \includegraphics[width=8cm]{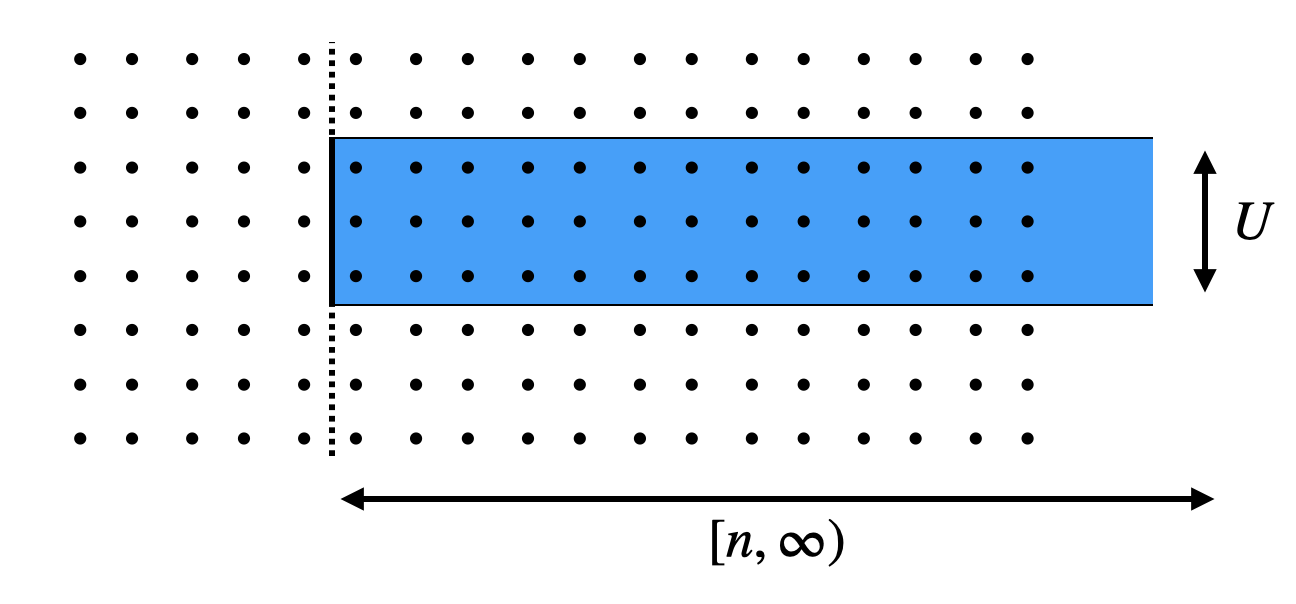}
    \caption{An illustration of Definition \ref{def:BALGfromSTATE} of the boundary algebra of a state $\omega$. For each $U \subset \ZZ^{d-1}$, We take all the quasi-local operators whose support lies in $[n,\infty) \times U$ and complete this algebra to a von Neumann algebra in the GNS representation of $\omega$. This gives a quasi-local algebra on $\ZZ^{d-1}$. In Proposition \ref{prop:SREtoinvBAlg} we show that for an invertible state $\omega$, this defines an invertible quasi-local algebra.}
    \label{fig:boundaryalgebraofstate}
\end{figure}

The following lemma shows that the von Neumann (super)algebras of a boundary algebra are all (super)factors.
\begin{lemma}   \label{lma:RestrictionIsASuperfactor}
Let $X$ be a uniformly locally finite metric space and $\omega$ be a graded, locally normal pure state on $\opMat(X)$. Then for any subset $Y \subset X$, the von Neumann algebra $\pi_{\omega}(\opMat(Y))''$ is a hyperfinite (super)factor. 
\end{lemma}
\begin{proof}
Let $\ql{A} = \opMat(X)$, $\ql{B} = \opMat(Y)$, $\ql{C} = \opMat(Y^c)$. Consider the von Neumann superalgebras $\cstar{P} = \pi_{\omega}(\ql{A})''$, $\cstar{M} = \pi_{\omega}(\ql{B})''$ and $\cstar{N} = \pi_{\omega}(\ql{C})''$. Since $\ql{B}$ and $\ql{C}$ supercommute, clearly $\cstar{N} \subset \cstar{M}'$ and $\cstar{M} \subseteq \cstar{N}'$. We claim $\cstar{P} = \cstar{M} \vee \cstar{N}$. Note that for any bounded $U$, $\ql{A}_{U} = \ql{B}_{U\cap Y} \otimes \ql{C}_{U\cap Y^{c}}$. In particular, $\pi_{\omega}(\ql{A}_{U}) \subseteq \cstar{M} \vee \cstar{N}$, and since $\ql{A}$ is generated by $\{\ql{A}_{U}\}_{U \subset X}$, this implies $\cstar{P} \subseteq \cstar{M} \vee \cstar{N}$. The reverse inclusion follows by definition, so $\cstar{P} = \cstar{M} \vee \cstar{N}$ as claimed.

Now, suppose $x\in Z(\cstar{M}) = \cstar{M} \cap \cstar{M}'$. Since $x \in \cstar{M}$ and $\cstar{M} \subseteq \cstar{N}'$, $x\in \cstar{N}'$. Then $x \in \cstar{M}'\cap \cstar{N}'=(\cstar{M} \vee \cstar{N})' = \cstar{P}'$. Thus, $x \in \cstar{P} \cap \cstar{P}' = Z(\cstar{P})$. But since $\omega$ is graded pure on $\ql{A}$, $Z(\cstar{P}) = Z(B(\hilb{H}_{\omega}))=\CCC 1$. Thus $Z(\cstar{M}) \subseteq \CCC 1$, so $\cstar{M}$ is a superfactor.

\end{proof}

\begin{prop} \label{prop:SREtoinvBAlg}
Let $X$ be a uniformly locally finite metric space. Let $\omega$ be an invertible state on $\ql{A} = \opMat(X\times\ZZ)$, and let $\psi$ be the induced state on $\ql{B}_{L}\otimes \ql{B}_{R}$. Then $\ql{B}_L$, $\ql{B}_R$ are invertible and there exists a bounded spread isomorphism $\beta: \ql{B}_L \otimes \ql{B}_R \to \opMat(X)$ such that $\psi\beta^{-1}$ is in a trivial phase.
\end{prop}
\begin{proof}
By Remark \ref{rmk:stablestates}, without loss of generality, we can assume that $\omega$ is stable. By definition, there exists a state $\tilde{\omega}$ on $\opMat(X)$, an unentangled state $\omega_0$ on $\opMat(X)$, and a quantum circuit $\alpha$ on $\ql{A} \otimes \ql{A}$, such that $\omega \otimes \tilde{\omega} = (\omega_0 \otimes \omega_0) \alpha$. Assume $\alpha$ has spread $l$.

For any $a \in (\ql{B}^{(-2l)}_L)_{U}$ and $b \in (\ql{B}^{(2l)}_R)_{U}$, we have $\lal \Omega_{\omega}, a b \Omega_{\omega} \ral = \lal \Omega_{\omega}, a \Omega_{\omega} \ral \lal \Omega_{\omega}, b \Omega_{\omega} \ral$. 
By Lemma \ref{lma:RestrictionIsASuperfactor} and Proposition \ref{prop:SuperalgebraSplitProperty}, we have isomorphisms $(\ql{B}^{(-2l)}_L)_{U} \vee (\ql{B}^{(2l)}_R)_{U} \cong (\ql{B}^{(-2l)}_L)_{U} \otimes (\ql{B}^{(2l)}_R)_{U}$. Therefore, for quasi-local algebras $\ql{C}^{(n)}$ defined via $U \to \pi_{\omega}(\ql{A}_{U \times (\ZZ_{< -n} \cup \ZZ_{\geq n})})''$, we have $\ql{C}^{(0)} \cong \ql{C}^{(n)} \cong \ql{B}^{(-n)}_L \otimes \ql{B}^{(n)}_R \cong \ql{B}_L \otimes \ql{B}_R$, where we have used that $\ql{C}^{(n)}, \ql{B}^{(n)}_L, \ql{B}^{(n)}_R$ are stable.

Let $\rho$ be an automorphism of $\ql{A} \otimes \ql{A} \cong \ql{A}$ that swaps $\ql{A}_{\ZZ_{<0}} \otimes 1$ and $1 \otimes \ql{A}_{\ZZ_{\geq 0}}$ and acts trivially on $\ql{A}_{\ZZ_{\geq 0}} \otimes \ql{A}_{\ZZ_{< 0}}$ as described in Proposition \ref{prop:ReflectedState}. The state $\varphi = (\omega \otimes \omega_0) \rho$ is invertible and its associated right boundary algebra is $\ql{C} \cong \ql{C}^{(0)}$ (see Figure \ref{fig:balg-inv}). Its restriction to $\opMat(X \times \ZZ_{<0})$ in unentangled, and since $\omega$ is stable, by Proposition \ref{prop:XxZg0istrivial}, there exists a quantum circuit $\beta$ with all gates being in $\opMat(X \times \ZZ_{\geq 0})$ that transforms $\varphi$ to an unentangled state. This circuit induces a bounded spread isomorphism between $\ql{C} \cong \ql{B}_L \otimes \ql{B}_R$ and the boundary algebra of an unentangled state which is isomorphic to $\opMat(X)$. Indeed, since $\cstar{C}$ is stable, it is isomorphic to the right boundary algebra of $\varphi$ at some cut $n < 0$. The circuit $\beta$ acts entirely in $\ZZ_{>n} \times X$ and so defines a bounded-spread isomorphism on this right boundary algebra. 

The quantum circuit $\beta$ induces a quantum circuit on $\ql{C} \cong \opMat(X)$ that transforms the state $\psi \beta^{-1}$ to an unentangled state.   
\end{proof}

\begin{figure}
    \centering
    \includegraphics[width=8cm]{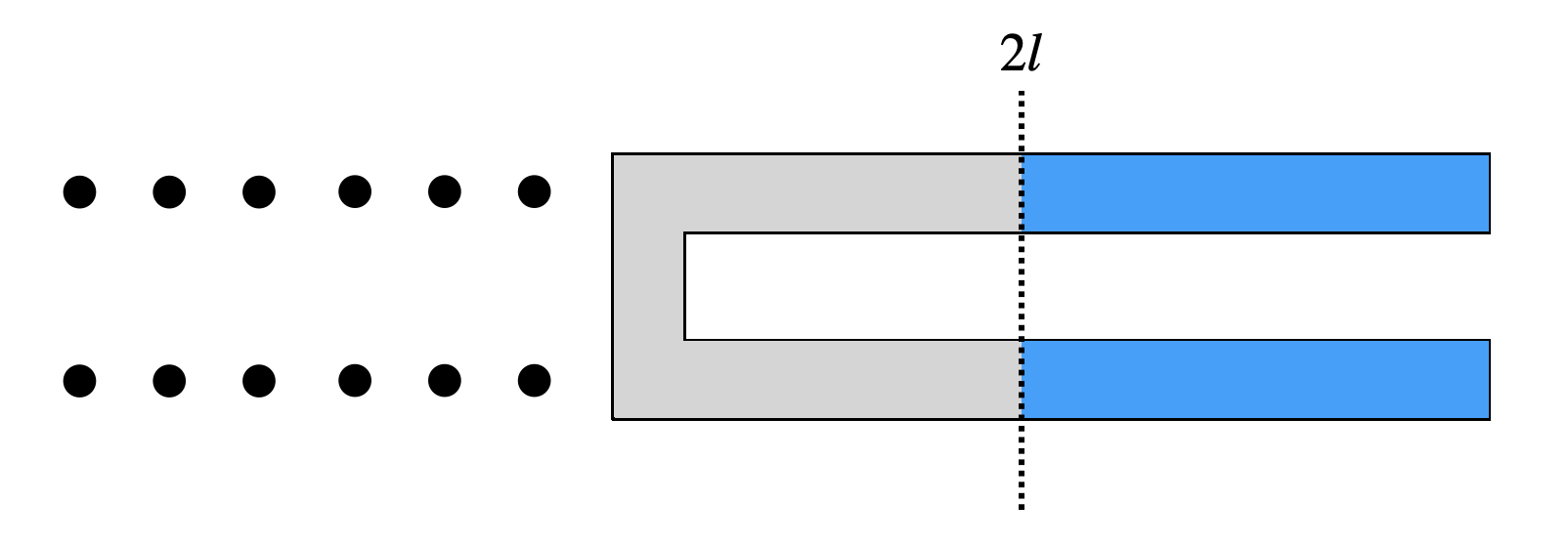}
    \caption{The right boundary algebra (blue) of the folded state is isomorphic to $\cstar{B}_L^{(-2l)}\otimes \cstar{B}_R^{(2l)}$. For an invertible state, this algebra is bounded-spread isomorphic to $\opMat(X)$ since the folded state is disentanglable by Proposition \ref{prop:XxZg0istrivial}, and therefore $\cstar{B}_L$ and $\cstar{B}_R$ are invertible quasi-local algebras, see Proposition \ref{prop:SREtoinvBAlg}. }
    \label{fig:balg-inv}
\end{figure}

\begin{prop} \label{prop:SamePhaseSameBAlg}
Let $\omega, \tilde{\omega}$ be invertible states on $\ql{A} = \opMat(X\times\ZZ)$ in the same phase. Then for the associated boundary algebras $\ql{B}_{L,R}$, $\tilde{\ql{B}}_{L,R}$, we have $[\ql{B}_{R}] = [\tilde{\ql{B}}_{R}]$, $[\ql{B}_{L}] = [\tilde{\ql{B}}_{L}]$ in $\opBrgr(X)$.
\end{prop}
\begin{proof}
By Remark \ref{rmk:stablestates}, without loss of generality, we can assume that $\tilde{\omega} = \omega \alpha$ for a quantum circuit automorphism $\alpha$.

Suppose $\alpha$ has spread $l$. We can represent $\alpha$ as a composition $\alpha_- \alpha_+$ of quantum circuits $\alpha_-$, $\alpha_+$ with all their unitary gates being in $\opMat(X \times \ZZ_{<0})$, $\opMat(X \times \ZZ_{\geq -l})$, respectively. The states $\omega$, $\omega \alpha_-$ have the same restrictions to $\opMat(X \times \ZZ_{\geq 0})$, and therefore have the same right boundary algebras. The states $\omega\alpha_-$, $(\omega \alpha_-)\alpha_+$ have the same restrictions to $\opMat(X \times \ZZ_{< -l})$, and therefore have the same left boundary algebras. Since by Proposition \ref{prop:SREtoinvBAlg} the left and the right boundary algebras are the inverses of each other, we obtain $[\ql{B}_{L,R}] = [\tilde{\ql{B}}_{L,R}]$.
\end{proof}

Thus, associating a boundary algebra with an invertible state provides a homomorphism $$B_{\text{state}\to\text{BAlg}}:\opIP(X\times\ZZ) \to \opBrgr(X).$$

\begin{prop}   \label{prop:SREtoBALGinj}
The map $B_{\text{state}\to\text{BAlg}}:\opIP(X\times\ZZ) \to \opBrgr(X)$ is injective.
\end{prop}

\begin{proof}

Let $\omega$ be a stable invertible state on $\ql{A} = \opMat(X\times\ZZ)$, and $\ql{B}_L$, $\ql{B}_R$ be the associated boundary algebras over $\ZZ^{d}$ with a bounded spread isomorphism $\beta: \ql{B}_L \otimes \ql{B}_R \to \opMat(X)$ from Proposition \ref{prop:SREtoinvBAlg}. We define a quasi-local algebra $\ql{C}$ over $X\times\ZZ$ by a formal infinite tensor product $\bigotimes_{n \in \ZZ} (\ql{B}_R^{(n)} \otimes \ql{B}_L^{(n)})$, $\ql{B}_R^{(n)} \cong \ql{B}_R$, $\ql{B}_L^{(n)} \cong \ql{B}_L$ (see Fig. \ref{fig:CWstate}). The composition of the bounded spread isomorphism $\beta^{-1}$ and $\ql{B}_L \otimes \ql{B}_R \cong \ql{B}_R \otimes \ql{B}_L$ applied to codimension one strips $\{\opMat(X \times \{n\})\}_{n \in \ZZ}$ induces a bounded spread isomorphism $\sigma: \ql{A} \to \ql{C}$. The state $\omega$ naturally induces a state $\psi$ on $\ql{B}_L \otimes \ql{B}_R$ and therefore states $\psi^{(n)}$ on $\ql{B}^{(n-1)}_L \otimes \ql{B}_R^{(n)}$ via isomorphisms $\ql{B}^{(n)}_{L,R} \cong \ql{B}_{L,R}$. Therefore, $\omega$ naturally induces a state $\omega_{\text{cw}}$ on $\ql{A}$ defined by composing the formal infinite tensor product state $(... \otimes \psi^{(n-1)} \otimes \psi^{(n)} \otimes \psi^{(n+1)} \otimes ...)$ with $\sigma$. We call $\omega_{\text{cw}}$ {\it a coupled wires state} associated with $\omega$.

Similarly, we define a quasi-local algebra $\ql{C}_{< 0}$ over $X\times\ZZ$ by a composition of $\opMat(X \times \ZZ_{\geq 0})$ and a formal infinite tensor product $\bigotimes_{n \in \ZZ_{< 0}} (\ql{B}_R^{(n)} \otimes \ql{B}_L^{(n)})$ (see Fig. \ref{fig:BCWstate}). The composition of the bounded spread isomorphism $\beta^{-1}$ and $\ql{B}_L \otimes \ql{B}_R \cong \ql{B}_R \otimes \ql{B}_L$ applied to codimension one strips $\{\opMat(X \times \{n\})\}_{n \in \ZZ_{<0}}$ induces a bounded spread isomorphism $\sigma_{< 0 } : \ql{A} \to \ql{C}_{<0}$. The state $\omega$ naturally induces a state $\phi$ on the quasi-local algebra over $X \times \ZZ_{ \geq -1}$ defined by a composition of $\opMat(X \times \ZZ_{\geq 0})$ and $\ql{B}_L$ on $X \times \{-1\}$. Therefore, given and isomorphism $\ql{B}^{(-1)}_L \cong \ql{B}_L$, $\omega$ naturally induces a state $\omega_{\text{bcw}}$ on $\ql{A}$ defined by composing the formal infinite tensor product state $(... \psi^{(-2)} \otimes \psi^{(-1)} \otimes \phi)$ with $\sigma_{<0}$. We call $\omega_{\text{bcw}}$ {\it a blend to coupled wires state} associated with $\omega$.

Suppose for a moment that $\omega$ has the form $\omega = \omega_0 \alpha$ for an unentangled state $\omega_0$ and a quantum circuit $\alpha$ with the support of all gates being inside $X \times [-l,l]$ for some $l>0$. In this case, there are isomorphisms $\gamma_R:\ql{B}_R \to \opMat(X \times [0,2l])$ and $\gamma_L: \ql{B}_L \to \opMat(X \times [-2l,-1])$ which map $(\ql{B}_R)_U \to \opMat(U \times [0,2l])$, $(\ql{B}_L)_U \to \opMat(U \times [-2l,-1])$ for $U \subset X$ such that the state $\psi(\gamma_L \otimes \gamma_R)^{-1}$ on $\opMat(X \times [-2l,2l])$ coincides with $\omega|_{\opMat(X \times [-2l,2l])}$. These isomorphisms provide zero-spread isomorphisms that transform $\omega_{\text{cw}}$ and $\omega_{\text{bcw}}$ to a tensor product of copies of $\omega|_{\opMat(X \times [-2l,2l])}$. Therefore, both states $\omega_{\text{cw}}$ and $\omega_{\text{bcw}}$ are invertible.

Suppose now that $\omega$ is a general invertible state. Let $\tilde{\omega}$ be its inverse and $\varphi = \omega \otimes \tilde{\omega}$. We have $\varphi_{\text{cw}} = \omega_{\text{cw}} \otimes \tilde{\omega}_{\text{cw}}$ and $\varphi_{\text{bcw}} = \omega_{\text{bcw}} \otimes \tilde{\omega}_{\text{bcw}}$ for the associated coupled wires and blend to coupled wires states, respectively. There are quantum circuits, $\alpha_{-}$ and $\alpha_{+}$ with all gates being supported inside $X \times \ZZ_{<0}$, $X \times \ZZ_{\geq 0}$, respectively, such that $\chi = \varphi \alpha_{-}\alpha_{+}$ has the form of $\omega$ from the previous paragraph. The automorphisms $\alpha_{-}$ and $\alpha_{+}$ induce bounded spread isomorphisms between boundary algebras for $\varphi$ and $\chi$, and therefore induce bounded spread isomorphisms between $\varphi_{\text{cw}}$, $\varphi_{\text{bcw}}$ and $\chi_{\text{cw}}$, $\chi_{\text{bcw}}$, respectively. Since $\chi_{\text{cw}}$, $\chi_{\text{bcw}}$ are invertible, the states $\varphi_{\text{cw}}$, $\varphi_{\text{bcw}}$ are invertible as well. It follows that $\omega_{\text{cw}}$, $\omega_{\text{bcw}}$ are invertible.

To prove injectivity of $B_{\text{state}\to\text{BAlg}}$, it is enough to show that the associated coupled wires state $\omega_{\text{cw}}$ is in the same phase as $\omega$. By construction, the states $\omega_{\text{bcw}}$ and $\omega$ have the same restrictions to $\opMat(X \times \ZZ_{\geq 0})$ and the states $\omega_{\text{bcw}}$ and $\omega_{\text{cw}}$ have the same restrictions to $\opMat(X \times \ZZ_{< 0})$. Since $\omega_{\text{cw}}$ and $\omega_{\text{bcw}}$ are invertible, by Lemma \ref{lma:stateonahalfplane}, the states $\omega$ and $\omega_{\text{cw}}$ are in the same phase.
\end{proof}

\definecolor{figblue}{RGB}{71,159,248}

\begin{figure}[ht]
  \centering

  \begin{subfigure}{\textwidth}
    \centering
    \begin{tikzpicture}[
        scale=0.7, transform shape,
        blackline/.style={black, line width=1.4pt},
        blueline/.style ={figblue, thick},
        shadestyle/.style={fill=figblue, opacity=0.4}
      ]
      \def\W{2.5}\def\Ra{0.75}\def\Lb{1.75}\def\H{5}
      \foreach \n in {-2,...,2}{
        \pgfmathsetmacro{\xa}{(\n+2)*\W + \Lb}
        \pgfmathsetmacro{\xb}{(\n+3)*\W + \Ra}
        \fill[shadestyle] (\xa,0) rectangle (\xb,\H);
      }
      \pgfmathsetmacro{\xLedge}{(-3+3)*\W + \Ra}  
      \fill[shadestyle] (0,0) rectangle (\xLedge,\H);
      \pgfmathsetmacro{\xRedge}{(2+3)*\W + \Lb}   
      \fill[shadestyle] (\xRedge,0) rectangle (6*\W,\H);
      \foreach \i in {0,...,6}{ \draw[blackline] (\i*\W,0) -- (\i*\W,\H); }
      \foreach \n in {-3,...,2}{
        \pgfmathsetmacro{\xR}{(\n+3)*\W + \Ra}
        \pgfmathsetmacro{\xL}{(\n+3)*\W + \Lb}
        \draw[blueline] (\xR,0) -- (\xR,\H);
        \draw[blueline] (\xL,0) -- (\xL,\H);
        \node[below, font=\footnotesize] at (\xR,0) {$\mathcal{B}^{(\n)}_{R}$};
        \node[below, font=\footnotesize] at (\xL,0) {$\mathcal{B}^{(\n)}_{L}$};
      }
    \end{tikzpicture}
    \caption{An illustration of a coupled wires state $\omega_{\text{cw}}$. Black lines split $\opMat(X\times\ZZ)$ into $\opMat(X \times \{n\})$, $n \in \ZZ$. Each $\opMat(X \times \{n\})$ is bounded spread isomorphic to $\ql{B}^{(n)}_R \otimes \ql{B}^{(n)}_L$. The state $\omega_{\text{cw}}$ is a tensor product of pure states on $\ql{B}^{(n-1)}_L \otimes \ql{B}^{(n)}_R$.}
    \label{fig:CWstate}
  \end{subfigure}

  \vspace{1.5em}

  \begin{subfigure}{\textwidth}
    \centering
    \begin{tikzpicture}[
        scale=0.7, transform shape,
        blackline/.style={black, line width=1.4pt},
        blueline/.style ={figblue, thick},
        shadestyle/.style={fill=figblue, opacity=0.4}
      ]
      \def\W{2.5}\def\Ra{0.75}\def\Lb{1.75}\def\H{5}
      \foreach \n in {-2,-1}{
        \pgfmathsetmacro{\xa}{(\n+2)*\W + \Lb}
        \pgfmathsetmacro{\xb}{(\n+3)*\W + \Ra}
        \fill[shadestyle] (\xa,0) rectangle (\xb,\H);
      }
      \pgfmathsetmacro{\xLedge}{(-3+3)*\W + \Ra}   
      \fill[shadestyle] (0,0) rectangle (\xLedge,\H);
      \pgfmathsetmacro{\xstart}{(-1+3)*\W + \Lb}   
      \pgfmathsetmacro{\xend}{6*\W}                
      \fill[shadestyle] (\xstart,0) rectangle (\xend,\H);
      \foreach \i in {0,...,6}{ \draw[blackline] (\i*\W,0) -- (\i*\W,\H); }
      \foreach \n in {-3,-2,-1}{
        \pgfmathsetmacro{\xR}{(\n+3)*\W + \Ra}
        \pgfmathsetmacro{\xL}{(\n+3)*\W + \Lb}
        \draw[blueline] (\xR,0) -- (\xR,\H);
        \draw[blueline] (\xL,0) -- (\xL,\H);
        \node[below, font=\footnotesize] at (\xR,0) {$\mathcal{B}^{(\n)}_{R}$};
        \node[below, font=\footnotesize] at (\xL,0) {$\mathcal{B}^{(\n)}_{L}$};
      }
      \node[below, font=\footnotesize] at (11.2,0) {$\opMat(X \times \ZZ_{\geq 0})$};
    \end{tikzpicture}
    \caption{An illustration of a coupled wires state $\omega_{\text{bcw}}$. Black lines split $\opMat(X\times\ZZ)$ into $\opMat(X \times \{n\})$, $n \in \ZZ$. Each $\opMat(X \times \{n\})$, $n \in \ZZ_{<0}$ is bounded spread isomorphic to $\ql{B}^{(n)}_R \otimes \ql{B}^{(n)}_L$. The state $\omega_{\text{bcw}}$ is a tensor product of pure states on $\ql{B}^{(n-1)}_L \otimes \ql{B}^{(n)}_R$ and a pure state on a $C^*$-algebra $\ql{B}^{(-1)}_L \otimes \opMat(X \times \ZZ_{\geq 0})$.}
    \label{fig:BCWstate}
  \end{subfigure}

\end{figure}

\begin{prop}   \label{prop:SREtoBALGsurj}
The map $B_{\text{state}\to\text{BAlg}}:\opIP(X\times\ZZ) \to \opBrgr(X)$ is surjective.
\end{prop}

\begin{proof}
Let $\ql{B}_R$, $\ql{B}_L$ be invertible quasi-local algebras over $X$ with a bounded spread isomorphism $\beta:\ql{B}_R \otimes \ql{B}_L \to \opMat(X)$, and let $\omega_0$ be an unentangled state on $\ql{A} = \opMat(X\times\ZZ)$. 

Let $\ql{C}$ be a quasi-local algebra over $X\times\ZZ$ defined by a formal infinite tensor product $\bigotimes_{n \in \ZZ} (\ql{B}_L^{(n)} \otimes \ql{B}_R^{(n)})$, where $\ql{B}_{L,R}^{(n)} \cong \ql{B}_{L,R}$. The bounded spread isomorphism $\beta$ induces a bounded spread isomorphism $\sigma: \ql{A} \to \ql{C}$. Let $\tau$ be a $*$-automorphism of $\ql{C}$ that shifts $\ql{B}_L^{(n)} \to \ql{B}_L^{(n+1)}$ while acting trivially on $\ql{B}^{(n)}_{R}$ for all $n \in \ZZ$. Then $\alpha = \sigma^{-1} \tau \sigma$ defines a QCA on $\ql{A}$. By Remark \ref{rmk:QCAtoUS}, the state $\chi = \omega_0 \alpha$ is invertible.

Let $\tilde{\ql{B}}_L$, $\tilde{\ql{B}}_R$ be boundary algebras associated with $\chi$. By construction, the automorphism $\alpha$ induces a bounded spread isomorphism between $\tilde{\ql{B}}_R$ and $\ql{B}_R \otimes \ql{B}_{0,R} \cong \ql{B}_R$, where $\ql{B}_{0,R} \cong \opMat(X)$ is the right boundary algebra associated with the state $\omega_0$. It follows that for any invertible quasi-local algebra $\ql{B}_R$, there is an invertible state $\omega$, whose right boundary algebra is Brauer equivalent to $\ql{B}_R$.
\end{proof}

Proposition \ref{prop:SREtoBALGinj} and Proposition \ref{prop:SREtoBALGsurj} together imply

\begin{theorem}
\label{thm:BAlgState}
The map $B_{\text{state}\to\text{BAlg}}:\opIP(X\times\ZZ) \to \opBrgr(X)$ is an isomorphism.
\end{theorem}

This theorem together with Proposition \ref{prop:ReflectedState} immediately provides a concrete inverse for an invertible quasi-local algebra over spaces of the form $X \times \ZZ$:
\begin{corollary}
Let $X$ be a uniformly locally finite metric space, and let $\ql{A}$ be an invertible quasi-local algebra over $X \times \ZZ$. Let $\tilde{\ql{A}}$ be the quasi-local algebra over $X \times \ZZ$ defined by $\tilde{\ql{A}}_U = \ql{A}_{\tilde{U}}$, where $\tilde{U}$ is a reflection of $U$ in the hyperplane $X \times \{1/2\}$. Then $\tilde{\ql{A}}$ is an inverse for $\ql{A}$.
\end{corollary}

\begin{remark}
Theorem \ref{thm:BAlgState} together with an isomorphism $\text{BAlg}:\opQCA(X\times\ZZ) \to \opBrgr(X)$ from Theorem \ref{theorembrauerblendiso} implies that any invertible state can be obtained by applying a QCA to an unentangled state. Further, it implies that if a QCA preserves an unentangled state then it must be a quantum circuit automorphism.
\end{remark}

\begin{remark}  \label{rmk:HyperfiniteRepresentative}
The isomorphism from Theorem \ref{thm:BAlgState} also implies that any class of $\opBrgr(X)$ admits a representative $\ql{A}$, such that $\ql{A}_U$ is a hyperfinite (super)factor for any bounded subset $U \subset X$. Indeed, we can choose an invertible state over $X \times \ZZ$ in the phase corresponding to this Brauer class and take the associated right boundary algebra, which satisfies the desired property by Lemma \ref{lma:RestrictionIsASuperfactor}.
\end{remark}

\section{A spectrum of invertible phases}   \label{sec:KitaevSpectrum}

It was proposed by A. Kitaev that with any dimension $d$, one can associate a classifying space of invertible quantum many-body systems whose homotopy groups classify invertible phases in dimension $d$ and lower. Physically, this space describes the subspace of parameters of a microscopic Hamiltonian for which the system is an invertible topological phase. While the details of this subspace depend on a model, it was suggested that the homotopy type of this space is universal. The homotopy type contains information not only about the abelian groups of invertible phases, but also describes the junctions between topological defects of various codimension (encoded in $k$-invariants). Furthermore, this homotopy type is naturally equipped with the structure of an infinite loop space (or a connective $\Omega$-spectrum) due to the composition operation on quantum-many body systems. 

The anticipated infinite-loop spaces, which we denote by $\Kit^{(d)}$ for bosonic invertible phase and $\sKit^{(d)}$ for fermionic ones, should satisfy
\beq
\pi_{k} (\Kit^{(d)}) = \IP(\ZZ^{d-k}),\,\, \pi_{k} (\sKit^{(d)}) = \sIP(\ZZ^{d-k})
\eeq
\beq
\Kit^{(d)} \cong \Omega\, \Kit^{(d+1)}, \,\, \sKit^{(d)} \cong \Omega\, \sKit^{(d+1)}.
\eeq
Altogether, they can be organized into an $\Omega$-spectrum 
\beq
\Kit := \colim_{d \to \infty} \Sigma^{-d} \, \Kit^{(d)}. 
\eeq
Furthermore, since in dimension zero invertible quantum systems are supposed to be described by lines in a Hilbert (super)space, it is expected that at the level of spaces we have 
\beq
\Kit^{(0)} \cong \Omega^{d}\,\Kit^{(d)} \cong \mathbb{CP}^{\infty},\,\, \sKit^{(0)} \cong \Omega^{d}\,\sKit^{(d)} \cong \mathbb{CP}^{\infty} \times \ZZ/2.
\eeq
The extra $\ZZ/2$ factor in the fermionic case corresponds to the value of the fermionic parity.

An analytic model for Kitaev spectrum has been recently proposed in \cite{kubota2025stable}. An analogous conjecture was formulated for QCAs \cite{tu2026anomalies, czajka2025anomalies} for which a purely algebraic model has been proposed in \cite{ji2026quantum} and further discussed in \cite{ludewig2026quantum}. Given an equivalence between invertible phases and Brauer equivalence classes of invertible quasi-local algebras shown in the previous section, we propose a model for the Kitaev spectrum in the spirit of \cite{ji2026quantum} that, in our opinion, is more elementary than \cite{kubota2025stable} as it avoids many analytic aspects. More precisely, we relate the Kitaev spectrum to the $K$-theory space of a specific symmetric monoidal category which is the category of invertible quasi-local algebras over $\ZZ^d$ with morphisms being bounded spread isomorphisms. We show that our proposal for the Kitaev spectrum satisfies all the required properties.

\subsection{$K$-theory of symmetric monoidal categories} Let us recall some basic definitions and facts about $K$-theory we use in this paper. We refer the reader to \cite[Chapter IV]{weibel2013k} for more details.

We say that a symmetric monoidal category $C$ {\it has faithful translations} if for any object $\ql{C}$, the maps $\Hom(\ql{A},\ql{B}) \to \Hom(\ql{A} \otimes \ql{C}, \ql{B} \otimes \ql{C})$ are injective. We say that an object $A$ in $C$ is {\it absorbing}, if for each other object $X$, $A \otimes X \cong A$.

The $K$-theory $K(C)$ of $C$ is a connective spectrum that can be naturally associated with any symmetric monoidal category. When $C$ is a symmetric monoidal groupoid with faithful translations, the zeroth space of $K(C)$ can be defined as the classifying space $B(C^{-1} C)$ of the category $C^{-1} C$ whose objects are pairs $(\ql{A},\ql{A}')$, and whose morphisms $(\ql{A},\ql{A}') \to (\ql{B},\ql{B}')$ are equivalence classes of triples $(\ql{C} \in C,\alpha \in \Hom(\ql{A} \otimes \ql{C}, \ql{B}),\beta \in \Hom(\ql{A}' \otimes \ql{C}, \ql{B}'))$ up to equivalence: $(\ql{C},\alpha,\beta) \sim (\tilde{\ql{C}},\tilde{\alpha},\tilde{\beta})$ if there is an isomorphism $\varphi:\ql{C} \to \tilde{\ql{C}}$ satisfying $\alpha = \tilde{\alpha} (\Id \otimes \varphi)$ and $\beta = \tilde{\beta} (\Id \otimes \varphi)$.

If in addition $C$ has an object $S$ satisfying $S \otimes S \cong S$ and such that any other object $X$ divides $S$, i.e. there exists $Y$ such that $X \otimes Y \cong S$, then by \cite[Theorem IV.4.10]{weibel2013k}, the zeroth space of $K(C)$ can be identified with
\beq    \label{eq:KCinTermsofPlusConstruction}
K(C) \cong K_0(C) \times B \Gamma^{+},
\eeq
where 
$$
\Gamma = \text{colim} \l \Aut_C(S) \xrightarrow{\alpha \otimes \Id} \Aut_C(S^{\otimes 2}) \xrightarrow{\alpha \otimes \Id} \Aut_C(S^{\otimes 3}) \xrightarrow{\alpha \otimes \Id} ... \r,
$$
and $B \Gamma^{+}$ is the Quillen's plus construction of the group $\Gamma$ for a commutator subgroup $[\Gamma,\Gamma]$, which is defined as a space $B\Gamma^+$ with a map $B\Gamma \to B\Gamma^{+}$ that induces an isomorphism on homology $H_*(B\Gamma) \xrightarrow{\raisebox{-2pt}{$\scriptstyle\sim$}} H_{*}(B\Gamma^+)$ and $\pi_1(B\Gamma) \to \pi_{1}(B\Gamma^+)$ is the canonical map $\Gamma \to \Gamma/[\Gamma,\Gamma]$. In particular, $K_1(C) = \Gamma/[\Gamma,\Gamma]$.

A monoidal functor $F: D \to C$ is called {\it cofinal} if for any object $\ql{A} \in C$ there are objects $\tilde{\ql{A}} \in C$, $\ql{B} \in D$ such that $\ql{A} \otimes \tilde{\ql{A}} \cong F(\ql{B})$. By the cofinality theorem \cite[Theorem IV.4.11]{weibel2013k}, if for all $\ql{B} \in D$ we have $\Aut_{D}(\ql{B}) = \Aut_{C}(F(\ql{B}))$, then the basepoint components of $K(D)$ and $K(C)$ are homotopy equivalent. In particular, $K_{n \geq 1}(D) = K_{n \geq 1}(C)$. A full subcategory of $C$ is called {\it cofinal} if the inclusion functor is cofinal.

\subsection{$K$-theory of the category of invertible quasi-local algebras}
Let $\Azcat(X)$ (resp. $\sAzcat(X)$)\footnote{We use the same notation as in \cite{ji2026quantum, ludewig2026quantum} where in the context of nets of finite-dimensional algebras, invertible nets were called {\it Azumaya nets}. Following this terminology, we can also call invertible quasi-local algebras by {\it Azumaya quasi-local algebras}.} be a symmetric monoidal category whose objects are bosonic (resp. fermionic) invertible quasi-local algebras over a uniformly locally finite metric space $X$, whose morphisms are bounded spread isomorphisms, and whose monoidal structure is given by a tensor product of quasi-local algebras. The unit for $\Azcat(X)$ is given by $\Mat(X,s)$ for $s$ satisfying $s(j)=1$ for any $j \in X$. The unit for $\sAzcat(X)$ is given by $\sMat(X,s)$ for $s$ satisfying $s(j)= 1|0$ for any $j \in X$. We let $\opMatCatD(X)$ be the full subcategory of $\opAzcat(X)$ whose objects are $\opMat(X,s)$ for some $s$. It is a cofinal subcategory. Both $\opAzcat(X)$ and $\opMatCatD(X)$ are symmetric monoidal groupoids with faithful translations. The category $\opMatCatD(X)$ has an absorbing object $\opMat(X)$.

Note that the connected components $\pi_0(K(\opAzcat(X))) = K_0(\opAzcat(X))$ by definition correspond to Brauer equivalence classes of invertible quasi-local algebras $\opBrgr(X)$.

A relation $K(C(X)) \cong \Omega K(C(X \times \ZZ))$ between $K$-theory spaces of the categories of invertible quasi-local algebras can be derived using essentially the same ideas as in \cite[Section 3]{ludewig2026quantum} and \cite[Section 4]{ji2026quantum}, where the same relation has been shown for invertible nets of finite-dimensional algebras. The only additional ingredient is the generalization of some structural facts to general quasi-local algebras that were proven in Section \ref{sec:QuasiLocalAlgebras}. We repeat the arguments from \cite[Section 3]{ludewig2026quantum} adapted to our setting for completeness.

\begin{lemma}[Theorem 3.3 in \cite{ludewig2026quantum}]   \label{lma:Kspacecontractibility}
Let $X$ be a uniformly locally finite metric space. Then $K(\opAzcat(X \times \ZZ_{\geq 0}))$ is contractible.
\end{lemma}
\begin{proof}
Let $C = \opAzcat(X \times \ZZ_{\geq 0})$. To show contractibility of $K(C)$, it is enough to show an existence of a symmetric monoidal endofunctor $S : C \to C$ together with a natural isomorphism $S \cong \Id \otimes S$.

Let $s: X \times \ZZ_{\geq 0} \to X \times \ZZ_{\geq 0}$ be a shift map $(x,n) \to (x,n+1)$, $x \in X$, $n \in \ZZ_{\geq 0}$. There is a natural faithful symmetric monoidal endofunctor $S:C \to C$ that maps $\ql{A} \in C$ to $\ql{A} \otimes s_* \ql{A} \otimes s^2_* \ql{A} \otimes ...$. We have a natural isomorphism between $S$ and $\Id \otimes S$. Therefore, $K(C) \cong *$.
\end{proof}

\begin{theorem} \label{thm:LoopSpectrum}
Let $X$ be a uniformly locally finite metric space. We have $ K(\opAzcat(X)) = \Omega K(\opAzcat(X \times \ZZ))$. In particular, for $0 \leq k \leq d$, we have $\pi_{k} (K(\opAzcat(\ZZ^d))) = \opBrgr(\ZZ^{d-k})$.
\end{theorem}
\begin{proof}
The proof is essentially identical to the proof of \cite[Theorem 3.4]{ludewig2026quantum} adapted to our setting.

For $n \in \NN_0$, let $L_n = \ZZ_{<-n}$, $I_n = [-n,n] \cap \ZZ$, $R_n = \ZZ_{> n}$, and let $C^{(n)}$ be the full subcategory of $\Azcat(X \times \ZZ)$ whose objects have the form $\Mat(X \times L_n,s_L) \otimes \ql{A} \otimes \Mat(X \times R_n, s_R)$ for some $s_L: X \times L_n \to \NN$, $s_R: X \times R_n \to \NN$ and $\ql{A} \in \Azcat(X \times I_n)$. We also introduce categories defined similarly with some restrictions on $s_L,s_R$:
\begin{itemize}
    \item $C^{(n)}_0$ is the category with $s_L(L_n) = \{1\}$ and $s_R(R_n) = \{1\}$;
    \item $C^{(n)}_-$ is the category with $s_R(R_n) = \{1\}$;
    \item $C^{(n)}_+$ is the category with $s_L(L_n) = \{1\}$.
\end{itemize}
Let $C, C_-, C_0, C_+$ be the filtered colimits of $C^{(n)}, C^{(n)}_-, C^{(n)}_0, C^{(n)}_+$ over $n$, respectively. We have the following commutative diagram of symmetric monoidal categories with obvious inclusion maps:
\beq
\begin{tikzcd}
C_0 \arrow[r, "\iota_+"] \arrow[d, "\iota_-"'] & C_+ \arrow[d, ""] \\
C_- \arrow[r, ""']               & C.
\end{tikzcd}
\eeq
In addition, we introduce a double mapping cylinder category $P$ for the span $C_- \leftarrow{} C_0 \rightarrow C_+$ defined in \cite{thomason1982first}:
\begin{enumerate}
    \item the objects of $P$ are triples $(\ql{A}_-, \ql{A}_0, \ql{A}_+)$ for $\ql{A}_{-,0,+} \in C_{-,0,+}$;
    \item the morphisms between $(\ql{A}_-, \ql{A}_0, \ql{A}_+)$ and $(\tilde{\ql{A}}_-, \tilde{\ql{A}}_0, \tilde{\ql{A}}_+)$ are equivalence classes of quintuples $(\ql{B}_-, \ql{B}_+, f_-, f_0, f_+)$ for $\ql{B}_{\pm} \in C_0$, isomorphisms $f_-: \ql{A}_- \otimes \iota_-(\ql{B}_-) \to \tilde{\ql{A}}_-$, $f_0: \ql{A}_0 \to \ql{B}_- \otimes \tilde{\ql{A}}_0 \otimes \ql{B}_+$, $f_+: \ql{A}_+ \otimes \iota_+(\ql{B}_+) \to \tilde{\ql{A}}_+$ with equivalence relation $(\ql{B}_-, \ql{B}_+, f_-, f_0, f_+) \sim (\tilde{\ql{B}}_-, \tilde{\ql{B}}_+, \tilde{f}_-, \tilde{f}_0, \tilde{f}_+)$ if $\ql{B}_{\pm} \cong \tilde{\ql{B}}_{\pm}$.
\end{enumerate}
The significance of $P$ comes from the fact that the $K$-functor induces a homotopy pushout diagram of connective spectra
\beq
\begin{tikzcd}
K(C_0) \arrow[r, ""] \arrow[d, ""'] & K(C_+) \arrow[d, ""] \\
K(C_-) \arrow[r, ""']               & K(P).
\end{tikzcd}
\eeq
We have a natural symmetric monoidal functor $T: P \to C$ that maps $(\ql{A}_-, \ql{A}_0, \ql{A}_+)$ to $\ql{A}_- \otimes \ql{A}_0 \otimes \ql{A}_+$ which is faithful. By Quillen's theorem A \cite{quillen2006higher}, if for every object $\ql{C} \in C$, the classifying spaces of the comma categories $\ql{C} \downarrow T$ are weakly contractible, then the functor $T$ induces a homotopy equivalence $K(P) \to K(C)$.

Let $\ql{C}$ be an object in $C$. The comma category $\ql{C} \downarrow T$ is defined as follows:
\begin{enumerate}
    \item the objects are quadruples $(\ql{A}_-,\ql{A}_0,\ql{A}_+, \varphi)$, where $\ql{A}_{-,0,+} \in C_{-,0,+}$ and $\varphi$ is an isomorphism $\varphi:\ql{C} \to \ql{A}_- \otimes \ql{A}_0 \otimes \ql{A}_+$.
    \item the morphisms between $(\ql{A}_-, \ql{A}_0, \ql{A}_+, \varphi)$ and $(\tilde{\ql{A}}_-, \tilde{\ql{A}}_0, \tilde{\ql{A}}_+, \tilde{\varphi})$ are morphisms $f$ in $P$ such that $T(f)\varphi = \tilde{\varphi}$.
\end{enumerate}
Note that since $T$ is faithful, there is at most one morphism between objects of $\ql{C} \downarrow T$. Suppose $\ql{C} \in C^{(k)}$. Then for $m \geq k$, we have a decomposition $\ql{C} = \ql{C}_{-} \otimes \ql{C}_0 \otimes \ql{C}_{+}$ with $\ql{C}_-$, $\ql{C}_0$, $\ql{C}_+$ supported on $L_m$, $I_m$, $R_m$, respectively, that defines an object $(\ql{C}_-,\ql{C}_0,\ql{C}_+, \ql{C} \cong \ql{C}_{-} \otimes \ql{C}_0 \otimes \ql{C}_{+})$.

Let $(\ql{C} \downarrow T)^{(n,l)}$ be the full subcategories of $\ql{C} \downarrow T$ parametrized by $n,l \in \NN_0$ whose objects $(\ql{A}_-,\ql{A}_0,\ql{A}_+,\varphi)$ are such that $\ql{A}_{-,0,+} \in C^{(n)}_{-,0,+}$ and $\varphi$ has spread $l$. Pick an integer $m > \text{max}(n+l,k)$, and the corresponding decomposition $\ql{C} = \ql{C}_{-} \otimes \ql{C}_0 \otimes \ql{C}_{+}$. We're going to show that any object of $(\ql{C} \downarrow T)^{(n,l)}$ admits a single morphism from $(\ql{C}_-,\ql{C}_0,\ql{C}_+, \ql{C} \cong \ql{C}_{-} \otimes \ql{C}_0 \otimes \ql{C}_{+})$.

Let $(\ql{A}_-,\ql{A}_0,\ql{A}_+, \varphi)$ be an object in $(\ql{C} \downarrow T)^{(n,l)}$ with $\ql{A} = \ql{A}_{-} \otimes \ql{A}_0 \otimes \ql{A}_{+}$. We have a bounded spread isomorphism $\varphi:\ql{C}_{-} \otimes \ql{C}_0 \otimes \ql{C}_{+} \to \ql{A}_{-} \otimes \ql{A}_0 \otimes \ql{A}_{+}$. Each $\ql{C}_{\pm}$ is a tensor factor of $\ql{C}$, with $\opComm(\ql{C}_{-},\ql{C}) = \ql{C}_0 \otimes \ql{C}_+$ and $\opComm(\ql{C}_{+},\ql{C}) = \ql{C}_- \otimes \ql{C}_0$. By Lemma \ref{lma:Lud228}, the assignments
$$
U \to \ql{A}_U \cap \varphi(\ql{C}_{\pm}),
\qquad
U \to \ql{A}_U \cap \varphi \bigl( \opComm(\ql{C}_{\pm},\ql{C}) \bigr)
$$
define quasi-local subalgebras $\ql{D}_{\pm}, \tilde{\ql{D}}_{\pm} \subseteq \ql{A}$ with
$\ql{D}_{\pm}$ being tensor factors of $\ql{A}$, $\opComm(\ql{D}_{\pm},\ql{A}) = \tilde{\ql{D}}_{\pm}$, and $\varphi: \ql{C}_{\pm} \to \ql{D}_{\pm}$ being a bounded spread isomorphism. By Lemma \ref{lma:Lud230} and the fact that $\varphi(\ql{C}_{\pm}) \subset \ql{A}_{\pm}$, $\ql{D}_{\pm}$ are tensor factors of $\ql{A}_{\pm}$. Therefore, there exist bounded spread isomorphisms $\ql{C}_{\pm} \otimes \ql{B}_{\pm} \to \ql{A}_{\pm}$, where $\ql{B}_{\pm} = \opComm(\ql{D}_{\pm}, \ql{A}_{\pm})$. Since $\ql{B}_{\pm}$ (super)commutes with $\varphi(\ql{C}_{\pm})$ and lies in $\ql{A}$, we have $\ql{B}_{\pm} \subseteq \tilde{\ql{D}}_{\pm}$ and $\ql{B}_{\pm} \subseteq \varphi(\opComm(\ql{C}_{\pm},\ql{C}))$. Since $\varphi$ has spread $l$, we get $\ql{B}_{\pm} \in C^{(m+l)}_0$. Thus, we have a morphism in $P$ between $(\ql{C}_-,\ql{C}_0,\ql{C}_+)$ and $(\ql{A}_-,\ql{A}_0,\ql{A}_+)$ such that its $T$-image defines a morphism in $(\ql{C} \downarrow T)$.

It follows that for any $n,l \in \NN_0$, there is an object in $(\ql{C} \downarrow T)$ with a morphism to any object in $(\ql{C} \downarrow T)^{(n,l)}$, which is unique by faithfulness of $T$. Thus the classifying spaces of $(\ql{C} \downarrow T)^{(n,l)}$ are contractible in the classifying space of $(\ql{C} \downarrow T)$, and since $(\ql{C} \downarrow T)$ is a filtered colimit of $(\ql{C} \downarrow T)^{(n,l)}$ over $n,l$, the classifying space of $(\ql{C} \downarrow T)$ is contractible.

It follows that we have the following homotopy pushout diagram of connective spectra
\beq    \label{eq:htpypushout}
\begin{tikzcd}
K(C_0) \arrow[r, ""] \arrow[d, ""'] & K(C_+) \arrow[d, ""] \\
K(C_-) \arrow[r, ""']               & K(C).
\end{tikzcd}
\eeq
The functor $C_0 \to \opAzcat(X)$ that associates a quasi-local algebra $\ql{A}$ defined via $\ql{A}_U = \ql{B}_{U \times I_m}$ to $\ql{B} \in C_0^{(m)}$ yields an equivalence of $\opAzcat(X)$ and $C_0$, implying $K(\opAzcat(X)) \cong K(C_0)$. 

Note that for $n \geq 0$, the categories $\opMatCatD(X \times \ZZ_{\leq n})$, $\opMatCatD(X \times \ZZ)$, $\opMatCatD(X \times \ZZ_{\geq -n})$ are cofinal subcategories of $C^{(n)}_-$, $C$, $C^{(n)}_+$ and $\opAzcat(X \times \ZZ_{\leq n})$, $\opAzcat(X \times \ZZ)$, $\opAzcat(X \times \ZZ_{\geq -n})$, respectively. Therefore, by the cofinality theorem \cite[Theorem IV.4.11]{weibel2013k}, we have $\Omega K(C^{(n)}_-) \cong \Omega K(\opAzcat(X \times \ZZ_{\leq n}))$, $\Omega K(C) \cong \Omega K(\opAzcat(X \times \ZZ))$, $\Omega K(C^{(n)}_+) \cong \Omega K(\opAzcat(X \times \ZZ_{\geq -n}))$. By Lemma \ref{lma:Kspacecontractibility}, we have $K(C_0) \cong K(\opAzcat(X))$.

Let $\ql{A} \in C^{(n)}_{-}$ with an inverse $\tilde{\ql{A}} \in C^{(n)}_-$, and let $s: X \times \ZZ \to X \times \ZZ$ be a shift map $(x,l) \to (x,l-1)$, $x \in X$, $l \in \ZZ$. There are bounded spread isomorphisms relating $\ql{A}$ to $\ql{A} \otimes s^{2n}_* \tilde{\ql{A}} \otimes s^{4n}_* \ql{A} \otimes s^{6n}_* \tilde{\ql{A}} \otimes ...$ and the latter to a product quasi-local algebra. Therefore, $K_0(C^{(n)}_-) = 0$. Similarly, $K_0(C^{(n)}_+) = 0$. Since $\Omega K(C^{(n)}_{\pm}) \cong *$, we have $K(C^{(n)}_{\pm}) \cong *$ which implies $K(C_{\pm}) \cong *$.

Thus, the homotopy pushout diagram Eq. (\ref{eq:htpypushout}) becomes
\beq
\begin{tikzcd}
K(C_0) \arrow[r, ""] \arrow[d, ""'] & * \arrow[d, ""] \\
* \arrow[r, ""']               & K(C)
\end{tikzcd}
\eeq
and together with $K(C_0) \cong K(\opAzcat(X))$, $\Omega K(C) \cong \Omega K(\opAzcat(X \times \ZZ))$ implies an equivalence $K(\opAzcat(X)) = \Omega K(\opAzcat(X \times \ZZ))$.

\end{proof}

\begin{remark}  \label{rmk:WeCanUseQuillenPlus}
Since $\opMat(X)$ is an object in $\opAzcat(X)$ such that $\opMat(X) \otimes \opMat(X) \cong \opMat(X)$ and for any object $\ql{A} \in \opAzcat(X)$ there exists an object $\tilde{\ql{A}} \in \opAzcat(X)$ such that $\ql{A} \otimes \tilde{\ql{A}} \cong \opMat(X)$, we can apply Quillen's plus construction Eq. \ref{eq:KCinTermsofPlusConstruction}
\beq    \label{eq:KCinTermsofPlusConstruction2}
K(\opAzcat(X)) = \opBrgr(X) \times B\Gamma^{+},
\eeq
where $\Gamma$ is the group of stable QCAs
\beq
\Gamma := \text{colim} \l \Autbs(\opMat(X)) \xrightarrow{\alpha \otimes \Id} \Autbs(\opMat(X)^{\otimes 2}) \xrightarrow{\alpha \otimes \Id}  ... \r.
\eeq
\end{remark}

We can give an explicit characterization of $K_1(\opAzcat(X))$ in terms of QCAs:

\begin{prop}    \label{prop:K1asQCA}
Let $X$ be a uniformly locally finite metric space. We have $K_1(\opAzcat(X)) \cong \opQCA(X)$.
\end{prop}
\begin{proof}
Note that $\opMatCatD(X)$ is a cofinal subcategory of $\opAzcat(X)$. That implies $K_1(\opAzcat(X)) = K_1(\opMatCatD(X))$. Since $\opMatCatD(X)$ has an absorbing object, we can use Eq. (\ref{eq:KCinTermsofPlusConstruction}) to identify $K_1(\opMatCatD(X))$ with $\Gamma/[\Gamma,\Gamma]$, where $\Gamma$ is the group of stable QCAs over $X$.

By Brown-Pearcy theorem, for a separable infinite-dimensional Hilbert space $\hilb{H}$, any element in the group $U(\hilb{H})$ is a multiplicative commutator. Any even unitary operator on a separable Hilbert superspace of superdimension $\infty|\infty$ is a direct sum of unitaries acting on the even and the odd subspaces. Therefore, by the same theorem, any even unitary is a multiplicative commutator of even unitaries. It follows that any quantum circuit automorphism that is a stable QCA is a composition of multiplicative commutators of quantum circuits which are stable QCAs.

By Lemma \ref{lma:QCACompositionInBrGr}, any commutator of stable QCAs define a trivial element in $\opQCA(X)$, and therefore is a quantum circuit.

Combining the implications of each paragraph, we obtain $K_1(\opAzcat(X)) \cong \opQCA(X)$.
\end{proof}

\begin{prop}    \label{prop:bosonicK(C(pt))}
We have $K(\Azcat(\text{pt}))\cong K(\RR/\ZZ, 2)$\footnote{Here and below, $\RR/\ZZ$ denotes the underlying abstract group, not a topological group.} and $K(\Azcat(\ZZ)) \cong \Brgr(\ZZ) \times K(\RR/\ZZ,3)$ at the level of zeroth $K$-spaces. 
\end{prop}
\begin{proof}
Let $C = \Azcat(\text{pt})$. The objects of $C$ are separable type $\rm{I}$ factors and its morphisms are $*$-isomorphisms. By Proposition \ref{prop:0dBrauerGroup}, we have $K_0(C) = 0$. By Remark \ref{rmk:WeCanUseQuillenPlus}, we can use the Quillen's plus construction Eq. (\ref{eq:KCinTermsofPlusConstruction2}) for $K(C)$. 

Let $U_n = U(\hilb{H}^{\otimes n})^{\delta}$ (i.e., $U(\hilb{H}^{\otimes n})$ with discrete topology) and $\tilde{\Gamma} = \text{colim}(U_1 \to U_2 \to ...)$ be the filtered colimit with respect to injections $u \to u \otimes 1_{\hilb{H}}$. Similarly, let $PU_n = \Aut(B(\hilb{H}^{\otimes n}))^{\delta}$ and $\Gamma = \text{colim}(PU_1 \to PU_2 \to ...)$ which coincides with the group of stable QCAs. By Brown-Pearcy theorem, any element of $U(\hilb{H})$ can be written as a multiplicative commutator. Therefore, any element of $\Gamma$ is a multiplicative commutator. Hence, we have $K_1(C) = \Gamma/[\Gamma,\Gamma] = 0$.

As was shown by Harpe and McDuff \cite{de1983acyclic}, the group $U(\hilb{H})^{\delta}$ is acyclic $H_{*}(U(\hilb{H})^{\delta}) = 0$. Since group homology commutes with filtered colimits along injections, $\tilde{\Gamma}$ is acyclic as well. Since $\tilde{\Gamma}$ is a central extension of $\Gamma$ by $\RR/\ZZ$, which is classified by a map to $K(\RR/\ZZ,2)$, we have a fibration sequence $B \tilde{\Gamma} \to B \Gamma \to K(\RR/\ZZ,2)$ with the simply connected base and the acyclic fiber. Since the Serre spectral sequence collapses, we get a homology isomorphism $B \Gamma \to K(\RR/\ZZ,2)$. By the universal property of the Quillen's plus construction, it factors through $B\Gamma^+ \to K(\RR/\ZZ,2)$, and since both $B\Gamma \to B \Gamma^{+}$ and $B \Gamma \to K(\RR/\ZZ,2)$ are homology isomorphisms, $B\Gamma^+ \to K(\RR/\ZZ,2)$ is a homology isomorphism as well. Since both $B\Gamma^+$ and $K(\RR/\ZZ,2)$ are simply connected, $K(C) \cong B\Gamma^+ \to K(\RR/\ZZ,2)$ is a homotopy equivalence. 

By Theorem \ref{thm:LoopSpectrum}, $\Omega K(\Azcat(\ZZ)) \cong K(\RR/\ZZ, 2)$ that implies $K(\Azcat(\ZZ)) \cong \Brgr(\ZZ) \times K(\RR/\ZZ, 3)$.
\end{proof}

\begin{prop}    \label{prop:fermionicK(C(pt))}
We have $K(\sAzcat(\text{pt}))\cong \ZZ/2 \times K(\ZZ/2,1) \times K(\RR/\ZZ, 2)$.  
\end{prop}
\begin{proof}
Let $C = \sAzcat(\text{pt})$. The objects of $C$ are separable type $\rm{I}$ superfactors and its morphisms are graded $*$-isomorphisms. By Proposition \ref{prop:0dBrauerGroup}, $K_0(C) = \ZZ/2$. By Remark \ref{rmk:WeCanUseQuillenPlus}, we can use the Quillen's plus construction Eq. (\ref{eq:KCinTermsofPlusConstruction2}) for $K(C)$. 

Let $\hilb{H} = \hilb{H}^{(0)} \oplus \hilb{H}^{(1)}$ be a separable Hilbert superspace of dimension $\infty|\infty$. Graded $*$-automorphisms of $B(\hilb{H})$ are inner $\Ad_u$ with $u$ being homogeneous. After choosing an isomorphism of (ungraded) Hilbert spaces $\iota:\hilb{H}^{(0)} \to \hilb{H}^{(1)}$, any homogeneous unitary $u$ has the form $(\phi \otimes 1)^{|u|}(v^{(0)} \oplus v^{(1)})$, where $v^{(0)} \in U(\hilb{H}^{(0)})$, $v^{(1)} \in U(\hilb{H}^{(1)})$ and 
$$
\phi = \begin{pmatrix}
    0 && \iota^{-1}\\
    \iota && 0
\end{pmatrix}.
$$
Let $U^{(h)}_n$ be the subgroup of homogeneous elements of $U(\hilb{H}^{\otimes n})^{\delta}$ (i.e., $U(\hilb{H}^{\otimes n})$ with discrete topology) and $\tilde{\Gamma} = \text{colim}(U^{(h)}_1 \to U^{(h)}_2 \to ...)$ be the filtered colimit with respect to injections $u \to u \otimes 1_{\hilb{H}}$. Similarly, let $PU^{(h)}_n$
be the subgroup of graded $*$-automorphisms of $B(\hilb{H}^{\otimes n})$ (with discrete topology) and $\Gamma = \text{colim}(PU^{(h)}_1 \to PU^{(h)}_2 \to ...)$ which coincides with the group of stable QCAs. The group $\tilde{\Gamma}$ is a central extension of $\Gamma$ by $\RR/\ZZ$: $1 \to \RR/\ZZ \to \tilde{\Gamma} \xrightarrow{\pi} \Gamma \to 1$. We also have a surjection $\text{deg}:\Gamma \to \ZZ/2$ that corresponds to the degree of the implementing unitary. Let $\tilde{\Gamma}_0 = \text{Ker}(\text{deg} \circ \pi) = U(\hilb{H}^{(0)}) \times U(\hilb{H}^{(1)})$ and $\Gamma_0 = \text{Ker}(\text{deg})$. By Brown-Pearcy theorem $\Gamma_0 = [\Gamma,\Gamma] = [\Gamma_0,\Gamma_0]$. Hence, we have $K_1(C) = \Gamma/[\Gamma,\Gamma]  = \ZZ/2$.

The central extension $1 \to \RR/\ZZ \to \tilde{\Gamma} \xrightarrow{\pi} \Gamma \to 1$ is classified by a map $B\Gamma \to K(\RR/\ZZ,2)$, and we have a fibration sequence $B \tilde{\Gamma} \xrightarrow{B \pi} B \Gamma \xrightarrow{k} K(\RR/\ZZ,2)$. Define $f: B\Gamma \to B(\ZZ/2) \times K(\RR/\ZZ,2)$ via $(B \text{deg},k)$. We have a homotopy pullback diagram
\beq
\begin{tikzcd}
B \Gamma_0 \arrow[r, ""] \arrow[d, ""'] & B\Gamma \arrow[d, "f"] \\
K(\RR/\ZZ,2) \arrow[r, ""']               & B\ZZ/2 \times K(\RR/\ZZ,2)
\end{tikzcd}
\eeq
By the same argument as in the proof of Proposition \ref{prop:bosonicK(C(pt))}, $B \Gamma_0 \to K(\RR/\ZZ,2)$ is acyclic, implying that $f$ is acyclic. Since the kernel of the induced by $f$ map on the fundamental group is $\Gamma_0$, by the universal property of the Quillen's plus construction, $B \Gamma^{+} \cong B\ZZ/2 \times K(\RR/Z,2)$. Thus, $K(C)\cong \ZZ/2 \times K(\ZZ/2,1) \times K(\RR/\ZZ, 2)$.

\end{proof}

\subsubsection{Kitaev spectra}

The spectra $K(\Azcat(\ZZ^{d-1}))$ and $K(\sAzcat(\ZZ^{d-1}))$ satisfy the desired properties of Kitaev spectra $\Kit^{(d)}$ and $\sKit^{(d)}$ with the exception that their highest homotopy group is $\pi_{d+1} = \RR/\ZZ$ instead of $\pi_{d+1}=0$, $\pi_{d+2} = \ZZ$. The origin of this mismatch comes from the fact that we use the discrete topology on the set of morphisms of the categories $\Azcat(X)$ and $\sAzcat(X)$. This can be easily fixed by the following definitions of the Kitaev spectra:
\bqa
\Kit^{(d)} := \text{cofib}\l \Sigma^{d+1} H \RR \to K(\Azcat(\ZZ^{d-1})) \r,\\
\sKit^{(d)} := \text{cofib}\l \Sigma^{d+1} H \RR \to K(\sAzcat(\ZZ^{d-1})) \r,
\eqa
where we used an exact sequence $0 \to \ZZ \to \RR \to \RR/\ZZ \to 0$.

\begin{remark}
One can alternatively use the topology on the set of morphisms of $\opAzcat(\ZZ^{d-1})$ and regard it an $\infty$-category. In that case, one can anticipate that the $K$-theory functor produces spectra with $\pi_{d+1} = 0$, $\pi_{d+2} = \ZZ$ providing Kitaev spectra directly. We postpone a detailed analysis.
\end{remark}

\subsection{SPT phases and anomaly indices}

The algebraic model described in the previous section immediately allows us to define anomaly indices for group actions on a quasi-local algebra, in the same way as they have been considered recently in \cite{ji2026k} in the setting of net of finite-dimensional algebras. Furthermore, our setting allows to define invariants of invertible states with symmetries (also known as {\it SPT states}). For simplicity, we consider the case of bosonic quasi-local algebras only.

Let $X$ be a uniformly locally finite metric space and $G$ be a group. Consider a category $(\Azcat(X))^{BG}$ of objects of $\Azcat(X)$ equipped with $G$-action (equivalently functors from $BG$ to $\Azcat(X)$). With any object $\ql{A}$ of $(\Azcat(X))^{BG}$, we have an associated map $BG \to K(\Azcat(X))$.

\begin{definition}
The {\it anomaly index} of an object $\ql{A}$ of $(\Azcat(X))^{BG}$ is the homotopy class $\omega(\ql{A}) \in H^0(BG,K(\Azcat(X)))$ of the associated map $BG \to K(\Azcat(X))$.
\end{definition}

By  definition, the anomaly index provides an invariant of $G$-equivariant Brauer equivalence classes, which we can define by $\Brgr_G(X) := K_0(\Azcat(X)^{BG})$.

Similarly, we can define $G$-equivariant generalization of invertible phases and their invariants for countable groups $G$.

\begin{definition}
Let $X$ be a uniformly locally finite metric space. We say that an action of a countable group $G$ on $\Mat(X,s)$ is {\it on-site} if there exist unitary representations $\{u^{(g)}_x \in \opMat(X)_{\{x\}}\}_{x \in X, g \in G}$ of $G$ on $\opMat(X)_{\{x\}}$ such that an element $g$ acts on $\opMat(X)_U$ by conjugations with $\prod_{x \in U} u^{(g)}_x$. We say that an action on $\Mat(X)$ is {\it stable on-site} if for each $x \in X$, the corresponding unitary representation is $l^2(G)^{\oplus \NN}$, where $l^2(G)$ is the left regular representation. We denote the quasi-local algebra $\Mat(X)$ equipped with a stable on-site $G$-action by $\Mat_G(X)$.
\end{definition}

\begin{definition}
Let $X$ be a uniformly locally finite metric space, and $G$ be a countable group. We say that $G$-invariant states $\omega_1$, $\omega_2$ on $\Mat_G(X)$ are {\it in the same SPT phase} if for an unentangled $G$-invariant state $\omega_0$ on $\Mat_G(X)$, the states $\omega_1 \otimes \omega_0$ and $\omega_2 \otimes \omega_0$ are related by a quantum circuit with $G$-invariant gates. This is an equivalence relation, and the equivalence classes are called {\it SPT phases}. Tensor product of states and a (zero-spread) $G$-equivariant isomorphism $\Mat_G(X) \otimes \Mat_G(X) \cong \Mat_G(X)$ equips this set with the structure of a commutative monoid. The abelian group of invertible SPT phases is denoted $\IP_G(X)$. The phase of the unit element of $\IP_G(X)$ is called {\it trivial}.
\end{definition}

Given an a state in a non-trivial invertible SPT phase, we can produce an ``anomalous'' $G$ action on the corresponding boundary algebra:

\begin{definition}
Let $X$ be a uniformly locally finite metric space, $G$ be a countable group, and let $\omega$ be a $G$-invariant state in an invertible SPT phase on $\Mat_G(X \times \ZZ)$. The $G$-action on $\Mat_G(X \times \ZZ)$ induces a $G$-action on the von Neumann algebras $\pi_{\omega}(\Mat_G(U \times \ZZ))''$ for bounded subsets $U \subset X$. The {\it associated $G$-action on the boundary algebra} is the $G$-action induced in this way.
\end{definition}

\begin{remark}
In the setting of lattice systems with finite-dimensional on-site Hilbert spaces, one generally needs a quantum circuit which disentangles a state $\omega$ in an invertible SPT phase and commutes (globally) with the action of $G$ to produce an anomalous boundary symmetry \cite{else2014classifying,zhang2022topologicalinvariantssptentanglers,czajka2025anomalies}. For example, it is expected that fermionic SPT phases over $\ZZ^2$ protected by $\ZZ/2$-symmetry are classified by $\ZZ/8$, and that the generator of $\ZZ/8$ does not correspond to an anomalous $\ZZ/2$-action on a 1d spin system, as the generator of $\ZZ/2$ would have to be a non-trivial 1d QCA \cite{Ellison_2019,jones20191dlatticemodelsboundary}. In contrast, with infinite-dimensional on-site Hilbert spaces, this issue does not arise. For example, we may construct from the invertible algebra of the chiral Majorana fermion and its inverse, constructed in Section \ref{subsec:ChiralMajorana}, a trivial quasi-local algebra equipped with a $\ZZ/2$ action coming from the fermion parity of just one of the chiral Majorana algebras. This anomalous action corresponds to a generator of the $\ZZ/8$ group.

\end{remark}

\begin{definition}
Let $X$ be a uniformly locally finite metric space and $G$ be a countable group, and let $\omega$ be a $G$-invariant invertible SPT state $\omega$ on $\Mat_G(X \times \ZZ)$.  We define the {\it SPT index} of $\omega$ to be the $G$-equivariant Brauer equivalence class of its boundary algebra, equipped with the associated $G$-action.
\end{definition}

An SPT index of a $G$-invariant invertible SPT state is an invariant of the SPT phase due to the following
\begin{prop}
Let $X \times \ZZ$ be a uniformly locally finite metric space and $G$ be a countable group. Two $G$-invariant invertible states on $\Mat_G(X \times \ZZ)$ in the same SPT phase have $G$-equivariant Brauer equivalent right boundary algebras.
\end{prop}
\begin{proof}
Let $\omega$ be a $G$-invariant invertible SPT state on $\Mat_G(X \times \ZZ)$. It is manifest that $\omega$ and $\omega \otimes \omega_0$ for an unentangled $G$-invariant state $\omega_0$ on $\Mat_G(X \times \ZZ)$ have $G$-equivariant Brauer equivalent right boundary algebras.

Suppose $\alpha$ is a quantum circuit on $\Mat_G(X \times \ZZ)$ of spread $l$ with gates being invariant under $G$-action. We can represent $\alpha$ as a composition of $\alpha_- \alpha_+$ of quantum circuits $\alpha_-$, $\alpha_{+}$ with all their unitary gates being in $\Mat_G(X \times \ZZ_{<0})$, $\Mat_G(X \times \ZZ_{\geq -l})$. The states $\omega$, $\omega \alpha_-$ have the same restrictions to $\Mat_G(X \times \ZZ_{\geq 0})$, and therefore have the same right boundary algebras. The 0-th and the $(-l)$-th right boundary algebras of the state $\omega \alpha_-$ are related by tensoring with $\Mat_G(X)$, and therefore $G$-equivariant Brauer equivalent. The quantum circuit automorphism $\alpha_+$ induces a $G$-equivariant bounded spread isomorphism between the $(-l)$-th right boundary algebras of $\omega \alpha_-$ and of $\omega \alpha_- \alpha_+$. Composition of all these Brauer equivalences implies that two states in the same SPT phase has $G$-equivariant Brauer equivalent right boundary algebras.
\end{proof}

\section{Invertible quasi-local algebras over $\ZZ$ from free fermions} \label{sec:invertiblefreefermions}

In this section we give an explicit construction of bosonic and fermionic quasi-local algebras over $\ZZ$ which appear from free fermionic systems. In Section \ref{sec:cmInvariant}, we prove that that these quasi-local algebras represent non-trivial elements in $\opBrgr(\ZZ)$.

\subsection{Self-dual CAR algebra and quasi-free states}

We make heavy use of the theory of quasi-free states on CAR algebras. Some conventions and notations we use include:

\begin{itemize}
    \item By a real structure on a Hilbert space we mean an anti-unitary involution. When a Hilbert space $\hilb{K}$ is equipped with a real structure, we denote the image of $f \in \hilb{K}$ under the involution by $\overline{f} \in \hilb{K}$. For an operator $A \in B(\hilb{K})$, we define its complex conjugate $\overline{A} \in B(\hilb{K})$ via $\overline{A} \,f := \overline{A \bar{f}}$. We say that $f \in \hilb{K}$ and $A \in B(\hilb{K})$ are real if $f = \overline{f}$ and $A = \overline{A}$, respectively. We say that $A$ is skew symmetric if $\overline{A} = - A^*$. We denote the group of real unitary operators by $O(\hilb{K}) \subset U(\hilb{K})$. The Lie algebra of skew self-adjoint bounded operators is denoted by $\mathfrak{u}(\hilb{K})$, and the Lie algebra of real skew self-adjoint bounded operators is denoted by $\mathfrak{o}(\hilb{K}) \subset \mathfrak{u}(\hilb{K})$. We say that an orthogonal projection $P \in B(\hilb{K})$ is a {\it basis projection} if $P + \overline{P} = 1$.
    
    \item We assume that Hilbert spaces $L^2(X)$ for measure spaces $X$ are equipped with a standard real structure that acts by complex conjugation $f \to \overline{f}$. By $M_f: L^2(X) \to L^2(X)$, we denote the multiplication by $f \in L^{\infty}(X)$ operator.
    
    \item We use the following conventions for the Fourier transform $\CF f$ of $f \in L^1(\RR)$
    $$
    (\CF f)(k) = \int_{-\infty}^{\infty} f(x) e^{i k x} d x,
    $$
    $$
    (\CF^{-1} \tilde{f})(x) = \int_{-\infty}^{\infty} \tilde{f}(k) e^{-i k x} \frac{dk}{2 \pi}.
    $$
    We denote its extension to a bounded operator on $L^2(\RR)$ by the same symbol $\CF:L^2(\RR) \to L^2(\RR)$.
    
    \item For a function $a \in L^{\infty}(\RR)$, we let $a\l\hD\r$ be a bounded operator on $L^2(\RR)$ defined by $a\l\hD\r = \CF^{-1} M_a \CF$.

    \item By {\it a bump function} on $\RR$, we mean a function in $C^{\infty}_{c}(\RR)$, i.e. a smooth compactly supported function.

\end{itemize}

\begin{definition}
Let $\hilb{K}$ be a Hilbert space equipped with a real structure. A {\it self-dual CAR algebra} associated with $\hilb{K}$ is a $C^*$-superalgebra $\CAR(\hilb{K})$ which is a unique $C^*$-norm completion of a $*$-superalgebra generated by $1$ and $\{c_f\}_{f \in \hilb{H}}$ which depend linearly on $f$ subject to the relations 
$$
\{c_f,c_{g}\} = \lal \overline{f},g \ral,\,\,\,\,\, c^*_f = c_{\overline{f}},$$
and with the fermionic parity defined by $\fparity(c_f)=-c_f$.
\end{definition}

Below we review some standard facts about self-dual CAR algebras and quasi-free states (see e.g. \cite{araki1971quasifree, araki1988schwinger}).
\begin{definition}
Let $\hilb{K}$ be a Hilbert space equipped with a real structure, and let $U \in O(\hilb{K})$. A {\it Bogolyubov automorphism} associated with $U$ is a $*$-automorphism $\beta_U$ of $\CAR(\hilb{K})$ defined on generators by $\beta_U(c_f) = c_{U f}$.
\end{definition}

\begin{definition}
Let $\hilb{K}$ be a Hilbert space equipped with a real structure, and let $A \in B(\hilb{K})$ be skew-symmetric. A {\it Bogolyubov derivation} associated with $A$ is a densely-defined derivation $\der{D}_A$ of $\CAR(\hilb{K})$ defined on generators by $\der{D}_A(c_f) = c_{A f}$. If $A \in \mathfrak{o}(\hilb{K})$, then the derivation $\der{D}_A$ is a $*$-derivation. It defines a one-parameter family of Bogolyubov automorphisms $\{ e^{t \der{D}_A} \}_{t \in \RR}$ via $e^{t \der{D}_A} = \beta_{U_t}$ for $U_t = e^{t A}$.
\end{definition}

\begin{definition}
Let $\hilb{K}$ be a Hilbert space equipped with a real structure. A state $\omega$ on $\CAR(\hilb{K})$ satisfying the following relations is called a {\it quasi-free state}:
\beq
\omega(c_{f_1}...c_{f_{2n-1}}) = 0,
\eeq
\beq
\omega(c_{f_1}...c_{f_{2n}}) = \frac{1}{2^{n} n! } \sum_{s} \sign(s) \prod_{j=1}^{n} \omega(c_{f_{s(j)}} c_{f_{s(j+n)}}).
\eeq
\end{definition}
By \cite[Lemma 3.2-3.3]{araki1971quasifree}, there is a one-to-one correspondence between quasi-free state on $\CAR(\hilb{K})$ and bounded self-adjoint operators $S \in B(\hilb{K})$ such that $0 \leq S \leq 1$ and $S + \overline{S} = 1$. We denote the state corresponding to $S$ by $\QFS{S}$ which is fully characterized by
\beq
\QFS{S} (c_{f_1} c_{f_2}) = \lal \overline{f_1}, S f_2 \ral.
\eeq
By \cite[Theorem 1]{araki1971quasifree}, two states $\QFS{S}$, $\QFS{S'}$ are quasi-equivalent if and only if $S^{1/2}-S'^{1/2}$ is a Hilbert-Schmidt operator. By \cite[Lemma 4.3]{araki1971quasifree}, the state $\QFS{P}$ is pure if and only if $P$ is a basis projection.

Given a basis projection $P$, we let $O_P(\hilb{K})$ be the subgroup of elements $U \in O(\hilb{K})$ such that $[P,U]$ is Hilbert-Schmidt. As explained in \cite[Section 4]{araki1988schwinger}, a Bogolyubov automorphism $\beta_U$ can be implemented in the GNS-representation of $\QFS{P}$ if and only if $U \in O_P(\hilb{K})$, i.e. one can find a unitary operator $\hat{U} \in B(\hilb{H}_{\QFS{P}})$ such that
$$
\pi_{\QFS{P}}(\bog_U(a)) = \hat{U}^{-1} \pi_{\QFS{P}}(a) \hat{U}, \,\,\, a \in \CAR(\hilb{K}).
$$
The unitary operator $\hat{U}$ is unique up to an overall phase factor.

Similarly, let $\mathfrak{o}_P(\hilb{K})$ be the Lie subalgebra of elements $A \in \mathfrak{o}(\hilb{K})$ such that $[P,A]$ is Hilbert-Schmidt. The Bogolyubov automorphisms $\beta_{U_t} = e^{t \der{D}_A}$, $U_t = e^{t A}$, $t \in \RR$ are implementable if and only if $A \in \mathfrak{o}_P(\hilb{K})$. We may associate a skew self-adjoint (generally,  unbounded) operator $\hat{A}$ acting in $\hilb{H}_{\QFS{P}}$ such that $\hat{U}_t = e^{t \hat{A}}$ implements $\beta_{U_t}$. The operator $\hat{A}$ is unique up to an overall constant which we can fix by the condition $\lal \Omega_{\QFS{P}}, \hat{A} \Omega_{\QFS{P}} \ral = 0$. We say that $\hat{A}$ implements the derivation $\der{D}_A$.

\subsection{Chiral Majorana fermion}\label{subsec:ChiralMajorana}

Let $\hilb{K} = L^2(\RR) \otimes \CCC^N$ for some $N>0$ with the standard real structure. Any state $\omega$ on the algebra $\CAR(\hilb{K})$ provides a quasi-local algebra $\ql{A}$ over $\RR$ defined by
the net of von Neumann superalgebras
\beq
U \to \ql{A}_U =  \pi_{\QFS{P}}(\CAR(L^2(\interior{U}) \otimes \CCC^N))''
\eeq
for a bounded subset $U \subset \RR$\footnote{The interior $\interior{U}$ of $U$ is automatically measurable, so that $L^2(\interior{U})$ is well-defined.}. In this case, we say that $\ql{A}$ is generated by a state $\omega$ on $\CAR(\hilb{K})$.

Let $\hilb{K} = L^2(\RR)$ with the standard real structure, and let 
\beq \label{eq:LeftMovingMajProj}
P = (1+\sign(-i \p_x))/2
\eeq
be a basis projection. The quasi-free state $\QFS{P}$ of the algebra $\CAR(\hilb{K})$ is the ground state 
$$-i \frac{d}{dt}\QFS{P}(a^* \tau_t(a))\bigg|_{t=0} \geq 0$$
of the dynamics $\{\tau_t\}_{t \in \RR}$ where $\tau_t$ is the automorphism shifting all the observables by $t$ to the left: $\tau_t(c_f) = c_{T_t f}$, $(T_t f)(x) = f(x+t)$.

\begin{definition}
A quasi-local algebra of a {\it chiral left-moving Majorana fermion} $\Amaj$ over $\RR$ is the quasi-local algebra generated by a quasi-free state $\QFS{P}$ for $P$ from Eq. (\ref{eq:LeftMovingMajProj}). Similarly, a quasi-local algebra of a {\it chiral right-moving Majorana fermion} $\overline{\Amaj}$ is the quasi-local algebra generated by $\QFS{\overline{P}}$ which coincides with the complex conjugate of $\Amaj$. The discrete versions $\Amajd$, $\overline{\Amajd}$ over $\ZZ$ are defined by the pushforward quasi-local algebras with respect to the map $f:\RR \to \ZZ$, $x \to \lfloor  x + 1/2 \rfloor$.
\end{definition}

In the following, by a {\it discretization of a quasi-local algebra over} $\RR$ we mean the pushforward quasi-local algebra over $\ZZ$ with respect to the map $f:\RR \to \ZZ$, $x \to \lfloor  x + 1/2 \rfloor$.

Our goal is to prove the following theorem.

\begin{theorem} \label{thm:InvertibilityofAmaj}
There is a bounded spread isomorphism between $\Amajd \otimes \overline{\Amajd}$ and $\sMat(\ZZ)$.
\end{theorem}

First we prove a technical lemma. 

\begin{lemma}   \label{lma:gfunction}
Let $g$ be a positive smooth function on $\RR$, such that $\tilde{g} = \CF^{-1} g \in \Bumpfns$ with  $\supp(\tilde{g}) \subset [-a,a]$. Then $f = \sqrt{g} e^{i \theta}$,
$$
\theta(x)= - \frac{a}{2} x - \frac{1}{2 \pi} \text{p.v.} \int_{-\infty}^{\infty} \l \frac{1}{t-x} - \frac{t}{1+t^2} \r \log g(t) dt
$$
is a well-defined smooth function and $\tilde{f} = \CF^{-1} f \in \Bumpfns$ with $\supp(\tilde{f}) \subset [-a/2,a/2]$.
\end{lemma}

\begin{proof}
By the smooth version of the Paley-Wiener theorem, the function $g$ extends to an entire function $G$ on $\CCC$ of exponential type $a$ with the following estimate: for any $N$, there exists $C_N$ such that $|G(z)| \leq C_N (1+|z|)^{-N} e^{a |\Im z|}$.

From the discussion on Page 3 of \cite{mitkovski2010polya}, it follows that the function $\log |g(t)|/(1+t^2)$ is integrable and therefore for $\Im z > 0$,
$$
A(z) = \frac{1}{2 \pi i} \int_{-\infty}^{\infty} \l \frac{1}{t-z} - \frac{t}{1+t^2} \r \log g(t) dt
$$
is well defined. We let $H(z)$ be an entire function defined by $H(z) = e^{A(z)}$ for $\Im z >0$, $H(z) = G(z) e^{A(z)}$ for $\Im z <0$, and $H(x) = f(x) e^{i a x/2}$ for $x \in \RR$. 

Since $\Re A(x + i y) = \frac{1}{2 \pi} \int_{-\infty}^{\infty} \frac{y}{(t-x)^2+y^2} \log g(t) d t$ and $\log g(t) < C$ for some constant $C$, the function $H(z)$ has essentially zero exponential growth in the upper half-plane. On the other hand, in the lower half-plane it has the same exponential growth as $G$. Moreover, $H|_{\RR}$ decays superpolynomially. Thus, by Paley-Wiener theorem, the Fourier transform of $H|_{\RR}$ is supported in $[0, a]$.

Let $F(z) = e^{-i a z/2} H(z)$. Then $f = F|_{\RR}$, and its inverse Fourier transform $\tilde{f} = \CF^{-1} f$ is a bump function supported in $[-a/2,a/2]$.
\end{proof}

\begin{proof}[Proof of Theorem \ref{thm:InvertibilityofAmaj}]
Let $\hilb{K} = L^{2}(\RR) \otimes \CCC^2$ with closed subspaces $\hilb{K}_I = L^{2}(I) \otimes \CCC^2$ for bounded open subsets $I \subset \RR$. Any state $\omega$ on $\CAR(\hilb{K})$ generates a quasi-local algebra via $I \to \pi_{\omega}(\CAR(\hilb{K}_I))''$. A quasi-free state $\QFS{P}$ for $P = (1 + \sign(\hD) \otimes Z)/2$ generates $\Amaj \otimes \overline{\Amaj}$. Below we explicitly construct a Bogolyubov automorphism $\beta_V$ satisfying $\beta_V(\CAR(\hilb{K}_I)) \subset \CAR(\hilb{K}_{I^{+ l}})$ for some $l>0$ and such that the quasi-free state $\QFS{\tilde{P}} = \QFS{P} \beta_V$ generates a quasi-local algebra whose discretization gives $\sMat(\ZZ)$. The extension of $\beta_V$ provides a desired bounded spread isomorphism $\alpha: \Amajd \otimes \overline{\Amajd} \to \sMat(\ZZ)$.

Let $\rho$ be an even non-negative smooth function on $\RR$, such that $\tilde{\rho} = \CF^{-1} \rho \in \Bumpfns$ with  $\supp(\tilde{\rho}) \subset [-l/2,l/2]$\footnote{For example, one can take $\rho = (\CF \sigma)^2$ for
$$\sigma(x) = \begin{cases}
			e^{-\frac{1}{l^2-16 x^2}}, & |x|<l/4\\
            0, & |x|\geq l/4.
		 \end{cases}
$$         
}, and let $h(x) = \int_{-\infty}^{x} \rho(x') dx' \big/ \int_{-\infty}^{\infty} \rho(x') dx'$. Using Lemma \ref{lma:gfunction} for $g = 4 h(1-h)$, we have smooth functions
\beq
b_1 = 2h -1,
\eeq
\beq
b_2 = \sqrt{g} \cos(\theta),
\eeq
\beq
b_3 = \sqrt{g} \sin(\theta),
\eeq
such that $b_1^2 + b_2^2 + b_3^2 = 1$, $b_1$ and $b_3$ are odd, $b_2$ is even, and 
\beq \label{eq:vofk}
v(k) = \frac{1}{\sqrt{2}}\l 1 + i b_1(k) X + i b_2(k) Y + i b_3(k) Z \r = \frac{1}{\sqrt{2}}(1+i \sign(k) X) + \CO(|k|^{-\infty}).
\eeq
We set $V = v(\hD)^{-1}$. By construction, the Bogolyubov automorphism $\beta_V$ satisfies $\beta_V(\CAR(\hilb{K}_I)) \subset \CAR(\hilb{K}_{I^{+l}})$. It remains to show that $\QFS{\tilde{P}}$ for $\tilde{P} = V^{-1} P V$ generates a quasi-local algebra whose discretization gives $\sMat(\ZZ)$. 

We define functions $p_0(k) = (1+Y)/2$ and $\tilde{p}(k)$ via $\tilde{P} = \tilde{p}(\hD)$, $k \in \RR$. By Eq. (\ref{eq:vofk}), we have $\tilde{p}(k) = p_0(k) + \CO(|k|^{-\infty})$. Let $\Pi_I$ be the closed orthogonal projection to $\hilb{K}_I$ and $P_0 = p_0(\hD)$. Since $\int_{-\infty}^{\infty} dk \Tr|\tilde{p}(k)-p_0(k)| < \infty$, the operator $\Pi_I (\tilde{P} - P_0) \Pi_I$ is trace-class. By Powers–Størmer inequality \cite[Lemma 4.1]{powers1970free}, the operator $(\Pi_I \tilde{P} \Pi_I)^{1/2} - (\Pi_I P_0 \Pi_I)^{1/2}$ is Hilbert-Schmidt. Thus, by \cite[Theorem 1]{araki1971quasifree}, the restrictions of $\QFS{\tilde{P}}$ and $\QFS{P_0}$ to $\CAR(\hilb{K}_I)$ are quasi-equivalent. Since the restrictions of $\QFS{P_0}$ are pure, $\QFS{\tilde{P}}$ generates a quasi-local algebra of type $\rm{I}$ superfactors. Its discretization gives a quasi-local algebra isomorphic to $\sMat(\ZZ)$.

\end{proof}

\subsubsection{$O(N)$-action}

Let $N \in \NN$, and let $\rho:O(N) \to \End(\CCC^N)$ be the fundamental representation of $O(N)$. We denote the corresponding map from the Lie algebra $\mathfrak{so}(N)$ to $\End(\CCC^N)$ by the same letter. Choose an orthonormal basis in $\RR^N$ and let $i t^a \in \mathfrak{so}(N)$, $a = 1,2,...,N(N-1)/2$ be elements generating rotations in two-planes spanned by pairs of basis vectors, which satisfy $[t^a,t^b]=i f^{abc}t^c$, $\Tr(\rho(t^a) \rho(t^b)) = 2 \delta^{ab}$, $e^{2 \pi i \rho(t^a)} = 1$ for structure constants $f^{abc}$ of $\mathfrak{so}(N)$.

The algebra $\CAR(L^2(\RR) \otimes \CCC^N)$ carries an action of $O(N)$ via Bogolyubov automorphisms $g \to \beta_{1 \otimes \rho(g)}$. Since it preserves the quasi-free state $\QFS{P}$ for $P = (1 + \sign(\hD))/2 \otimes 1$, this action extends to the action of $O(N)$ on the quasi-local algebra $(\Amaj)^{\otimes N}$. 

For $f \in C_c^{\infty}(\RR)$, let $Q^a_f = M_f \otimes \rho(t^a)$. To check that $i Q^a_f \in \mathfrak{o}_P(\hilb{K})$, it is enough to check that $P Q^a_f (1-P)$ is Hilbert-Schmidt. We have
\begin{multline}
\Tr(P Q^a_f (1-P) Q^b_f P) = 2 \delta^{ab} \int_{0}^{\infty} \frac{dk_1}{2 \pi} \int_{-\infty}^{0} \frac{dk_2}{2 \pi } \tilde{f}(k_1 - k_2) \tilde{f}(k_2 - k_1) = \\ = 2 \delta^{ab} \int_{0}^{\infty} \frac{dq}{2 \pi} \frac{q}{2 \pi} \tilde{f}(q) \tilde{f}(-q) < \infty
\end{multline}
Thus, the derivations $i \der{Q}^a_f = \der{D}_{i Q^a_f}$ can be implemented in the GNS representation of $\QFS{P}$ by (unbounded) skew self-adjoint operators $i \hat{Q}^{a}_f$.

Using Eq. (5.4-5.5) in \cite{araki1988schwinger}, we have
$$
[\hat{Q}^a_f , \hat{Q}^b_g] = i f^{abc} \hat{Q}^c_{fg} + \Tr(P[P,Q^a_f] [P,Q^b_g])
$$
and
\begin{multline}
\Tr(P[P,Q^a_f] [P,Q^b_g]) = \delta^{ab} \Tr((1-P) Q^a_f P Q^a_g (1-P) - P Q^a_f (1-P) Q^a_g P ) = \\ = \delta^{ab} \int_{0}^{\infty} \frac{dk_1}{2 \pi} \int_{-\infty}^{0} \frac{dk_2}{2 \pi } \l \tilde{f}(k_1-k_2) \tilde{g}(k_2 - k_1) - \tilde{f}(k_2-k_1) \tilde{g}(k_1 - k_2) \r = \frac{\delta^{ab}}{2 \pi i} \int dx f'(x) g(x),
\end{multline}
where $\tilde{f} = \CF f$, $\tilde{g} = \CF g$. Hence,
\beq \label{eq:schwingerterm}
[\hat{Q}^a_f , \hat{Q}^b_g] = i f^{abc} \hat{Q}^c_{fg} + \frac{\delta^{ab}}{2 \pi i} \int_{-\infty}^{\infty} f'(x) g(x) d x
\eeq
for any $f,g \in C_c^{\infty}(\RR)$. Thus, the derivations $\der{Q}^a$ together with their smooth restrictions $\der{Q}^a_f$ extend to derivations of the quasi-local algebra $(\Amaj)^{\otimes N}$ and their exponents define an action of $O(N)$.

\subsection{$E_8$-lattice chiral bosons}

We now define a bosonic quasi-local algebras via bosonization procedure that corresponds to the conformal net of chiral lattice bosons based on $E_8$-lattice.

Let $\hilb{K} =  L^2(\RR) \otimes \CCC^2 \otimes \CCC^8$, $Q = 1 \otimes Y \otimes 1$, and $\ql{A}$ be the quasi-local algebra $(\Amaj)^{\otimes 16}$ generated by the state $(\varphi_P)^{\otimes 16}$ for $P$ from Eq. (\ref{eq:LeftMovingMajProj}). The derivation $i\der{Q} = \der{D}_{iQ}$ of $\CAR(\hilb{K})$ induces a $U(1)$-action on the quasi-local algebra $\cstar{A}$. For a bounded smooth function $f$, let $i \der{Q}_f = \der{D}_{i Q_f}$ for $Q_f = M_f \otimes Y \otimes 1$. 

We say that $f \in C^{\infty}(\RR)$ is a smooth step function if $\supp(f') \subset I$ for an interval $I \subset \RR$, $f(x < I) = 0$ and $f(x > I) = 1$. For smooth step functions $f$, we let $\sigma_{f}$ be automorphisms $e^{\pi i \der{Q}_{f}}$ of $\cstar{A}$. Note that, by explicit computation, $(\sigma_{f})^2$ is an inner automorphism $\Ad_{u_{f}}$ for some $u_{f} \in (\Amaj^{\otimes 16})_I$, $I \supset \supp(f')$, and we have $\sigma_{f}(u_{f}) = u_{f}$.

Similarly, for smooth step functions $f_1,f_2$ such that $\supp(f_1') < \supp(f_2')$, the automorphism $\sigma_{f_1} \sigma^{-1}_{f_2}$ is inner $\Ad_{w_{f_1-f_2}}$ for some $w_{f_1-f_2} \in \cstar{A}$ which can be chosen to be such that $w_{f_1-f_2}^2 = u_{f_1} u_{f_2}^{-1}$. Moreover, for a triple of smooth step functions $f_1,f_2,f_3$ such that $\supp(f_1') < \supp(f_2') < \supp(f_3')$, using Eq. (\ref{eq:schwingerterm}), we have $w_{f_1-f_2} w_{f_2-f_3} = w_{f_2-f_3} w_{f_1-f_2}$\footnote{Note that the fact that we consider a multiple of 16 copies of a left-moving Majorana fermion is important for the commutativity of $w_{f_1-f_2}$, $w_{f_2-f_3}$.}.

\begin{definition}
A quasi-local algebra $\AEe$ over $\RR$ is defined in the following way. Let $\omega_0 = (\varphi_P)^{\otimes 16}$ for $P$ from Eq. (\ref{eq:LeftMovingMajProj}), $\omega_1 = (\varphi_P)^{\otimes 16} \sigma_{f_0}$ be (locally normal) states on $(\Amaj)^{\otimes 16}$ for a smooth step function $f_0$, and let $\pi = \pi_{\omega_0} \oplus \pi_{\omega_1}$. To open intervals $I \subset \RR$ we assign the von Neumann algebra $(\AEe)_I$ generated by the fixed point set $\pi\l(\Amaj^{\otimes 16})^{\ZZ/2}_I\r$ of the fermionic parity automorphism $\Theta$ and unitary elements $S_{f}$ for smooth step functions $f$, $\supp(f') \subset I$ with relations $S_{f} \pi(a) = \pi(\sigma_{f}(a)) S_{f}$ for $a \in (\Amaj^{\otimes 16})$, $S_{f}^2 = \pi(u_{f})$, $S_{f_1} S^{-1}_{f_2} = \pi(w_{f_1-f_2})$ for $\supp(f_1') < \supp(f_2')$. To a bounded subset $U$, it assigns a von Neumann algebra $(\AEe)_U$ generated by $\{(\AEe)_{I_a}\}$, where $\{I_a\}$ is the collection of disjoint open intervals such that $\bigcup_a I_a = \interior{U}$. The discrete version $\ql{A}^{\text{dis.}}_{E_8}$ is defined by $j_* \ql{A}_{E_8}$ for a map $j:\RR \to \ZZ$, $x \to \lfloor x + 1/2 \rfloor$.
\end{definition}

\begin{prop} \label{prop:AE8isinvertible}
The quasi-local algebra $\ql{A}^{\text{dis.}}_{E_8}$ is invertible.    
\end{prop}

\begin{proof}
Let $\Theta_n$ be the restrictions of the fermionic parity automorphism $\Theta$ of $\sMat(\ZZ)$ to $\ZZ_{<n}$. We can describe the quasi-local algebra $\Mat(\ZZ)$ as the algebra generated by even elements of $\sMat(\ZZ)$ and elements of the form $a S_{n}$ for odd elements $a \in \sMat(\ZZ)$ satisfying the relations $a S_n b = a \Theta_n(b) S_n$, $a S_n b S_n = a \Theta_n(b)$ (in the appropriate representation).

We can choose a bounded spread isomorphism $\alpha: (\Amajd)^{\otimes 16} \otimes (\overline{\Amajd})^{\otimes 16} \to \sMat(\ZZ) \otimes \sMat(\ZZ)$ from the proof of Theorem \ref{thm:InvertibilityofAmaj} such that for a certain step function $f$ with $\supp(f') \subset (n-1/2+\eps,n+1/2-\eps)$, $n \in \ZZ$, $0 <\eps <1/2$ the automorphisms $\alpha \sigma_f \alpha^{-1}$, $\alpha \overline{\sigma}_f \alpha^{-1}$ have the form $\Theta^{(L)}_n$, $\Theta^{(R)}_n$, respectively, where $\Theta^{(L,R)}$ are fermionic parity automorphisms acting on the first and the second factor of $\sMat(\ZZ) \otimes \sMat(\ZZ)$, $\Theta^{(L,R)}_n$ are their restrictions to $\ZZ_{<n}$.

The relations above induce a map $\AEed \otimes \overline{\AEed} \to \Mat(\ZZ) \otimes \Mat(\ZZ)$ that is naturally a bounded spread isomorphism.
\end{proof}

\begin{prop}    \label{prop:E8is16copiesOfMaj}
We have $[\AEed] = 16 [\Amajd]$ in $\sBrgr(\ZZ)$.    
\end{prop}
\begin{proof}
Given a bounded spread isomorphism $\alpha: (\Amajd)^{\otimes 16} \otimes (\overline{\Amajd})^{\otimes 16} \to \sMat(\ZZ) \otimes \sMat(\ZZ)$ and a relation between $\sMat(\ZZ)$ and $\Mat(\ZZ)$ as in the proof of Proposition \ref{prop:AE8isinvertible}, we get a bounded spread isomorphism $\AEed \otimes (\overline{\Amajd})^{\otimes 16} \to \Mat(\ZZ) \otimes \sMat(\ZZ)$. Hence, by Theorem \ref{thm:InvertibilityofAmaj}, we obtain $[\AEed] = 16 [\Amajd]$.

\end{proof}

\section{Invertible quasi-local algebras over $\ZZ$ from bosonic holomorphic CFTs} \label{sec:bosonicholomorphic}

In both of the above examples, the quasi-local algebras arise from holomorphic conformal field theories (CFTs). These kinds of CFTs can be axiomatized with conformal nets on $S^{1}$. Here, we consider only CFTs with \textit{bosonic} degrees of freedom, and we will follow the definitions and hypothesis of \cite{MorinelliTanimotoWeiner2018ConformalCovarianceSplit, LongoXu2004TopologicalSectorsDichotomy, KawahigashiLongoMuger2001MultiIntervalSubfactors}.

\begin{definition} A conformal net on $S^{1}$ consists of the following data:

\begin{enumerate}
\item 
a Hilbert space $\hilb{H}$;
\item 
an assignment of a von Neumann algebra $\ql{A}_{I}\subseteq B(\hilb{H})$ to each open, connected interval $I \subset S^{1}$;
\item 
a strongly continuous unitary projective representation $\pi:\text{Diff}(S^{1})\rightarrow U(\hilb{H})$.\footnote{where we use the standard Fréchet–Lie group topology on $\text{Diff}(S^{1})$}
\end{enumerate}

\bigskip

\noindent This data is subject to the following conditions:

\bigskip

\begin{enumerate}
\item 
(Isotony.) $I\subseteq J$ implies $\ql{A}_{I}\subseteq \ql{A}_{J}$.

\medskip

\item 
(Locality.) $I\cap J=\varnothing$ implies $[\ql{A}_{I},\ql{A}_{J}]=0$.

\medskip

\item 
(Vacuum.) There exists a unique (up to phase) vector $\Omega\in H$ which is invariant under the action of the Mobius subgroup $\text{PSL}(2,\mathbb{R})\subseteq \text{Diff}(S^{1})$, and $\Omega$ is cyclic for $\vee_{I} \ql{A}_{I}$.

\medskip

\item 
(Covariance.)
For $f\in \text{Diff}(S^{1})$, $\pi(f)\ql{A}_{I}\pi(f^{-1})=\ql{A}_{f(I)}$.
\medskip

\item 
(Positive energy.) The generator $L_{0}$ of the rotation subgroup $\text{U}(1)\subseteq \text{Diff}(S^{1})$ is positive.
\end{enumerate}
\end{definition}

There are some important consequences of this axioms (see \cite{GabbianiFroehlich1993OperatorAlgebrasCFT,FredenhagenJoerss1996ConformalHaagKastlerNets,GuidoLongo1996ConformalSpinStatistics}):

\begin{itemize}
\item 
(Haag duality.)  $(\ql{A}_I)^{\prime}=\ql{A}_{I^{\prime}}$, where $I^{\prime}$ denotes the interior of the complement of $I$
\item
(Factoriality) The local algebras $\ql{A}_{I}$ are type $\rm{III}_{1}$ factors.
\item 
(Reeh–Schlieder) $\Omega$ is cyclic and separating for each $\ql{A}_{I}$
\item 
(Split property) \cite{MorinelliTanimotoWeiner2018ConformalCovarianceSplit} For any $I\subset J$ with distinct endpoints, there exists a type $\rm{I}$ factor $\cstar{M}$ with $\ql{A}_{I}\subseteq \cstar{M} \subseteq \ql{A}_{J}$.
\end{itemize}

Associated to a conformal net is the ``de-compactification", which we describe as a quasi-local algebra over $\mathbb{R}$. To define this, given a conformal net, consider the stereographic projection $\sigma$ of $S^{1}-\{1\}$ onto the imaginary axis, analytically represented by the Cayley transform $z\mapsto \frac{1+z}{1-z}$. Identifying the imaginary axis with $\mathbb{R}$ in the usual way, we can express $\sigma$ as $S^{1}\rightarrow \mathbb{R}$ in radians via $\sigma(\theta):=\text{cot}\left(\frac{\theta}{2}\right)$.

This gives a bijection between the connected, open intervals in $S^{1}$ localized away from $1$ and the collection $\mathcal{I}_{\mathbb{R}}$ of all bounded, connected open intervals in $\mathbb{R}$, which preserves inclusion of intervals and disjointness. Then for any $I\in \mathcal{I}_{\mathbb{R}}$, we define $(\ql{A}_{0})_{I}:=\ql{A}_{\sigma^{-1}({I})},$ and set $\ql{A}_0$ to be the norm completion of $\bigcup_{I \in \mathcal{I}_{\mathbb{R}}} (\ql{A}_0)_I$. For a bounded subset $U \subset \RR$, we let $(\ql{A}_0)_U$ be $\bigvee_{a} (\ql{A}_0)_{I_a}$, where $\{I_a\}$ is a collection of disjoint open intervals such that $\bigcup_a I_a = \interior{U}$. This assignment defines a quasi-local algebra $\ql{A}_0$ over $\RR$.

One of the reasons this is so useful is that the representation theory of $\ql{A}$ is encoded via $*$-endomorphisms of $\ql{A}_{0}$. Recall that the representations of $\ql{A}$ (for example, as in \cite[Appendix B]{KawahigashiLongoMuger2001MultiIntervalSubfactors}) form a braided $W^*$-category called the category of \textit{superselection sectors}. Doplicher-Haag-Roberts (DHR) theory shows that this category is equivalent as a braided $W^*$-tensor category to the localized, transportable $*$-endomorphisms on $\ql{A}_{0}$. 

\begin{definition}
A (locally normal) $*$-endomorphism $\rho:\ql{A}_{0}\rightarrow \ql{A}_{0}$ is 
\begin{itemize}
\item 
\textit{localized} in an interval $I \in \mathcal{I}_{\mathbb{R}}$ if for all $J \subset \mathbb{R}\setminus I$, $\rho|_{(\ql{A}_{0})_{J}}=\text{Id}|_{(\ql{A}_{0})_{J}}.$
\item
\textit{transportable} if for any $J \in \mathcal{I}_{\mathbb{R}}$, there is a unitary $u\in \ql{A}_{0}$ such that $\text{Ad}_u \circ \rho$ is localized in $J$.
\end{itemize}
\end{definition}

The collection of all localized transportable $*$-endomorphisms is called $\DHRcat(\ql{A}_{0})$, and is a braided $W^*$-tensor category. If $\rho\in \DHRcat(\ql{A}_{0})$, then we denote by $\rho_{I}$ a representative of $\rho$ localized in $I$. By Haag duality, $\rho_{I}$ defines a $*$-endomorphism of $(\ql{A}_{0})_{I}$, and is unique up to unitary conjugacy. We can think of an endomoprhism of $\ql{A}_{0}$ as defining an $\ql{A}_{0}-\ql{A}_{0}$ bimodule $_{\rho}\ql{A}_{0}$ with right action defined by multiplication, and left action implemented by $\rho$, composed with multiplication. A Q-system in a unitary tensor category is (essentially) a $C^*$-Frobenius algebra object \cite{BischoffKawahigashiLongoRehren2015TensorCategories, ChenHernandezPalomaresJonesPenneys2022}. In a unitary braided tensor category, a commutative Q-system is a Q-system that is invariant under pre-composition with the braiding. Any commutative Q-system $L$ in $\DHRcat(\ql{A}_{0})$ defines a new quasi-local algebra $\ql{A}_{0}\rtimes L$, with local structure $(\ql{A}_{0}\rtimes L)_{I}:=(\ql{A}_{0})'_{I^{c}}\cap \ql{A}_{0}\rtimes L\cong (\ql{A}_{0})_{I}\rtimes L_{I}$. The inclusion $\ql{A}_{0}\subseteq \ql{A}_{0}\rtimes L$ is then a local inclusion \cite{LongoRehren1995NetsOfSubfactors, KawahigashiLongoMuger2001MultiIntervalSubfactors, BischoffKawahigashiLongoRehren2015TensorCategories, evans2026operatoralgebraicapproachfusion}.

We are interested in the case of \textit{rational} conformal nets, i.e. every irreducible $*$-endomorphism is dualizable (i.e. it has finite index), and there are finitely many isomorphism classes of these. In this case, the finite-index endomorphisms form a unitary modular tensor category (UMTC). Under the conformal covariance assumption we have given by \cite{MorinelliTanimotoWeiner2018ConformalCovarianceSplit, LongoXu2004TopologicalSectorsDichotomy}, rationality implies that $\ql{A}$ (hence $\ql{A}_{0}$) are strongly additive and have finite $\mu$ index, hence are completely rational in the sense of \cite{KawahigashiLongoMuger2001MultiIntervalSubfactors}. Thus under this version of the rationality assumption, we can freely use the results of \cite{KawahigashiLongoMuger2001MultiIntervalSubfactors} which apply to completely rational nets.

Conformal nets as defined above capture $\textit{chiral}$ conformal field theories. Associated to a chiral CFT is the canonically associated diagonal full CFT. This is defined on $1+1$D Minkowski space as in \cite{bischoff2015characterization2drationallocal}.
We focus here on the time-slice net, which is a net of von Neumann algebras over $\mathbb{R}$ given by slicing the $t=0$ hypersurface with double light-cones along their diagonals to define algebras for intervals along the spatial axis. This net, together with all of the symmetry actions, in fact remembers the full structure of the CFT (indeed, this is a fairly generic feature of algebraic quantum field theory that is usually assumed axiomatically as the time-slice axiom).

To define the quasi-local algebra of time slice net, we use the commutative $Q$-system realization $$\ql{B}:=(\ql{A}_{0}\otimes \ql{A}^{op}_{0})\rtimes L,$$
where $L$ denotes the Longo-Rehren $Q$-system \cite{KawahigashiLongoMuger2001MultiIntervalSubfactors, LongoRehren1995NetsOfSubfactors} 
$$\displaystyle L:=\bigoplus_{\rho\in \text{Irr}(\DHRcat(\ql{A}_{0}))} \rho_{I}\otimes \rho^{op},$$
with canonical multiplication.

The local structure is slightly tweaked for a spatial interval. Let $\widehat{I}$ be the intersection of the diamond $I\times I$ (in light-cone coordinates) with the line $t=0$.

$$\begin{tikzpicture}[scale=0.7,>=stealth]

  \def\a{1.2}
  \def\b{3.0}
  \pgfmathsetmacro{\m}{(\a+\b)/2}
  \pgfmathsetmacro{\h}{(\b-\a)/2}

  \draw[->] (-0.6,0) -- (4.2,0) node[right] {\(x\)};
  \draw[->] (0,-2.2) -- (0,2.4) node[above] {\(t\)};
  \node[below] at (3.95,0) {\(t=0\)};

  \draw[dashed] (-0.3,-0.3) -- (3.6,3.6);
  \draw[dashed] (-0.3,0.3) -- (3.6,-3.6);

  \node[above right] at (3.5,3.1) {\(u=x+t\)};
  \node[below right] at (3.5,-3.1) {\(v=x-t\)};

  \draw[thick,gray]
    (\a,0) -- (\m,\h) -- (\b,0) -- (\m,-\h) -- cycle;

  \draw[blue,very thick] (\a,0) -- (\b,0);
  \node[above] at ({(\a+\b)/2},0) {\(\widehat{I}\)};

  \draw[red,very thick] (\a/2,\a/2) -- (\b/2,\b/2);
  \draw[red,very thick] (\a/2,-\a/2) -- (\b/2,-\b/2);

  \node[red,above left] at ({(\a+\b)/4},{(\a+\b)/4}) {\(I\)};
  \node[red,below left] at ({(\a+\b)/4},{-(\a+\b)/4}) {\(I\)};

  \fill[blue] (\a,0) circle (1.2pt);
  \fill[blue] (\b,0) circle (1.2pt);

\end{tikzpicture}$$

\noindent Then 
$$\ql{B}_{\widehat{I}}:=(\ql{A}_{0}\otimes \ql{A}^{op}_{0})_{I}\rtimes L_{I}.$$

\noindent For convenience, we will abuse notation and write $\ql{B}_{I}$ in place of $\ql{B}_{\widehat{I}}$, which simply has the effect of scaling by $\sqrt{2}$, i.e. the interval $\widehat{I}=\sqrt{2}I$.

Before proving our main result, we have the following lemma, which gives a representation theoretic characterization of bosonic quasi-local algebras which are bounded spread isomorphic to $\Mat(\ZZ)$. For a quasi-local algebra $\ql{B}$ over $\RR$, we let $\ql{B}^{\text{dis.}}$ be a quasi-local algebra over $\ZZ$ defined by $j_* \ql{B}$, where $j$ is a map $j:\RR \to \ZZ$, $x \to \lfloor x + 1/2 \rfloor$.

\begin{lemma}   \label{lma:trivialtest}
Suppose that $\ql{B}$ is a quasi-local algebra over $\mathbb{R}$ and that
$\ql{B}_{I}$ is infinite-dimensional for every nonempty bounded interval
$I\subseteq \mathbb{R}$. Then the discretization
$\ql{B}^{\mathrm{dis.}}$ is bounded-spread isomorphic to
$\Mat(\mathbb{Z})$ if and only if there exist a faithful, locally normal
representation
\begin{equation*}
\pi_{0}\colon \ql{B}\longrightarrow B(\mathcal{H}_{0})
\end{equation*}
on a separable Hilbert space and a constant $R\geq 0$ such that the
following conditions hold:

\medskip

\begin{enumerate}
\item (Weak Haag duality). For every $a<b$, $\left(
\pi_{0}\bigl(\ql{B}{(-\infty,a)}\bigr)
\vee
\pi_{0}\bigl(\ql{B}{(b,\infty)}\bigr)
\right)'
\subseteq
\pi_{0}\bigl(\ql{B}_{(a-R,b+R)}\bigr).$

\bigskip

\item (Split property at infinity). For every $n\in\mathbb{Z}$, there is a type $\rm{I}$ factor
$\mathcal{M}_{n}\subseteq B(\mathcal{H}_{0})$ such that
\begin{equation*}
\pi_{0}\bigl(\ql{B}{(-\infty,n)}\bigr)''
\subseteq
\mathcal{M}_{n}
\subseteq
\pi_{0}\bigl(\ql{B}_{(-\infty,n+R)}\bigr)''.
\end{equation*}
\end{enumerate}
\end{lemma}

\begin{proof}
Suppose first that there is a bounded-spread isomorphism
\begin{equation*}
\Phi\colon
\Mat(\mathbb{Z})
\longrightarrow
\ql{B}^{\mathrm{dis.}}.
\end{equation*}
Choose a pure normal state on each single-site algebra
$B(\ell^{2}(\mathbb{N}))$, and let
\begin{equation*}
\rho\colon
\Mat(\mathbb{Z})
\longrightarrow
B(\mathcal{H}_{0})
\end{equation*}
be the corresponding infinite tensor-product representation. This
representation is faithful and irreducible, and $\mathcal{H}_{0}$ is
separable. Moreover, its half-chain von Neumann algebras are type $\rm{I}$
factors and satisfy exact Haag duality.

Define $\pi_{0}:=\rho\circ\Phi^{-1}.$ Let $s-1$ be a common spread bound for $\Phi$ and $\Phi^{-1}$. Then

$$
\Phi\bigl(\Mat(\mathbb{Z})_{(-\infty,a-s)}\bigr)
\subseteq
\ql{B}_{(-\infty,a)}\ \text{and}\ 
\Phi^{-1}\bigl(\ql{B}_{(-\infty,a)}\bigr)
\subseteq
\Mat(\mathbb{Z})_{(-\infty,a+s)},
$$

\medskip

\noindent with analogous inclusions for right half-lines. 

\noindent For each $n\in\mathbb{Z}$, set
\begin{equation*}
\mathcal{M}_{n}
:=
\rho\bigl(\Mat(\mathbb{Z})_{(-\infty,n+s)}\bigr)''.
\end{equation*}
Then $\mathcal{M}_{n}$ is a type $\rm{I}$ factor and
\begin{equation*}
\pi_{0}\bigl(\ql{B}_{(-\infty,n)}\bigr)''
\subseteq
\mathcal{M}_{n}
\subseteq
\pi_{0}\bigl(\ql{B}_{(-\infty,n+2s)}\bigr)''.
\end{equation*}
Thus the one-sided split property holds with $R=2s$. Similarly, exact Haag duality in the tensor-product representation gives

\begin{align*}\left(
\pi_{0}\bigl(\ql{B}{(-\infty,a)}\bigr)
\vee
\pi_{0}\bigl(\ql{B}{(b,\infty)}\bigr)
\right)'&=
\left(
\pi_{0}\bigl(\ql{B}{(-\infty,a)}\bigr)''
\vee
\pi_{0}\bigl(\ql{B}{(b,\infty)}\bigr)''
\right)'\\
&\subseteq
\left(
\rho\bigl(\Mat(\mathbb{Z})_{(-\infty,a-s)}\bigr)
\vee
\rho\bigl(\Mat(\mathbb{Z})_{(b+s,\infty)}\bigr)
\right)'\\
&=
\rho\bigl(\Mat(\mathbb{Z})_{(a-s,b+s)}\bigr)''
\subseteq
\pi_{0}\bigl(\ql{B}_{(a-2s,b+2s)}\bigr)''.
\end{align*}
Hence weak Haag duality also holds with $R=2s$.

Conversely, suppose that $\pi_{0}$ and $R$ satisfy the two stated
conditions. We identify $\ql{B}$ with its image under $\pi_{0}$. Choose an integer $L>2R$, and define
\begin{equation*}
\mathcal{N}_{k}:=\mathcal{M}_{kL},
\qquad
k\in\mathbb{Z}.
\end{equation*}
The factors $\mathcal{N}_{k}$ are nested. Indeed,
\begin{equation*}
\mathcal{N}_{k}
\subseteq
\ql{B}_{(-\infty,kL+R)}''
\subseteq
\ql{B}_{(-\infty,(k+1)L)}''
\subseteq
\mathcal{N}_{k+1}.
\end{equation*}

\noindent For each $k\in\mathbb{Z}$, let
\begin{equation*}
\mathcal{Q}_{k}
:=
\mathcal{N}'_{k}\cap\mathcal{N}_{k+1}.
\end{equation*}
Since
$\mathcal{N}_{k}\subseteq\mathcal{N}_{k+1}$
are type $\rm{I}$ factors, $\mathcal{Q}_{k}$ is again a type $\rm{I}$ factor, and
multiplication induces a (spatial) tensor product decomposition
\begin{equation*}
\mathcal{N}_{k+1}
\cong
\mathcal{N}_{k}\otimes \mathcal{Q}_{k}.
\end{equation*}

\noindent Locality of the quasi-local algebra gives
\begin{equation*}
\ql{B}_{(kL+R,(k+1)L)}
\subseteq
\mathcal{Q}_{k}.
\end{equation*}
Indeed, the algebra on the left commutes with $\mathcal{N}_{k}
\subseteq
\ql{B}_{(-\infty,kL+R)}''$ and is contained in $\ql{B}_{(-\infty,(k+1)L)}''
\subseteq
\mathcal{N}_{k+1}$. Because $\ql{B}_{(kL+R,(k+1)L)}$ is infinite-dimensional and the
representation is faithful, $\mathcal{Q}_{k}$ is infinite-dimensional.
It follows that
\begin{equation*}
\mathcal{Q}_{k}
\cong
B(\ell^{2}(\mathbb{N})).
\end{equation*}

\noindent Define $$\mathcal{D}_{[i,j]}:= \mathcal{N}'_{i}\cap \mathcal{N}_{j+1}\cong
\mathcal{Q}_{i}\otimes
\mathcal{Q}_{i+1}
\otimes
\cdots
\otimes
\mathcal{Q}_{j}.$$

\medskip

\noindent Now notice that $\ql{B}_{(iL+R,(j+1)L)}
\subseteq
\mathcal{D}_{[i,j]}$. Indeed, the algebra on the left commutes with $\mathcal{N}_{i}$ and is
contained in $\mathcal{N}_{j+1}$.

For the reverse inclusion, an element of $\mathcal{D}_{[i,j]}$ commutes
with
\begin{equation*}
\ql{B}_{(-\infty,iL)}''
\subseteq
\mathcal{N}_{i}.
\end{equation*}
It also belongs to
\begin{equation*}
\mathcal{N}_{j+1}
\subseteq
\ql{B}_{(-\infty,(j+1)L+R)}'',
\end{equation*}
and hence, by locality, commutes with
\begin{equation*}
\ql{B}_{((j+1)L+R,\infty)}''.
\end{equation*}
Weak Haag duality therefore gives
\begin{equation*}
\mathcal{D}_{[i,j]}
\subseteq
\ql{B}_{(iL-R,(j+1)L+2R)}.
\end{equation*}
Thus
$$
\ql{B}_{(iL+R,(j+1)L)}
\subseteq
\mathcal{D}_{[i,j]}
\subseteq
\ql{B}_{(iL-R,(j+1)L+2R)}.
$$

\noindent If we set $\ql{B}^{L}$ to be the $L$-block coarse-grained quasi-local algebra on $\mathbb{Z}$ with $\ql{B}^{L}_{[i,j]}:=\ql{B}^{\text{dis.}}_{[iL,(j+1)L-1]}$, then the tensor product quasi-local algebra $\mathcal{D}=\text{colim}_{I} \mathcal{D}_{I}$ is bounded spread isomorphic to $\mathcal{B}^{L}$ by Lemma \ref{bsinverse}. Thus if we take the identification $\mathcal{D}_{[k,k]}=\mathcal{Q}_{k}\cong B(\ell^{2}(\mathbb{N}))^{\otimes L}$ we can similarly view $\ql{D}$ as an $L$-block coarse graining of $\Mat(\mathbb{Z})$. Denoting more precisely, if we define the $L$ block coarse-graining $\Mat(\mathbb{Z})^{L}_{[i,j]}=\Mat(\mathbb{Z})_{[iL, (j+1)L-1]}$, then we have a bounded spread isomorphism $\mathcal{D}\cong \Mat(\mathbb{Z})^{L}$. Note that the identity map on $C^*$-algebra gives a map of quasi-local algebra $\Mat(\mathbb{Z})^{L}\cong \Mat(\mathbb{Z})$ which is not bounded spread, but scales the spread by $L$. We call this map the coarse-graining map. Composing the coarse-graining map $\Mat(\mathbb{Z})\cong \Mat(\mathbb{Z})^{L}\cong \mathcal{D}$ with the bounded spread isomorphism $\mathcal{D}\cong \mathcal{B}^{L}$ and applying the inverse coarse-graining map $\ql{B}^{L}\cong \ql{B}^{\text{dis.}}$, we obtain a bounded spread isomorphism $\Mat(\mathbb{Z})\cong \ql{B}^{\text{dis.}}$, as desired.

\end{proof}

\begin{theorem}\label{thm:latticedegreesoffreedom} Let $\ql{A}$ be a rational conformal net, and $\ql{B}$ the quasi-local algebra of the time-slice net of the full diagonal CFT (with conventions as above). Then the discretization $\ql{B}^{\text{dis.}}$ is bounded spread isomorphic to $\Mat(\mathbb{Z})$.
\end{theorem}

\begin{proof}
Let $I^{+}\subseteq S^{1}$ denote the interval (in radians) $(0,\pi)$. Define $\ql{A}^{+} := C^{*}(\{\ql{A}_{I}\ :\ I\Subset I^{+}\})$. Now, let us reconsider $\ql{A}_{0}$ as a subnet of $\ql{A}$ again. Then $\ql{A}_{0}$ is the $C^*$-algebra generated by $\ql{A}_{I}$, where $I$ is separated from $\{1\}$. We claim that there is an isomorphism of quasi-local $C^*$-algebras $\gamma: \ql{A}_{0}\cong \ql{A}^{+}$ such that, for intervals in $S^{1}$ localized away from $1$, $\gamma((\ql{A}_{0})_{I})=\ql{A}^{+}_{\sqrt{I}}$, where $\sqrt{I}$ is literally the image of the interval under the principal branch of the complex square root function (which is well-defined since intervals are localized away from $1$). Note that there is no single diffeomorphism we can use to implement this covariantly. Nevertheless, we can construct such a $\gamma$ as follows.

Choose a sequence $\omega_{n}$ increasing (in radians) to $2\pi$, for example set $\omega_{n}:=\frac{n }{n+1}$. Then choose $\phi_{n}:[0,1]\rightarrow [0,1]$ to be a smooth functions such that $\phi_{n}|_{[0,\omega_{n}]}(\theta)=\frac{\theta}{2}$, and $\phi_{n}(1)=0$. Then $\Phi_{n}(e^{2 \pi i \theta}):=e^{2\pi i \phi_{n}(\theta)}$ is a diffeomorphism of $S^{1}$, and by covariance we have $U_{\Phi_{n}} (\ql{A}_{(0,\omega_n)})U^{\dagger}_{\Phi_n}=\ql{A}_{(0,\sqrt{\omega_{n}})}$.

Thus $\gamma_{n}:=\text{Ad}(U_{\Phi_n})$ gives an isomorphism of quasi-local algebra $$\gamma_{n}:\ql{A}_{(0,\omega_n)}\cong \ql{A}_{(0,\sqrt{\omega_n})}$$

Furthermore $\gamma_{n+m}|_{\ql{A}_{(0,\omega_n)}}=\gamma_{n}$ since the diffeomorphism group is locally implemented.
But $\displaystyle \ql{A}_{0}=C^{*}\{ \ql{A}_{0,\omega_n} : n\in \mathbb{N}\}=\text{colim}_{n} \ \ql{A}_{(0,\omega_n)}$, where colimit is taken in the category of $C^*$-algebras. Thus the $\gamma_{n}$ define an inductive system, hence a $C^*$-algebra homomorphism $\gamma:\ql{A}_{0}\cong \ql{A}_{(0,\pi)}=\ql{A}^{+}$.

Now, let $(\pi_{0},\mathcal{H}_{0})$ denote the vacuum representation of the chiral net $\ql{A}$, and pick an identification $\mathcal{H}_{0}\cong L^{2}(\ql{A}_{I^{+}})$, so that for intervals $I\subseteq I^{+}$, modular conjugation gives an identification $J_{I^{+}}\ql{A}_{I}J_{I^{+}}= \ql{A}_{\overline{I}}$, where $\overline{I}\subseteq I^{-}$ is the reflection about the real axis in $S^{1}$.

$$\begin{tikzpicture}[scale=1.6,>=stealth]

  \draw[->] (-1.35,0) -- (1.35,0) node[right] {\(\Re\)};
  \draw[->] (0,-1.25) -- (0,1.25) node[above] {\(\Im\)};

  \draw[gray!60,thick] (0,0) circle (1);

  \draw[blue,very thick] ({cos(35)},{sin(35)}) arc[start angle=35,end angle=125,radius=1];
  \node[blue] at (-0.3,1.2) {\(I\)};

  \draw[red,very thick] ({cos(35)},{-sin(35)}) arc[start angle=-35,end angle=-125,radius=1];
  \node[red] at (-0.3,-1.2) {\(\overline{I}\)};

\end{tikzpicture}$$

\noindent Then the map $a\rightarrow J_{I^{+}}\gamma(a)^{*}J_{I^{+}}$ gives an identification $\ql{A}^{op}_{0}\cong \ql{A}_{I^{-}}$, which sends the algebra $(\ql{A}^{op}_{0})_{K}$ assigned to a bounded interval $K\subseteq \mathbb{R}$ to $\ql{A}_{\sigma^{-1}(K)^{\frac{1}{2}}}$.

By the split property, since any bounded interval $K\subseteq \mathbb{R}$ has $\sigma^{-1}(K)$ positive distance from $\{-1,1\}$, the map 

$\widetilde{\pi}_{0}:\ql{A}_{0}\otimes \ql{A}^{op}_{0}\rightarrow B(\mathcal{H}_{0})$, defined by

$$\widetilde{\pi}_{0}(a\otimes b^{op}):= \pi_{0}(\gamma(a)J_{I^{+}}\gamma(b)^{*}J_{I^{+}})$$

\noindent give a locally normal injection. Let $\beta: \mathbb{R}\rightarrow I^{+}$ be the homeomorphism $\sqrt{\sigma^{-1}(\cdot)}$. Then $\widetilde{\pi_{0}}$ and restricts on local subalgebras to isomorphisms

$$\ql{A}_{K}\overline{\otimes}\ql{A}^{op}_{K}\cong(\ql{A}_{\beta(K)}\vee \ql{A}_{\overline{\beta(K)}})^{\prime \prime}.$$

\noindent Now, for every negatively unbounded interval $(-\infty,a)\subseteq \mathbb{R}$, consider the corresponding bounded interval in $S^{1}$ parameterized in radians by

$$I_{a,-}:=\beta(-\infty,a)\cup \overline{\beta(-\infty,a)}\subseteq S^{1}$$

$$\begin{tikzpicture}[scale=1.2,>=stealth]

  \def\a{0.9}
  \pgfmathsetmacro{\theta}{atan(1/\a)} 

  \begin{scope}[xshift=0cm]
    \draw[->] (-3.2,0) -- (2.2,0) node[right] {\(\mathbb{R}\)};

    \fill (\a,0) circle (1.2pt);
    \node[below] at (\a,0) {\(a\)};

    \draw[blue,very thick] (-3.0,0) -- (\a,0);
    \node[blue,above] at (-1.1,0) {\((-\infty,a)\)};
  \end{scope}

  \draw[<->,thick] (2.8,0) -- (4.2,0);

  \begin{scope}[xshift=7.0cm]
    \draw[->] (-1.45,0) -- (1.45,0) node[right] {\(\Re\)};
    \draw[->] (0,-1.25) -- (0,1.25) node[above] {\(\Im\)};

    \draw[gray!60,thick] (0,0) circle (1);

    \coordinate (P) at ({cos(\theta)},{sin(\theta)});
    \coordinate (Q) at ({cos(\theta)},{-sin(\theta)});

    \fill (P) circle (1.2pt);
    \fill (Q) circle (1.2pt);

    \draw[blue,very thick]
      (P) arc[start angle=\theta,end angle=360-\theta,radius=1];

    \node[blue] at (-1.3,0.4) {\(I_{a,-}\)};

    \node[above right] at (P) {$\beta(a)$};
    \node[below right] at (Q) {\(\overline{\beta(a)}\)};
  \end{scope}

\end{tikzpicture}
$$

Define $I_{a,+}$ corresponding to $(a,\infty)$ similarly, i.e.

\begin{tikzpicture}[scale=1.2,>=stealth]

  \def\a{0.9}
  \pgfmathsetmacro{\theta}{atan(1/\a)} 

  \begin{scope}[xshift=0cm]
    \draw[->] (-3.2,0) -- (2.8,0) node[right] {\(\mathbb{R}\)};

    \fill (\a,0) circle (1.2pt);
    \node[below] at (\a,0) {\(a\)};

    \draw[red,very thick] (\a,0) -- (2.6,0);
    \node[red,above] at (1.8,0) {\((a,\infty)\)};
  \end{scope}

  \draw[<->,thick] (3.8,0) -- (4.8,0);

  \begin{scope}[xshift=8.0cm]
    \draw[->] (-1.45,0) -- (1.45,0) node[right] {\(\Re\)};
    \draw[->] (0,-1.25) -- (0,1.25) node[above] {\(\Im\)};

    \draw[gray!60,thick] (0,0) circle (1);

    \coordinate (P) at ({cos(\theta)},{sin(\theta)});
    \coordinate (Q) at ({cos(\theta)},{-sin(\theta)});

    \fill (P) circle (1.2pt);
    \fill (Q) circle (1.2pt);

    \draw[red,very thick]
      (Q) arc[start angle={-\theta},end angle=\theta,radius=1];
    \node[red] at (0.6,0.3) {\(I_{a,+}\)};

    \node[above right] at (P) {\(\beta(a)\)};
    \node[below right] at (Q) {\(\overline{\beta^{-1}(a)}\)};
  \end{scope}

\end{tikzpicture}

Now, we have the following facts that follow from our definition of conformal net and the results of \cite{KawahigashiLongoMuger2001MultiIntervalSubfactors}:

\begin{enumerate}
\item 
$\pi_{0}\left(\ql{A}_{0}\otimes \ql{A}^{op}_{0}\right)_{(-\infty,a)}^{\prime \prime} =\ql{A}_{I_{a,-}}$, which follows from strong additivity.

\item 
$\pi_{0}\left(\ql{A}_{0}\otimes \ql{A}^{op}_{0}\right)_{(b,\infty)}^{\prime \prime} =\ql{A}_{I_{b,+}}$, which follows from strong additivity.
\item 
\cite{KawahigashiLongoMuger2001MultiIntervalSubfactors} There is a canonical extension of $\pi_{0}$ to $\ql{B}=\ql{A}_{0}\otimes \ql{A}^{op}_{0}\rtimes L$ such that

$\pi_{0}(\ql{B}_{(a,b)})=(\ql{A}_{I_{a,-}}\vee  \ql{A}_{I_{b,+}})^{\prime}$
\end{enumerate}

\noindent By Haag duality for $\ql{A}$, it follows immediately that

\begin{enumerate}[resume]
\item 
$\left(\pi_{0}(\ql{B}_{(-\infty, a)})\vee \pi_{0}(\ql{B}_{(b,\infty)})\right)^{\prime}=\pi_{0}(\ql{B}_{(a,b)})$
\item
$\pi(\ql{B}_{(-\infty,a)})^{\prime \prime}=\ql{A}_{I_{a,-}}$
\end{enumerate}

\[
\ql{B}_I \cong \bigoplus_{\rho\in \text{Irr}(\mathcal{C})} (\ql{A}_{I})^{\rho}\ \overline{\otimes}\ (\ql{A}^{op}_{I})^{\rho^{op}}:= \bigoplus_{\rho\in \mathrm{Irr}(\mathcal C)}
\qquad
\begin{tikzpicture}[scale=1.6,>=stealth,baseline={(current bounding box.center)}]
  \draw[->] (-1.35,0) -- (1.35,0) node[right] {\(\Re\)};
  \draw[->] (0,-1.25) -- (0,1.25) node[above] {\(\Im\)};
  \draw[gray!60,thick] (0,0) circle (1);
  \draw[red,very thick]
    ({cos(35)},{sin(35)}) arc[start angle=35,end angle=145,radius=1];
  \node[red] at (0,0.98) {\(I\)};
  \draw[red,very thick]
    ({cos(35)},{-sin(35)}) arc[start angle=-35,end angle=-145,radius=1];
  \node[red] at (0,-0.98) {\(\overline{I}\)};
  \draw[blue,very thick] (0,0.92) -- (0,-0.92);
  \node[blue,right] at (0,0) {\(\rho\)};
\end{tikzpicture}
\]

Now, for each integer $n\in \mathbb{Z}$, we have $I_{\beta(n),-}\subseteq S^{1}$, and the associated inclusions $\ql{A}_{I_{\beta(n)}}\subseteq \ql{A}_{I_{\beta(n+1)}}$ are split by hypothesis, hence there are intermediate type $\rm{I}$ factors

$$\pi(\ql{B}_{(-\infty,n)})^{\prime \prime}=\ql{A}_{I_{\beta(n)}}\subseteq B_{n}\subseteq \ql{A}_{I_{\beta(n+1)}}=\pi(\ql{B}_{(-\infty,n+1)})^{\prime \prime}.$$

\medskip

\noindent Note that the $\ql{B}_{I}$ contain the type $\rm{III}$ factors $\ql{A}_{I}\otimes \ql{A}^{op}_{I}$ hence are infinite dimensional, thus $\pi_{0}:\ql{B}\rightarrow B(\mathcal{H}_{0})$ satisfies the hypotheses of Lemma \ref{lma:trivialtest}, hence $\ql{B}^{\text{dis}}\cong \Mat(\mathbb{Z})$.

\end{proof}

Now recall a rational conformal net $\ql{A}$ is \textit{holomorphic} if the superselection category is trivial, i.e. $\DHRcat(\ql{A}_{0})=\Hilbcat$. The $E_{8}$ chiral lattice bosons are the canonical example.

\begin{corollary}
If $\ql{A}$ is any holomorphic conformal net, the discretization $\ql{A}^{\text{dis.}}_{0}$ is invertible as a net of algebras over $\mathbb{Z}$. 
\end{corollary}

\begin{proof}
This follows immediately from the previous result, together with the fact that the Longo-Rehren Q-system is trivial for holomorphic conformal nets. 
\end{proof}

This corollary gives another proof of the fact that the $E_{8}$ net $\ql{A}_{E_8}$ considered in the previous section is invertible.

\section{An invariant of invertible quasi-local algebras over $\ZZ$}    \label{sec:cmInvariant}

In this section, we prove that invertible quasi-local algebra obtained by a discretization of a holomorphic conformal net can be Brauer non-trivial.

We explicitly define the anomaly index of the symmetric group action on $N$ copies of bosonic invertible quasi-local algebras. The definition coincides with the usual anomaly index of conformal net when the quasi-local algebra is obtained by a discretization of a holomorphic conformal net. We use the anomaly index to show non-triviality of bosonic and fermionic quasi-local algebras constructed in the previous sections.

\subsection{Anomaly index}

Let $\ql{B}$ be a stable bosonic invertible quasi-local algebra over $\ZZ$ equipped with an action of a finite group $G$ via bounded spread automorphisms $\rho^{(g)}$, $g \in G$. 
\begin{lemma}
For any $g \in G$ and a large enough interval $I \subset \ZZ$, there exists a $*$-automorphism $\sigma^{(g)}$ of $\ql{B}$ such that
$$
\sigma^{(g)}(x) = x, \text{ for }x\in \ql{A}(J),\, J < I,
$$
$$
\sigma^{(g)}(x) = \rho^{(g)}(x), \text{ for }x\in \ql{A}(J),\, I < J.
$$
Moreover, for any two such automorphisms $\sigma^{(g)}, \sigma'^{(g)}$, we have $\sigma^{(g)} = \sigma'^{(g)} \Ad_{u^{(g)}}$, $u \in \ql{B}_I$. 
\end{lemma}
\begin{proof}
There exists a bounded spread isomorphism $\alpha: \ql{B} \otimes \tilde{\ql{B}} \to \Mat(\ZZ)$ of spread $l$ for some $l>0$ and a quasi-local algebra $\tilde{\ql{B}}$. Let $\ql{A} = \ql{B} \otimes \tilde{\ql{B}}$ with a $G$-action by $\rho^{(g)} \otimes \Id$. Since $\QCA(\ZZ) = 0$, $\beta:=\alpha (\rho^{(g)} \otimes \Id) \alpha^{-1}$ is a quantum circuit automorphism of $\Mat(\ZZ)$, and thus for any large enough interval $I$, has a truncation $\beta^{(g)}$, with $\beta^{(g)}|_{\ql{A}>I}=\alpha (\rho^{(g)}\otimes \Id) \alpha^{-1}|_{\ql{A}>I}$ and $\beta^{(g)}|_{\ql{A}<I}=\Id_{\ql{A}<I}$. Therefore for large enough interval $I$, there exists a bounded spread automorphism $\ups^{(g)}=\alpha^{-1}\beta^{(g)}\alpha$ such that
$$
\ups^{(g)}(x) = x, \text{ for }x\in (\ql{B} \otimes \tilde{\ql{B}})(J),\, J < I,
$$
$$
\ups^{(g)}(x) = (\rho^{(g)} \otimes \Id)(x), \text{ for }x\in (\ql{B} \otimes \tilde{\ql{B}})(J),\, I < J.
$$
Since $\ql{B}$ is stable $\ql{B} \cong \ql{B} \otimes \Mat(\ZZ)$ and since $\Mat(\ZZ)$ is bounded spread isomorphic to $\tilde{\ql{B}} \otimes \ql{B}$, existence of $\ups^{(g)}$ implies an existence of $\sigma^{(g)}$.

Let $\eta = \sigma^{-1} \sigma'$. Then $\eta \otimes \Id$ is a $*$-automorphism that acts trivially on $(\ql{B} \otimes \tilde{\ql{B}})_{I^c}$. Since $\ql{B} \otimes \tilde{\ql{B}}$ is bounded spread isomorphic to $\Mat(\ZZ)$, it follows that $\eta \otimes \Id$, and hence $\eta$, is inner.
\end{proof}
We call $*$-automorphisms from lemma above {\it twist automorphisms} associated with $I$.

Let $I \subset \ZZ$ be an interval of large enough length. Choose twist automorphisms $\sigma^{(g)}$ associated with $I$ for each $g \in G$ and define unitary elements $w^{(g,h)} \in \ql{B}_I$ such that $\sigma^{(g)} \sigma^{(h)} = \Ad_{w^{(g,h)}} \sigma^{(gh)}$. The anomaly index is defined by the class $[\omega] \in H^3(G,\RR/\ZZ)$ of the 3-cocycle $\omega^{(g,h,k)}$ defined by $w^{(g,h)} w^{(gh,k)} = \omega^{(g,h,k)} \sigma^{(g)}(w^{(h,k)}) w^{(g,hk)}$. 

\begin{prop}
The anomaly index $[\omega] \in H^3(G,\RR/\ZZ)$ does not depend on the choice of an interval $I$, twist automorphisms $\sigma^{(g)}$, and unitaries $w^{(g,h)}$. Moreover, if quasi-local algebras with $G$-action are related by a $G$-equivariant bounded spread isomorphism, then their anomaly indices are the same.
\end{prop}
\begin{proof}
Different choices of $w^{(g,h)}$ are related by a phase factor and modify $\omega(g,h,k)$ by a coboundary term. Different choices of $\sigma^{(g)}$ are related by $\Ad_{u^{(g)}}$ for unitary elements $u^{(g)} \in \ql{A}_I$. By explicit computation, the corresponding $\omega(g,h,k)$ are also related by a coboundary term. An ambiguity in the choice of $I$ can be absorbed into an ambiguity in the choice if $\sigma^{(g)}$.

Suppose $\varphi:\ql{A} \to \ql{B}$ is a $G$-equivariant bounded spread isomorphism of quasi-local algebras equipped with $G$-action. Given a choice of twist automorphisms $\sigma^{(g)}$ for $\ql{A}$, the automorphisms $\varphi \sigma^{(g)} \varphi^{-1}$ provide a choice of twist automorphisms for $\ql{B}$ that gives rise to the same $\omega^{(g,h,k)}$.
\end{proof}

\subsection{Non-triviality of $\AEed$ and $\Amajd$}

When $\ql{B}$ is a holomorphic conformal net with a $G$-action, the anomaly index of the quasi-local algebra obtained by discretization coincides with the $H^3(G,\RR/\ZZ)$ class defined in \cite{M_ger_2005} by construction. 

Consider $N$ copies of a given holomorphic conformal net with the natural $S_N$ symmetric group action. It was noted in \cite{johnson2019moonshine,Bischoff2020AnomaliesCyclicPermutationOrbifolds,evans2022reconstruction} that the associated class $[\omega_N] \in H^3(S_N,\RR/\ZZ)$ is non-trivial when $c_- \bmod 3 \neq 0$. That allows us to show that quasi-local algebras obtained by discretizations of such conformal nets are Brauer non-trivial. More precisely, \cite[Theorem 2]{evans2022reconstruction} implies

\begin{theorem} \label{thm:NontrivialityOfHolomorphicCFTs}
Let $\ql{A}$ be a quasi-local algebra of the discretization of a (bosonic) holomorphic conformal net with central charge $c_- \bmod 3 \neq 0$. Then the corresponding anomaly index $[\omega_N] \in H^3(S_N,\RR/\ZZ)$ for the $S_N$ symmetric group action on $\ql{A}^{\otimes N}$ is non-trivial. In particular, $\ql{A}$ is Brauer non-trivial.
\end{theorem}
\begin{remark}
It can be shown that any holomorphic conformal net with non-zero $c_-$ defines a non-trivial element in the Brauer group of quasi-local algebras using a method similar to the one developed in \cite{sopenko2025reflection}. The details will be discussed elsewhere. 
\end{remark}

As a corollary, we obtain

\begin{theorem} \label{thm:NontrivialityOfAmaj}
The quasi-local algebras $\AEed$ and $\Amajd$ are Brauer non-trivial.    
\end{theorem}
\begin{proof}
Non-triviality of the Brauer class of $\AEed$ follows from Theorem \ref{thm:NontrivialityOfHolomorphicCFTs}. 

Suppose $\Amajd$ is Brauer trivial. By Proposition \ref{prop:E8is16copiesOfMaj}, the quasi-local algebras $\AEed \otimes \sMat(\ZZ)$ and $(\Amajd)^{\otimes 16}$ are bounded spread isomorphic. Hence, $\AEed \otimes \sMat(\ZZ)$ is Brauer trivial in $\sBrgr(\ZZ)$. Therefore, by a relation between $\sMat(\ZZ)$ and $\Mat(\ZZ)$ as in the proof of Proposition \ref{prop:AE8isinvertible}, $\AEed \otimes \Mat(\ZZ)$ is Brauer trivial in $\Brgr(\ZZ)$, that contradicts non-triviality of $\AEed$. Thus, $\Amajd$ is Brauer non-trivial as well.
\end{proof}

\printbibliography

\end{document}